\documentclass[11pt,reqno]{amsart}
\usepackage{etex}
\usepackage{amsmath}
\usepackage{amssymb}
\usepackage{amsfonts}
\usepackage{amsthm}
\usepackage{a4wide}
\usepackage{bbm}
\usepackage{bm}
\usepackage{graphics}
\usepackage{color}
\usepackage[sans]{dsfont}
\usepackage[toc,page]{appendix}

\usepackage{hyperref}
\usepackage{protosem}
\usepackage{linearA}
\usepackage{pgfkeys}
\usepackage{longtable}
\usepackage{pdflscape}
\usepackage{float}
\usepackage{cancel}
\usepackage{dsfont}
\usepackage{setspace}
\usepackage{array}
\usepackage{appendix}
\usepackage{amssymb}
\usepackage{epstopdf}
 \usepackage{multirow}
\usepackage{epsfig}
\usepackage{lscape}
\usepackage{yhmath}
\usepackage[english]{babel}
\usepackage[latin1]{inputenc}
\usepackage[T1]{fontenc}
\usepackage{wasysym}
\usepackage{tikz,pgflibraryshapes}
\usepackage{pstricks,pst-plot,pstricks-add}
\usepackage{enumerate}
\usetikzlibrary{arrows}
\usetikzlibrary{snakes}
\usepackage{textcomp}
\usepackage{mathrsfs}

\newtheorem{theorem}{Theorem}[section]
\newtheorem{lemma}[theorem]{Lemma}

\newtheorem{definition}[theorem]{Definition}

\newtheorem*{remark}{Remarks}

\numberwithin{equation}{section}

\def\XXint#1#2#3{{\setbox0=\hbox{$#1{#2#3}{\int}$}
     \vcenter{\hbox{$#2#3$}}\kern-.5\wd0}}

\newcommand{\captionfonts}{\footnotesize}

\makeatletter  
\long\def\@makecaption#1#2{%
  \vskip\abovecaptionskip
  \sbox\@tempboxa{{\captionfonts #1: #2}}%
  \ifdim \wd\@tempboxa >\hsize
    {\captionfonts #1: #2\par}
  \else
    \hbox to\hsize{\hfil\box\@tempboxa\hfil}%
  \fi
  \vskip\belowcaptionskip}
\makeatother   

\title[Analyticity of the self-energy]{Analyticity of the  self-energy in total momentum of an atom coupled to the quantized  radiation field}
 
\begin{document}

\author[J. Faupin]{J{\'e}r{\'e}my Faupin}
\address[J. Faupin]{Institut Elie Cartan de Lorraine, Universit\'e de Lorraine,\\
57045 Metz Cedex 1, France}
\email{jeremy.faupin@univ-lorraine.fr}
\author[J. Fr\"ohlich]{J\"urg Fr\"ohlich}
\address[J. Fr{\"o}hlich]{Institut f{\"u}r Theoretische Physik, ETH H{\"o}nggerberg, CH-8093 Z{\"u}rich, Switzerland}
\email{juerg@phys.ethz.ch}
\author[B. Schubnel]{Baptiste Schubnel}
\address[B. Schubnel]{Departement Mathematik, ETH Z{\"u}rich, CH-8092 Z{\"u}rich, Switzerland}
\email{baptiste.schubnel@math.ethz.ch}

\date \today

\maketitle

\begin{abstract}
 We study a neutral atom with a non-vanishing electric dipole moment coupled to the quantized electromagnetic field. For a sufficiently small dipole moment and small momentum, $\vec{p}$, the one-particle (self-) energy of an atom is proven to be a real-analytic function of its momentum. The main ingredient of our proof is a suitable form of the Feshbach-Schur spectral renormalization group.
\end{abstract}

\section{Introduction}
\label{intro}

\subsection{Description of the model and statement of the main result}

In this paper, we consider models of a neutral atom with a non-zero electric dipole moment coupled to the quantized electromagnetic field.  We are interested in analyzing properties of the self-energy -- or dispersion law -- as a function of the  momentum $\vec{p}$ of the atom. When the coupling of the atom to the electromagnetic field is turned off, its dispersion law is taken to be $\vec{p}^2/2m$, where $m$ is its  bare mass. Our purpose is to study the radiative corrections to this law (self-energy) when the atom is coupled to the electromagnetic field. We will prove real analyticity of the self-energy in the momentum $\vec{p}$, for $\vert \vec{p} \vert<m$, provided the dipole moment of the atom is sufficiently small; (the speed of light and Planck's constant are set to $1$, throughout this paper). Our result has interesting applications to the study of Compton scattering of atoms and of the effective dynamics of an atom when it moves through an external potential. These matters, as well as the study of atomic resonances for atoms of finite total mass,  are left  for the subject of a forthcoming work.

\subsubsection{The Hamiltonian }
We study the simplest Hamiltonian  describing a freely moving neutral atom coupled to the electromagnetic field via an electric dipole moment. We consider the atom as a two-level system coupled to the electromagnetic field via  an interaction Hamiltonian 
\begin{equation}
H_{I}:=- \vec{d} \cdot \vec{E},  
\end{equation}
where $\vec{d}$ is the dipole moment of the atom and $\vec{E}$ the quantized electric field. For a two-level system, $\vec{d}$ can be expanded in the basis $(\sigma_x, \sigma_y,\sigma_z)$ of Pauli matrices. The quantization of the electromagnetic field in the Coulomb gauge is accomplished  within the usual   second quantization formalism. Readers  not familiar with it  are encouraged to consult \cite{Bratteli} and the references given therein.

The Hilbert space of the system is the tensor product
\begin{equation*}
\mathcal{H}=\mathcal{H}_{at} \otimes \mathcal{H}_f,
\end{equation*}
where  
\begin{equation*}
\mathcal{H}_{at}:=L^{2}(\mathbb{R}^{3}) \otimes \mathbb{C}^2  \qquad   \text{and}  \qquad \mathcal{H}_f:= \mathcal{F}_{+}(L^{2}( \underline{\mathbb{R}}^3 ))
\end{equation*}
are   the atomic and the field Hilbert spaces, with $ \mathcal{F}_{+}(L^{2}( \underline{\mathbb{R}}^3 ))$ the symmetric Fock space over  $L^{2}( \underline{\mathbb{R}}^3 )$. Here and in what follows, we use  the shorthand
\begin{equation}\label{notation_R_souligne}
\underline{\mathbb{R}}^3:= \mathbb{R}^3\times\{1,2\} = \left\{\underline{k} := (\vec{k},\lambda)\in\mathbb{R}^3\times\{1,2\} \right\},
\end{equation}
where $\lambda$ is the polarization index of the field. To shorten  notations, we also set $\underline{\mathbb{R}}^{3n}:=(\underline{\mathbb{R}}^3)^n$, and, for $A\subset\mathbb{R}^3$,   
\begin{equation}
\label{short2}
\underline{A}:=A\times\{1,2\}, \qquad \int_{\underline{A}} d \underline{k}:=\sum_{\lambda=1,2}\int_{A} d \vec{k}.
\end{equation}
Notations (\ref{notation_R_souligne}) and (\ref{short2}) are used  throughout  this paper.

 The  dynamics of the system  is given by the Hamiltonian
 \begin{equation}
 \label{H}
 H=H_{at}+H_{f} + \lambda_0 H_{I},
\end{equation}
where
\begin{equation}
\label{H1}
H_{at}:=-\frac{\Delta}{2m} \otimes \mathds{1}  \otimes \mathds{1} +    \mathds{1}  \otimes  \left( \begin{array}{cc} \omega_0 &0\\0&0 \end{array}   \right)\otimes \mathds{1}
\end{equation}
is the free atomic Hamiltonian,
\begin{equation}
\label{H2}
H_{f}= \mathds{1} \otimes \mathds{1} \otimes \int_{ \underline{\mathbb{R}}^3 } \vert \vec{k} \vert \text{ } a^{*}( \underline{k} ) a ( \underline{k} ) d \underline{k}
\end{equation}
is the Hamiltonian of the free electromagnetic field, and
\begin{equation}
\label{H3}
H_{I}=i  \int_{\underline{B}_{1} } \vert \vec{k} \vert^{\frac12} \left(  e^{i \vec{ k} \cdot \vec{x}}  \otimes   \vec{\epsilon} (\underline{k} ) \cdot   \vec{\sigma} \otimes  a ( \underline{k} ) - h.c. \right) d \underline{k}
\end{equation}
is the interaction Hamiltonian;  $\lambda_0 \ge 0$, in \eqref{H}, is the coupling constant,  $\Delta$, in \eqref{H1}, is the Laplace operator,  $m>0$ and $\omega_0>0$, in  \eqref{H1}, are the mass of the atom and the  energy of the  internal excited state of  the atom,   respectively, and
\begin{equation}   
a (\underline{k}):=a _{\lambda}(\vec{k}), \qquad a ^*(\underline{k}):=a^{*}_{\lambda}(\vec{k}),  \qquad \vec{\epsilon}(\underline{k}):=\vec{\epsilon}_{\lambda}(\vec{k}),
\end{equation}
are the annihilation and creation operators on $\mathcal{H}_f$  and the polarization vectors of the electromagnetic field in the Coulomb gauge. Furthermore, we denoted by $B_{1}$ in \eqref{H3} the closed ball in $\mathbb{R}^3$ centered at $0$ with radius $1$, by  $\vec{\sigma}$  the vector  of Pauli matrices and by $\vec{x}$ the position operator of the atom.

\begin{remark}
$\quad$ 
\begin{itemize}
\item  The Hamiltonian $H$ in \eqref{H} can be derived from  the  Hamiltonian of a localized neutral system of charges interacting with the quantized electromagnetic field,  under the assumption that the typical size   of the system of charges is very small in comparison to the typical wavelength   of the  radiation field. This approximation is called the ``dipole'' approximation. We refer the reader to \cite{photons_atoms} for  details concerning the unitary transformation (the G\"oppert-Mayer transformation) that brings the Hamiltonian of the system of charges in interaction with the electromagnetic field  to an electric dipole Hamiltonian similar to the one in \eqref{H}.

\item Standard estimates show that $H_I$ is $H_{f}^{1/2}$ relatively bounded. Therefore, by Kato's Theorem, $H$ is self-adjoint  on the domain of $H_{at}+H_f$. We also notice that the interaction term $H_I$ in \eqref{H3} behaves well in the infrared and that  our model has no infrared  catastrophe. This feature simplifies  our analysis, as compared to the analysis of charged particles.

\item To ease notations, the ultraviolet cut-off imposed in the interaction Hamiltonian $H_I$ has been chosen to be equal to $1$.  We point out that  our  result and our proofs  hold for an arbitrary ultraviolet cut-off.

\item The   Hamiltonian $H$  has the very  important property to be translation invariant. Indeed,  an easy calculation shows  that each component of  the total momentum operator
\begin{equation}
\vec{P}_{tot}:=-i \vec{\nabla}_{x} \otimes \mathds{1} \otimes \mathds{1} + \mathds{1} \otimes \mathds{1} \otimes \int_{ \underline{ \mathbb{R} }^3 } \vec{k} \text{ }a^{*}  (\underline{k}) a  (\underline{k}) d \underline{k} ,
\end{equation}
commutes with  $H$.
\end{itemize}
\end{remark}

The  property of translation invariance implies the existence of a unitary  map  $U:\mathcal{H} \rightarrow \tilde{\mathcal{H}}$, where 
\begin{equation}
\label{dec}
\tilde{\mathcal{H}} := \int_{\mathbb{R}^{3}}^{\oplus} \mathcal{H}_{\vec{p}} \text{ } d \vec{p}. 
\end{equation}
More precisely, $U$ is the generalized Fourier transform, defined, for any $\psi \in \mathcal{H}$ decaying sufficiently rapidly at infinity, by,  
\begin{equation}
\label{UU}
(U \psi) ( \vec{p}) = \frac{1}{(2 \pi)^{3/2}}\int_{\mathbb{R}^{3}}  e^{-i(\vec{p} -\vec{P}_f) \cdot \vec{y}} \psi(\vec{y}) d \vec{y},
\end{equation} 
where $$\vec{P}_f := \int_{ \underline{ \mathbb{R} }^3 } \vec{k} \text{ }  
a^*( \underline{k} ) a( \underline{k} ) d \underline{k}.$$ 
The Hilbert spaces $ \mathcal{H}_{\vec{p}}$ are isomorphic to $ \mathbb{C}^{2 } \otimes \mathcal{F}_{+}( L^{2} ( \underline{ \mathbb{R} }^3 ) )$. We rewrite the Hamiltonian $H$ in the representation (\ref{dec}). We introduce the notations 
\begin{equation}
b  (\underline{k}):= U e^{   i \vec{k} \cdot \vec{x}} a (\underline{k}) U^{-1} , \qquad \qquad b^{*} (\underline{k}):= U e^{ - i \vec{k} \cdot \vec{x}} a^{*} (\underline{k}) U^{-1}.
\end{equation}
 The operator-valued distributions $b (\underline{k})$, $b^{*} (\underline{k}')$  satisfy the canonical commutation relations $[b (\underline{k}),b^{*} (\underline{k}')]= \delta_{\lambda,\lambda'} \delta(\vec{k}-\vec{k}')$ and  $[b^{\sharp} (\underline{k}),b^{\sharp} (\underline{k}')] =0$, where $b^{\sharp} $ stands for $b$ or $b^{*}$. Furthermore,
\begin{equation*}
\left(\left(U  \frac{-\Delta}{2m} U^{-1} \right) \psi \right)(\vec{p})=\frac{\left(\vec{p}-\vec{P}_f \right)^2}{2m} \psi(\vec{p}).
\end{equation*}
It follows that the Hamiltonian $H$ can be decomposed as a direct integral of  fiber Hamiltonians,  $H(\vec{p})$, with 
\begin{align}
H(\vec{p}) =& \frac{\left(\vec{p}-\vec{P}_f \right)^2}{2m} +  \omega_0 \left(\begin{array}{cc} 1 & 0 \\ 0& 0 \end{array} \right) \notag \\
& + i \lambda_0 \int_{ \underline{B}_{1} } \vert \vec{k} \vert^{\frac12} \vec{\epsilon} (\underline{k}) \cdot \vec{\sigma}     \left( b (\underline{k})  - b^{*} (\underline{k}) \right) d \underline{k} +  \int_{ \underline{ \mathbb{R} }^3 }  \vert \vec{k}  \vert b ^{*}(\underline{k}) b (\underline{k}) d \underline{k}. \label{Hp}
\end{align}
Here and in what follows, we suppressed the tensor product notation,  $\mathds{1} \otimes \dots$. Using standard estimates and the Kato-Rellich theorem, one can prove that the fiber Hamiltonians $H(\vec{p})$  are self-adjoint on the dense domain $D:=D(H_f+\vec{P}_f^2)$ of $\mathcal{H}_{ \vec{p} }$,  and  bounded from below.

The  question to know whether
\begin{equation}
E(\vec{p}):= \inf\left( \sigma(H(\vec{p}) )\right)
\end{equation}
is an eigenvalue of $H(\vec{p})$ and whether it has some regularity property  with respect to the total momentum $\vec{p}$ has been already addressed in the literature for different models \cite{AbHa12_01,AmGrGu06_01,BaChFrSi07_01,Ch08_01,ChFr07_01,ChFrPi09_01,DyPi13_01,FrPi10_01,HaHe08_01,LoMiSp07_01,Pi03_01,ChFrPi10_01,FrGrSc04_01,Mo05_01}. In particular, for the Nelson model, it has been shown in \cite{Froehlich}   that $E(\vec{p})$ is a non-degenerate eigenvalue if and only if   an infrared regularization is imposed. In \cite{DyPi13_01,Pi03_01}, using iterative perturbation theory and a multiscale analysis, it has been proven that $\vec{p} \mapsto E( \vec{p} )$ is twice differentiable near $\vec{0}$. Without imposing any infrared regularization to the  form factor in the Nelson model, \cite{AbHa12_01} has established using in particular a cluster expansion that $E( \vec{p} )$ is a real analytic function of $\vec{p}$ and  the coupling constant. For charged particles in non-relativistic QED, similar results have been obtained in \cite{BaChFrSi07_01,Ch08_01,ChFr07_01,ChFrPi09_01,FrPi10_01,HaHe08_01}: It is known that $E( \vec{p} )$ is an eigenvalue of $H( \vec{p} )$ if and only if $\vec{p} = \vec{0}$ \cite{Ch08_01,ChFr07_01,HaHe08_01} (unless and infrared regularization is imposed) and that $\vec{p} \mapsto E( \vec{p} )$ is twice differentiable near $\vec{0}$. The latter property has been  proven using an application of the spectral renormalization group (see \cite{Ch08_01}) and by iterative perturbation theory \cite{ChFrPi09_01,FrPi10_01}. For results concerning models describing moving atoms or ions coupled to the electromagnetic fields, we refer to \cite{AmGrGu06_01,FrGrSc07_01,HaHe08_01,LoMiSp07_01,faupin2008}

For the Hamiltonian model (\ref{H}) studied in this paper, we can prove:

\begin{theorem}
\label{main1} Let $0 < \nu < m$. There exists  a constant $\lambda_c ( \nu )>0$ such that, for any coupling constant $\lambda_0 \ge 0$ satisfying $ \lambda_0 < \lambda_c( \nu )$, the map $\vec{p} \mapsto E(\vec{p})$ and its associated eigenprojection $\vec{p} \mapsto \Pi( \vec{p})$ are  real analytic on $B_\nu = \{ \vec{p} \in \mathbb{R}^3 , | \vec{p} | < \nu \}$. 
\end{theorem}

\begin{remark}
$\quad$
\begin{itemize}
\item The restriction to momenta $\vert \vec{p} \vert <m$ is crucial for our result to hold. Indeed, for  $\vert \vec{p} \vert > m$, we expect $E(\vec{p})$ to dissolve in  the continuum. This feature is the mathematical relic of the Cherenkov radiation emitted by the particle when it travels faster than the velocity of the light in the medium;  see  e.g. \cite{Ceren}.

\item  The absence of infrared divergences in our model implies that the one-particle states (or fiber ground states) $\psi(\vec{p})$ stay in $\mathcal{H}_{\vec{p}}$ for all $\vert \vec{p} \vert<m $. This result would not hold for a particle with a  non-zero net charge,  because the infrared catastrophe would force the eigenstate to leave the Fock space, for all $\vec{p} \neq \vec{0}$. We refer the reader to \cite{Froehlich} for an extensive discussion of this problem.

\item  The critical constant $\lambda_c(\nu)$ decays like a power law in $m-\nu$ when $\nu$approaches  $m$.  
\item We observe that, by rotation invariance, $E( \vec{p} )$ only depends on the norm of $\vec{p}$. Theorem \ref{main1} shows that the map $| \vec{p} | \mapsto E( | \vec{p} | )$ is real analytic on the interval $[ 0 , \mu )$. The regularity of $E( | \vec{p} | )$ with respect to $|\vec{p}|$ is an important physical property, that allows one, in particular, to define the renormalized mass of the dipole by the formula $m_{\mathrm{eff}} = ( \partial^2_{ | \vec{p} | } E( | \vec{p} | ) )^{-1}$. In a companion paper \cite{FaFrSc13_01}, we study the effective dynamics of the dipole placed in a slowly varying external potential; We justify that the renormalized mass and the kinetic mass of the dipole coincide; see \ref{faupin} for a similar result in the case of an electron placed in a slowly varying potential. The results of Theorem \ref{main1} are also expected to be useful in the framework of scattering theory, as we plan to consider in  future work; (see also \cite{FrGrSc07_01,Pi05_01}). 
\end{itemize}
\end{remark}

\subsection{Outline of the strategy}
\label{out}
Our analysis is based on the operator theoretic renormalization group method introduced in   \cite{BaFrSi98_01}, \cite{BaFrSi98_02} and \cite{BaChFrSi03_01}. For a concise review, the reader  is referred to \cite{FrGrIr08}. This method, based on an iteration of the smooth Feshbach-Schur map, has been useful in the study of  the existence of  ground states and resonances in models of matter coupled to a  quantized field;  Some results on the analyticity of ground states and  their associated eigenvalues have been obtained recently with this method, but the models that have been investigated  in the literature so far,  (as far as the question of analyticity with respect to a parameter is concerned) all deal with fixed or confined atoms; see \cite{GrHa09_01,HaHe10_01,HaHe11_01}. For instance, it has been shown in  \cite{HaHe11_01} that the ground state of the spin boson model is analytic in the coupling constant $\lambda_0$.  

In this paper, we investigate the analyticity of the dispersion law with respect to the total momentum $\vec{p}$. The novelty of our result lies in the fact that $\vec{p}$ appears in the marginal term $\vec{p} \cdot \vec{P}_f$ in the initial Hamiltonian $H(\vec{p})$, and not in the perturbation. The perturbation contracts under the iteration of the renormalization map, whereas the size of the marginal term  does not change much. The control of the evolution of  the marginal terms under the renormalization flow is therefore an important issue in our analysis. A similar issue appears in \cite{BaChFrSi07_01,Ch08_01}, but our approach to deal with it differs  from that in  \cite{BaChFrSi07_01,Ch08_01} and is slightly simpler in some respects. In particular, we consider a different Banach space of effective Hamiltonians, and we apply the Feshbach-Schur map at each step of the renormalization procedure in a different way. Besides, \cite{Ch08_01} shows that $\vec{p} \mapsto E( \vec{p} )$ is twice differentiable near $\vec{0}$, while our main result establishes the real analyticity of this map.

Next, we describe the different steps of our analysis and give an overview of  the way  the spectral renormalization group works in the present setup.

\subsubsection{ Complexification of the total momentum}
We start  by fixing a vector $\vec{p}^*$ in $\mathbb{R}^3$  of length smaller than $m$. We set $2 \mu =(m-\vert \vec{p}^* \vert )/m$, and consider the open set
\begin{equation*}
U[\vec{p}^*]:=\lbrace \vec{p} \in \mathbb{C}^3  \mid  \vert \vec{p} - \vec{p}^* \vert < \mu m \rbrace,
\end{equation*}
centered at $\vec{p}^*$.  For any $\vec{p} \in U[\vec{p}^*]$, we subtract the constant term $\vec{p}^2/2m$ in \eqref{Hp}, and consider the  operator
\begin{equation}
\label{Hp2}
H(\vec{p}) := \frac{\vec{P}_f^2}{2m} - \frac{\vec{p}}{m} \cdot \vec{P_f}+  \omega_0 \left(\begin{array}{cc} 1 & 0 \\ 0& 0 \end{array} \right) \notag +\lambda_0  H_I + H_{f}, 
\end{equation}
where  
\begin{equation}
\label{Hredef}
 H_I := i \int_{\underline{B}_1} \vert \vec{k} \vert^{\frac12} \vec{\epsilon} ( \underline{k} ) \cdot \vec{\sigma}    \left( b ( \underline{k} )  - b^{*}(\underline{k}) \right) d \underline{k},  \quad H_{f}:=  \int_{ \underline{ \mathbb{R} }^3 } \vert \vec{k} \vert  b^{*}( \underline{k} ) b ( \underline{k} ) d \underline{k}.
\end{equation}
The subtraction of $\vec{p}^2/2m$  leads to a shift in the spectrum of the original Hamiltonian  by  $\vec{p}^2/2m$. For $\vec{p} \in U[\vec{p}^*] \cap (\mathbb{C}^3 \setminus \mathbb{R}^3)$, the operator $H(\vec{p})$ is not self-adjoint, anymore,  but it is   closed on $D$.   We  show that the spectral renormalization group method can be used to investigate the spectrum of  the operator $H(\vec{p})$  near the origin of the complex plane. We  will find an eigenvalue  of $H(\vec{p})$ in this neighborhood of the origin, for any $\vec{p} \in U[\vec{p}^*]$. This eigenvalue  turns out to  be equal to $\text{inf}\left(\sigma(H(\vec{p})) \right)$,  for $\vec{p} \in U[\vec{p}^*] \cap \mathbb{R}^3$; see Subsection \ref{itere}.

\subsubsection{The first decimation step and the Feshbach-Schur map} 
  The purpose of the Feshbach-Schur map introduced in  \cite{BaFrSi98_02,BaFrSi98_01}, is to construct a new operator -- called ``effective Hamiltonian'' -- that acts on a subspace of ``low-energy photons'' of  $\mathcal{H}$, with the property that the spectrum and the eigenstates of $H(\vec{p})$  near the origin can be uniquely reconstructed from the study of the kernel of  this effective Hamiltonian. In this paper, we use the smooth version of the Feshbach-Schur map that has been developed in \cite{BaChFrSi03_01}. The precise relation between the spectrum, resolvent and eigenvectors of a closed operator and its Feshbach-Schur  transform are recalled in Subsection \ref{Fesch}. We restrict the values of $z$ to the complex open disc 
  \begin{equation}
  D_{\mu/2}:= \{ z \in \mathbb{C} \mid \vert z \vert < \frac{\mu}{2}\}.
  \end{equation}
  The effective Hamiltonian constructed from $H(\vec{p})- z \mathds{1}$, using the Feshbach-Schur map,  can be cast into the form
\begin{equation}
\label{H(0)}
H^{(0)}(\vec{p},z)=  \sum_{m+n \geq 0}  W^{(0)}_{m,n}(\vec{p},z),
\end{equation}
where the operators $W^{(0)}_{m,n}(\vec{p},z)$, $m+n \geq 0$,  are called "Wick monomials". The operator $H^{(0)}(\vec{p},z)$ acts on the subspace $\mathcal{H}_{\text{red}}:=\mathds{1}_{H_f \leq 1} \mathcal{H}_f$ of  $\mathcal{H}_f$. The Wick monomials are bounded operators on $\mathcal{H}_{\text{red}}$. They are associated to a sequence $$\underline{w}^{(0)}:=(w_{m,n}^{(0)})_{m+n\geq 0}$$ of bounded measurable functions $$w^{(0)}_{m,n}:  U[\vec{p}^*] \times D_{\mu/2} \times \mathbb{R}_+ \times \mathbb{R}^3  \times \underline{\mathbb{R}}^{m} \times  \underline{\mathbb{R}}^{n}  \rightarrow \mathbb{C}$$  that are $\mathrm{C}^1$ in their third and fourth arguments, and symmetric in $ \underline{\mathbb{R}}^{m}$ and  $ \underline{\mathbb{R}}^{n}$.  The bounded operators $W^{(0)}_{m,n}(\vec{p},z)$, $m+n \geq 0$, are defined in the sense of quadratic forms by 
\begin{align}
W^{(0)}_{m,n}(\vec{p},z) := \mathds{1}_{H_f \leq 1} \int_{\underline{B}_{1}^{m+n}} &    \left( \prod_{i=1}^{m} b^{*}(\underline{k}_i) \right) w^{(0)}_{m,n} (\vec{p},z, H_f , \vec{P}_f ,\underline{k}_1,...,\underline{k}_{m},\underline{\tilde{k}}_1,...,\underline{\tilde{k}}_n ) \notag \\
&\quad  \left( \prod_{j=1}^{n} b(\underline{\tilde{k}}_j) \right) \prod_{i=1}^{m} d \underline{k}_i    \prod_{j=1}^{n} d \underline{\tilde{k}}_j   \mathds{1}_{H_f \leq 1}  , \label{cmoii}
\end{align} 
where $ \mathds{1}_{H_f \leq 1}$  projects   on the subspace of photons with energy smaller than one. The functions $w^{(0)}_{m,n}$ are called "kernels" in the literature. They satisfy additional properties that  are stated in Section \ref{renorm}.

\subsubsection{Functional calculus}
\label{funcpar}
Before continuing the outline of our proof, we explain how the operators  $w^{(0)}_{m,n} (\vec{p},z, H_f , \vec{P}_f ,\underline{k}_1,...,\underline{k}_{m},\underline{\tilde{k}}_1,...,\underline{\tilde{k}}_n )$ and $w^{(0)}_{0,0}(\vec{p},z,H_f, \vec{P_f})$   in \eqref{cmoii} are defined. Let $f: \mathbb{R} \times \mathbb{R}^3 \rightarrow \mathbb{C}$ be a measurable function.  Any element $\Psi \in \mathcal{F}_{+}(L^{2}( \underline{ \mathbb{R} }^3 ) )$ can be written as a sequence of totally symmetric functions $( \psi^{(n)})_{n \geq 0}$ of momenta, with $\psi^{(n)} \in L^2_s( \underline{ \mathbb{R} }^{3n} )$. We set
\begin{equation}
\label{defi}
f(H_f,\vec{P}_f) \Psi :=  ( f^{(n)} \psi^{(n)})_{n \geq 0},
\end{equation}
where
\begin{equation}
f^{(0)} \psi^{(0)} := f(0,\vec{0}) \psi^{(0)} ,
\end{equation}
and, for $n \geq 1$,  $( \underline{k_1}, \dots , \underline{k_n})$ in $ \underline{ \mathbb{R} }^{3n}$,
\begin{equation}
(f^{(n)} \psi^{(n)})(\underline{ k_1}, \dots ,\underline{ k_n}):=f(\vert \vec{ k}_1 \vert + \dots +\vert \vec{ k}_n \vert, \vec{k_1}+ \cdots + \vec{k_n}) \text{ } \psi^{(n)}  ( \underline{k_1}, \dots , \underline{k_n}).
\end{equation}
 \eqref{defi} defines an (unbounded) operator on $\mathcal{F}_{+}(L^{2}( \underline{ \mathbb{R} }^3 ) )$ with domain
$$ D( f ( H_f , \vec{P}_f ) ) := \{ \Psi \in \mathcal{F}_{+}(L^{2}( \underline{ \mathbb{R} }^3 ) ) , \| f(H_f,\vec{P}_f) \Psi \| <\infty \}.$$ 
Assuming that   $f$ is essentially bounded on the subset 
\begin{equation}
\label{RR}
\mathcal{R}:=\lbrace (r, \vec{l}) \in   \mathbb{R} \times \mathbb{R}^3 \mid  \vert \vec{l} \vert \leq r \rbrace \subset \mathbb{R}^4,
\end{equation}
it is easy to check that $f ( H_f , \vec{P}_f )$ is bounded with
\begin{equation}
\label{essbo}
\| f ( H_f , \vec{P}_f ) \| \le \text{ess sup} \, \{ | f( r , \vec{l} ) | , ( r , \vec{l} ) \in \mathcal{R}  \}.
\end{equation}

\subsubsection{The renormalization map and its iteration} 
Under certain conditions on the norm of the operators $W_{m,n}^{(0)}(\vec{p},z)$, it is possible to apply to $H^{(0)}(\vec{p},z)$ a renormalization transformation $\mathcal{R}_{\rho}$ that leads to a new effective Hamiltonian $H^{(1)}(\vec{p},z)$ acting on $\mathcal{H}_{\text{red}}$. $H^{(1)}(\vec{p},z)$ can also be cast into the form
 \begin{equation}
\label{H(1)}
H^{(1)}(\vec{p},z)=  \sum_{m+n \geq 0}  W^{(1)}_{m,n}(\vec{p},z),
\end{equation}
where the operators $ W^{(1)}_{m,n}(\vec{p},z)$, $m+n \geq 0$, are Wick monomials associated to a sequence $\underline{w}^{(1)}$ of kernels. The renormalization map  $\mathcal{R}_{\rho}$, constructed in Section \ref{renorm}, is ``isospectral'' in the same sense  as  the Feshbach-Schur map  discussed below. The map $\mathcal{R}_{\rho}$ removes the photon degrees of freedom of energy higher than $\rho$. Under certain assumptions  for $H^{(0)}(\vec{p},z)$,  the map $\mathcal{R}_{\rho}$ can be iterated indefinitely. We then get a sequence of operators $H^{(N)}(\vec{p},z):=(\mathcal{R}^{N}_{\rho} H^{(0)})(\vec{p},z)$ on $\mathcal{H}_{\text{red}}$. We  prove  in Subsection \ref{itere} that  the sequence of operators $(W^{(N)}_{m,n}(\vec{p},z))_{N \in \mathbb{N}} $, $m+n \geq 1$,  tends to zero in norm when $N$ tends to infinity and that  $W^{(N)}_{0,0}(\vec{p},0)=w_{0,0}^{(N)}(\vec{p},0,H_f, \vec{P_f})$ tends to an element $\alpha(\vec{p})  H_f + \vec{\beta} (\vec{p}) \cdot \vec{P}_{f}$ of  the fixed-points manifold $\mathcal{M}_{fp}:=\mathbb{C} H_{f}+ \mathbb{C}P_{f,x}+\mathbb{C} P_{f,y}+ \mathbb{C}P_{f,z}$ of  the flow   $\mathcal{R}^{N}_{\rho}$. This marginal operator has a non-zero eigenvector, the vacuum,  and, using the isospectral character of the renormalization map,  we can reconstruct the one-particle state of momentum $\vec{p}$  and its associated eigenvalue, for any $\vec{p} \in U[\vec{p}^*]$.   This procedure is explained in Subsection \ref{itere}.  One of the key points of the proof concerns the control of $ W^{(N)}_{0,0}(\vec{p},z)=w_{0,0}^{(N)}(\vec{p},z,H_f,\vec{P}_f)$. In particular, to iterate the renormalization map, we need to show that the restriction of  $w_{0,0}^{(N)}(\vec{p},z,H_f,\vec{P}_f)$ to the range of $\mathds{1}_{H_f \geq 3\rho/4}$ is bounded invertible, for all $N \in \mathbb{N}$ and for   sufficiently small values of $\vert z \vert$.  For $N$ large enough, $w_{0,0}^{(N)}(\vec{p},0,H_f,\vec{P}_f)$ is close to  $\alpha(\vec{p})  H_f + \vec{\beta} (\vec{p}) \cdot \vec{P}_{f}$, and it is important to check that $\alpha(\vec{p})$ stays  close to $1$ and that $\vec{\beta} (\vec{p})$ stays  close to $-\vec{p}/m$ if we want $w_{0,0}^{(N)}(\vec{p},z,H_f,\vec{P}_f)$ to be invertible on the range of  $\mathds{1}_{H_f \geq 3\rho/4}$. For momenta $\vec{p}$ in the open set $U[\vec{p}^*]$, we are able to control the size of the deviations $\vert \alpha(\vec{p})- 1 \vert$ and $\vert \vec{\beta}(\vec{p})+ \vec{p}/m \vert$ by a fine tuning of the coupling constant $\lambda_0$. This fine tuning depends on the parameter $ \mu= (m-\vert \vec{p}^*\vert)/2m$.  

\subsubsection{Renormalization preserves analyticity} 
A key ingredient of the proof of Theorem \ref{main1} is the fact that the renormalization map preserves analyticity. If $(\vec{p},z) \mapsto H^{(0)}(\vec{p},z)$ is analytic on $U[\vec{p}^*] \times D_{\mu/2}$, we can show inductively  that the operator-valued functions $(\vec{p},z) \mapsto H^{(N)}(\vec{p},z)$ are analytic on $U[\vec{p}^*] \times D_{\mu/2}$, for all $N \geq 1$. We prove this result in Subsection  \ref{presan}. The analyticity of the eigenvalues and the eigenprojections is   established in Section \ref{fin}.

\subsubsection*{Acknowledgement} J. Fa. is grateful to I.M. Sigal for many useful discussions. His research is supported by ANR grant ANR-12-JS01-0008-01. J.Fr. thanks V.Bach, T.Chen, A.Pizzo and I.M. Sigal for numerous illuminating discussions  on problems related to the ones studied in this paper. B.S. thanks V.Bach and M. Ballesteros for interesting discussions.

\section{Spectral analysis tools}
\label{renorm}
In this section, we explain in details the spectral tools sketched in Subsection \ref{out}. We begin with the ``smooth'' isospectral decimation method introduced in  \cite{BaChFrSi03_01} and further improved in \cite{GrHa08_01}. This method maps a closed operator acting on a Hilbert space $\mathcal{H}$ to a new operator acting on a subspace of $\mathcal{H}$. We refer the reader to  \cite{BaChFrSi03_01,GrHa08_01} for proofs.

\subsection{The Feshbach-Schur map}
\label{Fesch}
\subsubsection{Heuristic derivation of the Feshbach-Schur map}
\label{Fesch1}
Let $\mathcal{H}$ be a separable Hilbert space and let $H$ be a self-adjoint operator on  $\mathcal{H}$. We assume that $H$ has an eigenvalue $E$, with associated eigenvector $\Psi$. Let $P$ be an orthogonal projection with the property that $P\Psi  \neq 0$. Introducing $P$ and $P^{\perp}:=1-P$ on both sides of the identity $H \Psi=E \Psi$, we obtain that
\begin{align}
( PHP +PHP^{\perp})\Psi &=E P \Psi,\\
 ( P^{\perp}HP +P^{\perp}HP^{\perp})\Psi&=E P^{\perp}\Psi.
 \label{333}
\end{align}
Now we  assume that $P^{\perp}(H-E)P^{\perp}$ is bounded invertible on $P^{\perp} \mathcal{H}$.  In this case, we deduce from \eqref{333} that
\begin{equation} 
P^{\perp}\Psi=-(P^{\perp} (H-E)P^{\perp})^{-1} P^{\perp} HP\Psi.
\end{equation}
Reporting into (2.1),
\begin{equation} 
P(H-E)P\Psi- PH  (P^{\perp} (H-E)P^{\perp})^{-1} P^{\perp} HP  \Psi=0.
\end{equation}
 $F_{P}(H ):=PHP- PH  (P^{\perp}HP^{\perp})^{-1} P^{\perp} HP$   is called the Feshbach-Schur map. The operator $F_{P}(H)$ is defined on  $P  \mathcal{H}$ and has the remarkable property that $ \Psi \in \text{ker}(F_{P}(H-E))$. There is a one-to-one correspondence between  $\text{Ker}(F_{P}(H-E))$ and $\text{Ker}(H-E)$. In particular,  $F_{P}(H-E) P   \phi =0$ if and only if $(H-E) (P-P^{\perp} (P^{\perp}(H-E)P^{\perp})^{-1} P^{\perp} HP) \phi=0$. We can therefore reconstruct the eigenvectors of $H$ with eigenvalue $E$  from the kernel of $F_{P}(H-E)$.

\subsubsection{The smooth Feshbach-Schur map} The smooth Feshbach-Schur map is a generalization of the discrete Feshbach-Schur map discussed above. We generalize the construction to closed operators $H$ and to  bounded operators $\chi$ and $\overline{\chi}$ that share some common properties with the orthogonal projection $P$ above. Let $\chi$ and $\overline{\chi}$ be two commuting and non zero bounded operators on $\mathcal{H}$ satisfying the identity $\chi^2+\overline{\chi}^2=\mathds{1}$. For any subspace $V \subset \mathcal{H}$, we say that an operator $B: D(B) \subset \mathcal{H} \rightarrow \mathcal{H}$ is bounded invertible in $V$ if $B: D(B) \cap V \rightarrow V$ is a bijection with bounded inverse.

\begin{definition}[Feshbach pair]
\label{Fesh}
Let $H$ and $T$ be two closed operators defined on the same domain $D \subset \mathcal{H}$. We define $W=H-T :D \rightarrow \mathcal{H}$ and introduce the operators
\begin{equation}
H_{\chi}:= T+\chi W \chi, \qquad \qquad H_{\overline{\chi}}:= T+ \overline{\chi} W \overline{\chi}.
\end{equation}
We say that $(H,T)$ is a Feshbach pair for $\chi$ if it satisfies the  following properties:
\begin{enumerate}[(a)]
\item $\chi T \subset T \chi$,  $\overline{\chi} T \subset T \overline{\chi}$,
\item $T$ and $H_{\overline{\chi}}$ are bounded invertible from $D \cap \mathrm{Ran}(\overline{\chi})$ to $\mathrm{Ran}(\overline{\chi})$,
\item $\overline{\chi} H^{-1}_{\overline{\chi}} \overline{\chi} W \chi :D \rightarrow \mathcal{H}$ is bounded.
\end{enumerate}
\end{definition}

If $(H,T)$ is a Feshbach pair for $\chi$, we can define the Feshbach-Schur map $F_{\chi}(H,T): D \rightarrow \mathcal{H}$ by
\begin{equation}
\label{Fe}
F_{\chi}(H,T) := H_{\chi}-\chi W \overline{\chi} (H_{\overline{\chi}})^{-1} \overline{\chi}  W \chi.
\end{equation}
The most important feature of the map $(H,T) \mapsto F_{\chi}(H,T)$ is its isospectrality, meant in the following sense:
\begin{theorem}[\cite{BaChFrSi03_01},\cite{GrHa08_01}] 
\label{th1}
Let $(H,T)$ be a Feshbach pair for $\chi$ with associated Feshbach-Schur map $ F_{\chi}(H,T)$.  Let $V$ be a subspace  with $\mathrm{Ran}(\chi) \subset V \subset \mathcal{H}$, and  such that
$$T: D \cap V \rightarrow V, \qquad \overline{\chi}T^{-1} \overline{\chi} V \subset V.$$
 
\begin{enumerate}[(1)]
\item   $H: D \rightarrow \mathcal{H}$ is bounded invertible iff $ F_{\chi}(H,T): V \cap D \rightarrow V$ is bounded invertible in $V$.  Furthermore, if we define the auxiliary operators
\begin{equation}
\label{QQ}
Q_{\chi}:=\chi - \overline{\chi} H^{-1}_{\overline{\chi}} \overline{\chi} W \chi, \qquad Q^{\sharp}_{\chi}:=\chi - \chi W  \overline{\chi}  H^{-1}_{\overline{\chi}} \overline{\chi}, 
\end{equation}
then
\begin{equation}
H^{-1}= \overline{\chi} H^{-1}_{\overline{\chi}}\overline{\chi} + Q_{\chi} F_{\chi}^{-1}(H,T)  Q^{\sharp}_{\chi}, \qquad  F^{-1}_{\chi}(H,T) =\chi H^{-1} \chi + \overline{\chi} T^{-1} \overline{\chi}.
\end{equation}
\item  The restrictions $\chi: \mathrm{Ker}(H) \rightarrow \mathrm{Ker}(F_{\chi}(H,T))$ and $Q_{\chi} : \mathrm{Ker}(F_{\chi}(H,T)) \rightarrow \mathrm{Ker}(H) $ are linear isomorphisms and  inverse to each other.
\end{enumerate}
\end{theorem}
We reformulate the results (1) and (2) in order to  highlight  the isospectral property of the Feshbach-Schur map:  (1) implies that $0$ is in the resolvent set of $H$ if and only if  it  is in the resolvent set of $F_{\chi}(H,T)$. (2) means that the kernels of $H$ and $F_{\chi}(H,T)$ are isomorphic and that   the eigenvectors of $H$ and $F_{\chi}(H,T)$ corresponding to the eigenvalue $0$ are in one-to-one correspondence: If $\psi$ satisfies $H \psi=0$, then $\phi:= \chi \psi$ solves  $F_{\chi}(H,T) \phi=0$. If $\phi \in \text{Ran}(\chi)$ satisfies $F_{\chi}(H,T) \phi=0$, then $\psi:=Q_{\chi} \phi$ solves  $H \psi=0$.

\begin{remark}
$\quad$
\begin{itemize}
\item As we wish to study the spectrum of $H$, we consider the operator $H- z \mathds{1}=T+ W-z \mathds{1}$ instead of $H$.  If   $H \psi=z\psi$, then $F_{\chi}(H-z \mathds{1},T-z \mathds{1}) \chi \psi=0$. Furthermore, if  $F_{\chi}(H-z \mathds{1},T-z \mathds{1})\phi=0$, then $HQ_{\chi} \phi =zQ_{\chi} \phi$.
\vspace{1.5mm}
\item In the following, we choose for $\chi$ an operator of the type $\chi(H_f)$, where $H_{f}$ is the free field Hamiltonian operator defined in \eqref{Hredef}, and $\chi$ is a positive smooth function with compact support included in the interval $[0,1]$. The operator  $\chi(H_f)$ is defined by  the functional calculus for self-adjoint operators.  The smoothness of $\chi$ is useful in our setting because we will need to take (and to bound) derivatives of $\chi$.  Moreover, we choose for $V$  a closed subspace of $\mathcal{H}$ that contains the closure of the range of $\chi(H_{f})$.
\end{itemize}
\end{remark}

\subsection{Wick monomials and their kernels}
We already outlined in Subsection \ref{out} that the operator constructed from $H(\vec{p})-z \mathds{1}$ with the help of the smooth Feshbach-Schur map can be cast into the form
\begin{equation}
\label{eqH0}
H^{(0)}(\vec{p},z) =  \sum_{m+n \geq 0}  W_{m,n}^{(0)}(\vec{p},z), 
\end{equation}
where the operators $W^{(0)}_{m,n}(\vec{p},z)$, $m+n \geq 0$, are defined in the sense of quadratic forms by 
\begin{align}
W^{(0)}_{m,n}(\vec{p},z) := \mathds{1}_{H_f \leq 1} \int_{\underline{B}_{1}^{m+n}} &    \left( \prod_{i=1}^{m} b^{*}(\underline{k}_i) \right) w^{(0)}_{m,n} (\vec{p},z, H_f , \vec{P}_f ,\underline{k}_1,...,\underline{k}_{m},\underline{\tilde{k}}_1,...,\underline{\tilde{k}}_n ) \notag \\
&\quad  \left( \prod_{j=1}^{n} b(\underline{\tilde{k}}_j) \right) \prod_{i=1}^{m} d \underline{k}_i    \prod_{j=1}^{n} d \underline{\tilde{k}}_j   \mathds{1}_{H_f \leq 1}.   \label{eq:Wmn}
\end{align} 

\noindent The calculation leading to \eqref{eqH0} is displayed explicitly in Appendix \ref{polydlem}. For the time being, we focus on the properties of the kernels $w^{(0)}_{m,n}$, for $m+n \geq 0$. Lemma \ref{preserved} below  shows that  these properties are preserved under the renormalization map if  the initial sequence of kernels $\underline{w}^{(0)}$ lies in a suitably small polydisc $\mathcal{B}(\gamma , \delta  , \varepsilon)$.

\subsubsection{Notation}
\label{notation}
It is a little difficult  to come up with concise formulae involving the  kernels $w_{m,n}$,  because  these kernels depend on many arguments.  This is the reason why we introduce the following shorthand notations that are used throughout our text:
\begin{align*}
&\underline{k}^{(m)} := (\underline{k}_1,\dots,\underline{k}_m)\in \underline{\mathbb{R}}^{3m}, \quad \underline{\tilde{k}} \phantom{}^{(n)} := (\underline{\tilde{k}}_1,\dots,\underline{\tilde{k}}_n)\in \underline{\mathbb{R}}^{3n} , \\
&\underline{K}^{(m,n)} := ( \underline{k}^{(m)} , \underline{\tilde{k}} \phantom{}^{(n)} ) , \quad d \underline{K}^{(m,n)} := \prod_{i=1}^m d \underline{k}_i  \prod_{j=1}^n d \underline{\tilde{k}}_j  , \\
& | \underline{K}^{(m,n)} | := | \underline{k}^{(m)} | \, | \underline{\tilde{k}} \phantom{}^{(n)} | , \quad | \underline{k}^{(m)} | := \prod_{i=1}^m | \vec{k}_i | , \quad |Ê\underline{\tilde{k}} \phantom{}^{(n)} | := \prod_{j=1}^n | \vec{\tilde{k}}_j | , \\
& \Sigma [ \underline{k}^{(m)} ] := \sum_{i=1}^m | \vec{k}_i | , \quad \Sigma[ \underline{\tilde{k}} \phantom{}^{(n)} ] := \sum_{j=1}^n | \vec{\tilde{k}}_j | , \quad \vec{\Sigma} [ \underline{k}^{(m)} ] := \sum_{i=1}^m \vec{k}_i  , \quad \vec{\Sigma}[ \underline{\tilde{k}} \phantom{}^{(n)} ] := \sum_{j=1}^n \vec{\tilde{k}}_j ,  \\
&b^*( \underline{k}^{(m)} ) := \prod_{ i = 1 }^m b_{\lambda_i}^*( \vec{k}_i ) , \quad b( \underline{\tilde{k}} \phantom{}^{(n)} ) := \prod_{ j = 1 }^n b_{\lambda_j}( \vec{\tilde{k}}_j ) .
\end{align*}
\vspace{1mm}
\noindent
For $\rho \in \mathbb{C}$, we set 
\begin{align*}
\rho \underline{k}^{(m)} := ( \rho \vec{k}_1 , \lambda_1 , \dots , \rho \vec{k}_m , \lambda_m ), \quad \rho \underline{K}^{(m,n)} := ( \rho \underline{k}^{(m)} , \rho \underline{ \tilde{k} } \phantom{}^{(n)} ).
\end{align*}
We remind the reader that  $B_1 = \{ \vec{k} \in \mathbb{R}^3 , | \vec{ k} | \le 1 \}$ and we  introduce  the set
\begin{equation}
\label{BB}
\mathcal{B} := \{ ( r , \vec{l} ) \in [ 0 , 1 ] \times B_1 , | \vec{l} | \le r \}.
\end{equation}

Many numerical constants appear in our estimates. Keeping track of all these constants would be very cumbersome and is not necessary for mathematical rigor.  The shorthand
 \begin{equation}
\text{``}\text{for all } a \ll b , \dots \text{''}
 \end{equation}
means that ``there exists a (possibly very small, but) positive numerical constant $C$ such that, for all $a \le C b$, \dots'' .
\vspace{1mm}

\subsubsection{Kernels}
In Section \ref{firstep}, we will show that, for all $\vec{p} \in U[\vec{p}^*] =\lbrace \vec{p} \in \mathbb{C}^3  \mid  \vert \vec{p} - \vec{p}^* \vert < \mu m \rbrace$ and $z \in D_{ \mu / 2 } = \{ z \in \mathbb{C} , |z| < \mu / 2 \}$, the kernels $w^{(0)}_{m,n}$, $m+n \ge 0$, defined in \eqref{eqH0}--\eqref{eq:Wmn} belong to certain Banach spaces of functions that we introduce in the following definition.


\begin{definition}[The Banach spaces $\mathcal{W}^{\sharp}_{0,0}$ and  $\mathcal{W}^{\sharp}_{m,n}$]\label{def_kernels}
$\quad$ 
\begin{itemize} 
\item  The set of functions $w_{0,0} \in \mathrm{C}^1( \mathcal{B} ; \mathbb{C} )$ equipped with the norm
\begin{equation}
\label{taunorm}
 \| w_{0,0} \|^{\sharp} := | w_{0,0}( 0 , \vec{0} ) | + \| \partial_{r} w_{0,0} \|_{\infty} + \sum_{i=1}^{3} \| \partial_{l_i} w_{0,0} \|_\infty\end{equation}
is denoted by $\mathcal{W}^{\sharp}_{0,0}$; $\mathcal{W}^{\sharp}_{0,0}$ defines a Banach space. We remind the reader that $f \in \mathrm{C}^1( \mathcal{B} ; \mathbb{C} )$ if $f$ is continuous on $\mathcal{B} $, $\mathrm{C}^1$ in the interior of  $\mathcal{B}$, and if its partial derivatives can be extended to $\partial \mathcal{B}$ by continuity.
\vspace{2mm}

  \item For $m+n \ge 1$, the set of functions $w_{m,n} : \mathcal{B} \times \underline{B}_{1}^{m}\times \underline{B}_{1}^{n}    \rightarrow  \mathbb{C}$ that are measurable on  $ \underline{B}_{1}^{m}\times \underline{B}_{1}^{n} $,  totally symmetric on $ \underline{B}_{1}^{m}$ and $ \underline{B}_{1}^{n}$, of class $\mathrm{C}^1(\mathcal{B} )$ for almost every $\underline{K}^{(m,n)} \in \underline{B}_{1}^{m}\times \underline{B}_{1}^{n}$,  and that obey the norm bound
 \begin{equation}
  \|w_{m,n}\|^{\sharp}  :=\|w_{m,n}\|_{\frac{1}{2}}+\|\partial_r w_{m,n} \|_{\frac{1}{2}} + \sum_{i=1}^{3} \| \partial_{l_i} w_{m,n} \|_{\frac{1}{2}}< \infty , \label{sharp}
\end{equation}
where
\begin{equation}\label{def_norme_f_m,n}
\| w_{m,n} \|_{\frac{1}{2}} := \underset{ ( r , \vec{l} ) \in \mathcal{B} }{ \sup } \text{ } \text{ }\underset{  \underline{K}^{(m,n)} \in \underline{B}_1^{m+n}}{\mathrm{ess} \ \mathrm{sup}} \left| w_{m,n}( r , \vec{l} ,\underline{K}^{(m,n)} ) \right |  | \underline{K}^{(m,n)}|^{-1/2} ,
\end{equation}
  defines a Banach space that we denote by  $\mathcal{W}^{\sharp}_{m,n}$.
  \end{itemize}
  \end{definition}

 The choice of the exponent in the factor $| \underline{K}^{(m,n)}|^{-1/2}$ in \eqref{def_norme_f_m,n} 
is related to the infrared behavior of the model we consider, and insures an optimal rate of convergence to $0$ of the renormalized kernels $w_{m,n}^{(N)}$, $m+n \ge 1$, obtained by the renormalization procedure. More precisely, if the form factor $|k|^{1/2}$ in the interaction Hamiltonian $H_I$ in \eqref{Hredef} is replaced by $|k|^{-1/2 + \varepsilon}$, our method would work in the same way, provided that $| \underline{K}^{(m,n)}|^{-1/2}$ in \eqref{def_norme_f_m,n} is replaced by $| \underline{K}^{(m,n)}|^{1/2-\varepsilon}$. Furthermore,  we will prove below estimates of the form $\| w_{m,n}^{(N)} \|^{\sharp} = \mathcal{O}( \rho^{N+1} )$ where $\rho < 1$ is a scale parameter. Replacing $| \underline{K}^{(m,n)}|^{-1/2}$ in \eqref{def_norme_f_m,n} by $| \underline{K}^{(m,n)}|^{1/2-\varepsilon}$ with $0 < \varepsilon \le 1$ would lead to estimates of the form $\| w_{m,n}^{(N)} \|^{\sharp} = \mathcal{O}( \rho^{\varepsilon (N+1)} )$, $m+n \geq 1$.

To a kernel $w_{0,0} \in \mathcal{W}_{0,0}^{\sharp}$, we can associate the bounded operator $w_{0,0}( H_f , \vec{P}_f ) \mathds{1}_{ H_f \le 1}$ defined by the functional calculus of Subsection \ref{out}. The choice of the norm $\| \cdot \|^{\sharp}$ will allow us to express the fact that the ``free'' effective Hamiltonian $w_{0,0}^{(0)} (\vec{p},z,H_f, \vec{P_f})$ in \eqref{eqH0} is close to the operator $H_f - m^{-1} \vec{p} \cdot \vec{P}_f - z$, in the sense that
\begin{equation*}
\| w_{0,0}^{(0)}(\vec{p}, z , r , \vec{l}) - ( r -  m^{-1}  \vec{  p}  \cdot \vec{l} - z) \|^{\sharp}
 \end{equation*}
tends to $0$ as the coupling constant $\lambda_0 \to 0$ (see Section \ref{firstep}). We remark that another possible choice for $\| w_{0,0} \|^{\sharp}$ would be given by the equivalent norm
\begin{equation*}
\| w_{0,0} \|_\infty + \| \partial_{r} w_{0,0} \|_{\infty} + \sum_{i=1}^{3} \| \partial_{l_i} w_{0,0} \|_\infty.
\end{equation*}

To a kernel $w_{m,n} \in \mathcal{W}_{m,n}^{\sharp}$, we associate the Wick monomial
\begin{align}\label{definition_WMN0}
&W_{m,n}( w_{m,n} ) :=  \mathds{1}_{ H_f \le 1 }  \int_{\underline{B}_{1}^{m+n}}  b^{*}(\underline{k}^{(m)})  w_{m,n} [ H_f , \vec{P}_f , \underline{K}^{(m,n)} ] b(\underline{\tilde{k}} \phantom{}^{(n)}) d \underline{K}^{(m,n)}  \mathds{1}_{ H_f \le 1 } ,
\end{align}
where, for a.e. $\underline{K}^{(m,n)} \in \underline{B}_{1}^{m+n}$, $w_{m,n} [ H_f , \vec{P}_f , \underline{K}^{(m,n)} ]$ is defined by the functional calculus of Subsection \ref{out}. The expression in  \eqref{definition_WMN0} defines a quadratic form on the set of vectors in $\mathcal{H}_{ f }$ with finitely many particles. The following lemma shows that it extends to a bounded quadratic form on $\mathcal{H}_{ \mathrm{red} }$. The associated bounded operator will be denoted by the same symbol.
\begin{lemma}\label{opboundl}
Let $m+n \ge 1$. For all $w_{m,n} \in \mathcal{W}^{\sharp}_{m,n}$, the quadratic form defined in \eqref{definition_WMN0} extends to a bounded quadratic form with norm satisfying
\begin{equation}
\label{opbound}
\| W_{m,n} ( w_{m,n} ) \| \le (m! n!)^{-\frac12}  (8 \pi)^{(m+n)/2} \| w_{m,n} \|_{\frac{1}{2}}.
\end{equation}
\end{lemma}
Lemma \ref{opboundl} asserts that we can control the norm of Wick monomials by controlling the norm of their associated kernels.  We remark that we considered a $L^\infty$-norm in \eqref{def_norme_f_m,n} while a $L^2$-norm is used in \cite{BaChFrSi03_01}. Moreover the operators $w_{m,n} [ H_f , \vec{P}_f , \underline{K}^{(m,n)} ]$ in \eqref{definition_WMN0} depend both on $H_f$ and $\vec{P}_f$ while the corresponding operators considered in \cite{BaChFrSi03_01} only depend on $H_f$.  Nevertheless, the proof of Lemma \ref{opboundl} is a straightforward adaptation of the one of \cite[Theorems 3.1]{BaChFrSi03_01} (the proof is in fact slightly easier with the $L^\infty$-norm defined in \eqref{def_norme_f_m,n} than with the $L^2$-norm used in \cite{BaChFrSi03_01}).

\subsubsection{Hamiltonians associated with a sequence of kernels}  
We want to bound the series of Wick monomials in \eqref{eqH0}. This amounts to assume that the kernels $w_{m,n}$ satisfy some summability properties. We  introduce the Banach space
\begin{equation}
\label{Wsharp}
\mathcal{W} ^{\sharp}:=  \bigoplus_{m+n\geq 0} \mathcal{W}^{\sharp}_{m,n},
\end{equation}
equipped with the norm 
\begin{equation}
\|\underline{w}\|^{\sharp}_{\xi}:=\sum_{m+n\geq 0}\xi^{-(m+n)}\|w_{m,n}\|^{\sharp},
\end{equation}
where $\xi$ is a fixed positive number smaller than one, and $\underline{w}=(w_{m,n})_{m+n \geq 0}$. Let $\mathcal{H}_{\text{red}}$ be the subspace of photons of energy smaller than 1, i.e.
\begin{equation}
\mathcal{H}_{ \mathrm{red} } := \mathds{1}_{ H_f \le 1 } \mathcal{H}_f .
\end{equation}
In order to associate with any element of $\mathcal{W}^{\sharp}$ a bounded  operator (similar to the one in \eqref{eqH0}), we introduce a map $H:\mathcal{W} ^{\sharp}\rightarrow \mathcal{B}(\mathcal{H}_{\text{red}})$.
\begin{definition}[The linear map $H(\underline{w})$]\label{def:def_H(w)}
For all $\underline{w} \in \mathcal{W}^{\sharp}$, we set
\begin{equation}\label{def_H(w)}
H(\underline{w}) := \sum_{m+n\geq 0} W_{m,n} ( \underline{w} ) = W_{0,0}( \underline{w} ) + W_{ \ge 1 }( \underline{w} ),
\end{equation}
where
\begin{equation}\label{W00}
W_{0,0}( \underline{w} ) := w_{0,0}[H_f , \vec{P}_f]   \mathds{1}_{ H_f \le 1 }  , \quad W_{ \ge 1 }( \underline{w} ) := \sum_{m+n\geq 1}W_{m,n}( \underline{w} ), 
\end{equation}
and, for $m+n\geq 1$, $W_{m,n}( \underline{w} ) := W_{m,n}( w_{m,n} )$ is the Wick monomial defined in \eqref{definition_WMN0}.
\end{definition}
The following lemma shows that $H( \underline{w} )$ defines a bounded operator on $\mathcal{H}_{ \mathrm{red} }$ for all $\underline{w} \in \mathcal{W}^{\sharp}$. It is a direct consequence of Lemma \ref{opboundl}.
\begin{lemma}\label{lm:Hbounded_injective}
For all $0 < \xi < 1/\sqrt{8\pi}$, the linear map $H : \mathcal{W}^{\sharp} \to \mathcal{B} ( \mathcal{H}_{ \mathrm{red} } )$ defined by $\eqref{def_H(w)}$  satisfies, 
\begin{equation}
\label{ineqnorm}
\left \| H(\underline{w}) \right\|  \le \| \underline{w} \|^{\sharp}_{\xi},
\end{equation}
for all $\underline{w} \in \mathcal{W}^{\sharp}$.
\end{lemma}
\begin{remark}
As in \cite[Theorem 3.3]{BaChFrSi03_01} (see also \cite[Theorem 5.4]{HaHe11_01}), we can verify that the map $H$ is injective provided that the kernels $w_{m,n}$ are restricted to the set 
\begin{equation*}
\mathcal{B}_{m,n} := \{ ( r , \vec{l} , \underline{K}^{(m,n)}) \in \mathcal{B} \times \underline{B}_{1}^{m}\times \underline{B}_{1}^{n} ,  r \le 1 - \max \big ( \Sigma [ \underline{k}^{(m)} ] , \Sigma [ \underline{k}^{(n)} ] \big ) \big \} .
\end{equation*}
The proof is similar to that of Lemma \ref{compo} given in Appendix \ref{lemacom}.
\end{remark}

Next, we take into account the fact that the operators considered in the following depend on the total momentum $\vec{p}$ and the spectral parameter $z$. We remind the reader that $U[\vec{p}^*] $ is the set of complex momenta $\vec{p}$ such that $\vert \vec{p}-\vec{p}^* \vert<\mu m$, where $\mu=( m-\vert \vec{p}^* \vert)/2m$, and that $D_{\mu/2}$ is the complex open disc of radius $\mu/2$ around the origin. We will assume that $(\vec{p} , z ) \in U[ \vec{p}^* ] \times D_{ \mu / 2 }$.

The kernels $w_{m,n}(\vec{p},z,r,\vec{l}, \underline{K}^{(m,n)})$ can  be viewed as functions 
\begin{align*}
U [ \vec{p}^* ] \times D_{ \mu / 2 } \ni (\vec{p},z) \mapsto  w_{m,n}(\vec{p},z) \in \mathcal{W}_{m,n}^{\sharp}.
\end{align*}
The set of functions from $U [ \vec{p}^* ] \times D_{ \mu / 2 }$ to $\mathcal{W}_{m,n}^{\sharp}$ is denoted by $ \mathcal{W}_{m,n}$. Likewise, the set of functions from $U [ \vec{p}^* ] \times D_{ \mu / 2 }$ to $\mathcal{W}^{\sharp}$ is denoted by $\mathcal{W}$.

Following \cite{FrGrIr08}, we introduce a notation that proves to  be convenient for our analysis: We denote the set of  functions $U [ \vec{p}^* ] \times D_{ \mu / 2 } \ni (\vec{p},z) \mapsto H(\underline{w}(\vec{p},z)) \in \mathcal{B}(\mathcal{H}_{\text{red}})$ by $\mathcal{W}_{\text{op}}$.

\subsection{The renormalization map}
We define the renormalization map that we shall use in the sequel. Our construction is similar to the one in \cite{BaChFrSi03_01}. The main difference, as already appears in Definition \ref{def_kernels}, is that the kernels we consider depend on the variables $(r,\vec{l})$ corresponding to the operators $(H_f,\vec{P}_f)$, while the kernels in \cite{BaChFrSi03_01} only depend on $r$.

Let $0<\rho<1$ be a fixed scale parameter. The renormalization map maps any operator $H(\underline{w}) \in \mathcal{B}(\mathcal{H}_{\text{red}})$  defined in \eqref{def_H(w)}, to a new operator $\mathcal{R}_{\rho}H(\underline{w}) \in \mathcal{B}(\mathcal{H}_{\text{red}})$. The spectrum  of $\mathcal{R}_{\rho}H(\underline{w})$ may be easier to analyze than the spectrum of $H(\underline{w})$, because $\mathcal{R}_{\rho}$ eliminates the degrees of freedom of energy bigger than $\rho$. Thanks to the isospectrality of $\mathcal{R}_{\rho}$, we can reconstruct some spectral properties of the initial operator $H(\underline{w})$. The renormalization map is  a composition of three distinct transformations:
\begin{itemize}
\item An analytic transformation, $E_\rho$, of the spectral parameter $z$,
\item An application of the Feshbach-Schur map,
\item A scale transformation $S_\rho$.
\end{itemize}
\vspace{2mm}

In what follows, we give an overview of each of these transformations.  The $(\vec{p},z)$ dependence of the kernels is kept explicit when needed.

\subsubsection{The scale transformation $S_\rho$ }
\label{rescaling}
For $\rho > 0$,  we define the unitary map $\Gamma_\rho$ on $\mathcal{H}_f$ by
\begin{align*}
& \Gamma_\rho \, \Omega := \Omega , \\ 
& \big ( \Gamma_\rho \, \Psi \big )^{(p)} ( \underline{k}^{(p)} ) :=  \rho^{\frac{3p}{2} } \Psi^{(p)} ( \rho \underline{k}^{(p)}) .
\end{align*}
We recall the following scaling properties that can be verified easily, using the definitions of the operators involved:
\begin{align}\label{scaling_properties}
\Gamma_\rho  H_f \Gamma_\rho^* = \rho H_f , \qquad \Gamma_\rho  \vec{P}_f \Gamma_\rho^* = \rho \vec{P}_f .
\end{align}
This implies that $\Gamma_\rho : \mathrm{Ran} ( \mathds{1}_{ H_f \le \rho } ) \to \mathrm{Ran} ( \mathds{1}_{ H_f \le 1 } )$. Recalling the definition $\mathcal{H}_{ \mathrm{red} } = \mathds{1}_{ H_f \le 1 } \mathcal{H}_f$, we define the map $S_\rho : \mathcal{B}( \mathcal{H}_{ \mathrm{red} } ) \to \mathcal{B}( \mathcal{H}_{ \mathrm{red} } )$ by
\begin{align*}
S_\rho(A) := \rho^{-1} \Gamma_\rho A \Gamma_\rho^* ,
\end{align*}
for any bounded operator $A$ on $\mathcal{H}_{ \mathrm{red} }$. It follows from \eqref{scaling_properties} that
\begin{align*}
S_\rho( H_f ) = H_f , \qquad S_\rho( \vec{P}_f ) = \vec{P}_f ,\qquad   S_\rho( f(H_f,\vec{P}_{f} )) =  \rho^{-1} f(\rho H_f, \rho \vec{P_f}),
\end{align*}
for any measurable function $f$. Furthermore, one can verify that, for any $m,n \in \mathbb{N}\cup \{ 0 \}$,
\begin{equation*}
S_\rho ( W_{m,n}( \underline{w} ) ) = W_{m,n}( s_\rho( \underline{w} ) ),
\end{equation*}
with
\begin{equation}
\label{scaling}
 s_\rho ( \underline{w} )_{m,n} ( H_f , \vec{P}_f , \underline{K}^{(m,n)} )  := \rho^{ \frac{3}{2}(m+n)-1} w_{ m , n } (\rho H_f , \rho \vec{P}_f , \rho\underline{K}^{(m,n)} ).
\end{equation}
\vspace{1mm}

\noindent The definition (\ref{def_norme_f_m,n})  of the norm $ \|  \cdot  \|_{\frac{1}{2}}$ implies that  
\begin{align*}
 \|  s_\rho ( \underline{w} )_{m,n} \|_{\frac{1}{2}} &= \rho^{ \frac{3}{2}(m+n)-1} \underset{ ( r , \vec{l} ) \in \mathcal{B} }{ \sup } \ \  \underset{  \underline{K}^{(m,n)} \in \underline{B}_1^{m+n}}{\mathrm{ess} \ \mathrm{sup}} \left| w_{m,n}( \rho r , \rho \vec{l} , \rho \underline{K}^{(m,n)} ) \right |  | \underline{K}^{(m,n)}|^{-1/2} \\
 &= \rho^{ 2(m+n)-1} \underset{ ( r , \vec{l} ) \in \mathcal{B} }{ \sup } \ \  \underset{  \underline{K}^{(m,n)} \in \underline{B}_1^{m+n}}{\mathrm{ess} \ \mathrm{sup}} \left| w_{m,n}( \rho r , \rho \vec{l} , \rho \underline{K}^{(m,n)} ) \right |  | \rho \underline{K}^{(m,n)}|^{-1/2} \\
& \le  \rho^{ 2(m+n)-1}   \| w_{m,n} \|_{\frac{1}{2}}, 
\end{align*}
and likewise that
\begin{align*}
& \| \partial_r  s_\rho ( \underline{w} )_{m,n} \|_{\frac{1}{2}} \le \rho^{ 2(m+n)} \| \partial_r w_{m,n} \|_{\frac{1}{2}} , \\
& \| \partial_{l_i}  s_\rho ( \underline{w} )_{m,n} \|_{\frac{1}{2}} \le \rho^{ 2(m+n)} \| \partial_{l_i} w_{m,n} \|_{\frac{1}{2}} .
\end{align*}

As  $\|  s_\rho ( \underline{w} )_{m,n} \|_{\frac{1}{2}} $ controls the norm of the operators $S_\rho ( W_{m,n}( \underline{w} ) )$, these estimates show that the operators  $W_{m,n}( \underline{w} )$, $m+n \geq 1$, contract under   rescaling. Operators possessing this property are called \emph{irrelevant}. In contrast, $\mathds{1}$ expands under $S_\rho$ by a factor $\rho^{-1}$ and is called a \emph{relevant} operator. Terms linear in $H_{f}$ and $\vec{P}_f$ stay unchanged and are called \emph{marginal} operators. 
\vspace{2mm}

\subsubsection{Transformation of the spectral parameter $E_\rho$}
The operator $H(\underline{w})$, introduced in Definition \eqref{def:def_H(w)}, can be canonically decomposed into the sum of three bounded operators on $\mathcal{H}_{ \mathrm{red} }$:
\begin{itemize}
\item A marginal operator $w_{0,0}( H_f , \vec{P}_f ) - w_{0,0}( 0 , \vec{0} ) \mathds{1}$,
\item A relevant operator $w_{0,0}( 0 , \vec{0} ) \mathds{1}$,
\item An irrelevant operator $ W_{ \ge 1 }( \underline{w} )$.
\end{itemize}
We introduce a polydisc,  $\mathcal{B}(\gamma , \delta , \varepsilon)$,  contained in the space of operator-valued functions $\mathcal{W}_{\text{op}}$ defined above (the elements of $\mathcal{W}_{\text{op}}$ are the maps $U [ \vec{p}^* ] \times D_{ \mu / 2 } \ni (\vec{p},z) \mapsto H(\underline{w}(\vec{p},z))$). 
\begin{definition}[The polydisc $\mathcal{B}(\gamma , \delta , \varepsilon)$]
\label{discc}
Let $\gamma,\delta,\varepsilon>0$. The polydisc  $\mathcal{B}(\gamma , \delta , \varepsilon) \subset \mathcal{W}_{ \mathrm{op} }$ is defined by
\begin{align*}
\mathcal{B}(\gamma , \delta , \varepsilon) := \Big \{ & H(\underline{w}(\cdot,\cdot)) \in\mathcal{W}_{\mathrm{op}} , \phantom{\underset{z\in D_{1/2}}{\sup}} \\
&\underset{(\vec{p},z) \in U[\vec{p}^*] \times D_{\mu/2}}{\sup} \big \| w_{0,0}(\vec{p}, z , r , \vec{l})-w_{0,0}(\vec{p}, z , 0 , \vec{0})  - ( r -  m^{-1}  \vec{  p}  \cdot \vec{l} ) \big \|^{\sharp} \leq \gamma, \\
& \underset{(\vec{p},z) \in U[\vec{p}^*] \times D_{\mu/2}}{\sup} \big | w_{0,0}(\vec{p},z,0,\vec{0})+ z \big | \le \delta, \\
& \underset{(\vec{p},z) \in U[\vec{p}^*] \times D_{\mu/2}}{\sup} \big \| \underline{w}(\vec{p},z) \big \|_{\xi,\ge 1}^{\sharp}\le \varepsilon \Big \} ,
\end{align*}
where 
\begin{equation}
\| \underline{w} \|_{\xi,\ge 1}^{\sharp}:=\sum_{m+n \geq 1} \xi^{-(m+n)} \| w_{m,n} \|^{\sharp}, \qquad \xi < \frac{1}{\sqrt{8 \pi}}.
\end{equation}
\end{definition} 
An element of $\mathcal{B}(\gamma , \delta , \varepsilon)$ is close to the operator-valued functions $(\vec{p},z) \mapsto H_f - m^{-1} \vec{p} \cdot \vec{P}_f - z$, in the sense that 
\begin{itemize}
\item The marginal part $w_{0,0}(\vec{p}, z , r , \vec{l})-w_{0,0}(\vec{p}, z , 0 , \vec{0})$ is at a distance $\le \gamma$ of $r -  m^{-1}  \vec{  p}  \cdot \vec{l}$ (for the norm $\| \cdot \|^{\sharp}$), uniformly in $(\vec{p},z)$,
\item The relevant part $w_{0,0}(\vec{p},z,0,\vec{0})$ is at a distance $\le \delta$ of $-z$, uniformly in $(\vec{p},z)$,
\item The irrelevant part $(w_{m,n}(\vec{p} , z ))_{m+n\ge1}$ is smaller than $\varepsilon$ for the norm $\| \cdot \|_{\xi}^{\sharp}$, uniformly in $(\vec{p},z)$. 
\end{itemize}

The instability of  the identity operator $\mathds{1}$ under $S_{\rho}$ forces us to fine-tune the choice of the spectral parameter.  For $w_{0,0} \in \mathcal{W}_{0,0}$, we introduce the  set
\begin{align*}
\mathcal{U}[ w_{0,0} ] := \big \{ (\vec{p},z) \in  U[\vec{p}^*] \times D_{\mu/2} , \ | w_{0,0}(\vec{p},z,0,\vec{0}) | < \frac{ \mu \rho}{2} \big \}.
\end{align*}
If $(\vec{p},z) \mapsto w_{0,0}(\vec{p},z,0,\vec{0})$ is continuous,  $\mathcal{U}[ w_{0,0} ] $ is an open subset of $\mathbb{C}^4$. Setting
\begin{align*}
\mathcal{U}[ w_{0,0}(\vec{p}) ] := \big \{ z  \in    D_{\mu/2} , \ | w_{0,0}(\vec{p},z,0,\vec{0}) | < \frac{ \mu \rho}{2} \big \},
\end{align*}
we have that 
$$\mathcal{U}[ w_{0,0} ] = \lbrace (\vec{p},z)  \in  U[\vec{p}^*] \times D_{\mu/2} \mid  \mathcal{U}[ w_{0,0}(\vec{p}) ] \neq \emptyset \text{ and } z \in \mathcal{U}[ w_{0,0}(\vec{p}) ] \rbrace. $$
 The following picture illustrates how the values of the spectral parameter $z$ are restricted.

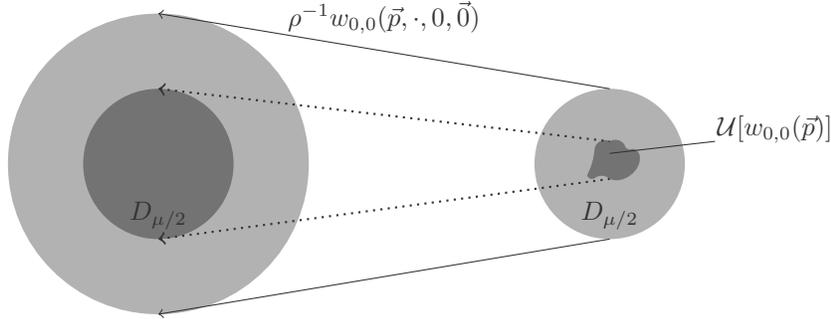
\begin{center}
\begin{tikzpicture}

      \fill[ color=gray!60] (0,0) circle (2) ;
      \fill[fill=black!55] (0,0) circle (1) ;

      \fill[color=gray!60] (6, 0) circle (1) ;
     
            \color{black!80}

    \fill[black!55]     (6,0.3)to[bend left=45] (6.2,0.25)  to[bend left=-60] (6.3,0.2)   to[bend left=40] (6.4,0.1)  to[bend left=60] (6,-0.2) to[bend left=-40] (5.8,-0.19) to[bend left=60] (5.7,-0.15) to[bend left=-10] (5.8,0.2) to[bend left=60] (6,0.3);

 \draw[<-] (0,2)--(6,1);
 \draw[<-] (0,-2)--(6,-1);
 
  \draw[<-,thick,dotted] (0,1)--(6,0.3);
 \draw[<-,thick,dotted] (0,-1)--(6,-0.2);
 
    \draw (6,-1) node[above] {\small $ D_{\mu/2}$};
   \draw (0,-1) node[above] {\small$ D_{\mu/2}$};
     \draw (8.2,0.15) node[above] { \small $\mathcal{ U}[w_{0,0}(\vec{p})] $};
    
     \draw (3,1.6) node[above] {\small$ \rho^{-1} w_{0,0}(\vec{p},\cdot, 0,\vec{0})$}; 
     \draw[-] (6,0.14)--(7.4,0.3);
     
\end{tikzpicture}

\end{center}

\begin{definition}[The rescaling map   $E_\rho$]
Let $w_{0,0} \in \mathcal{W}_{0,0}$. The rescaling map $E_\rho : \mathcal{U}[ w_{0,0} ] \to  U[\vec{p}^*] \times D_{\mu/2}$,  is defined by 
\begin{align}
\label{Erho}
 E_\rho (\vec{p}, z ) := (\vec{p}, - \rho^{-1} w_{0,0}(\vec{p},z,0,\vec{0})).
\end{align}
\end{definition}
\vspace{1mm}

Note that we only rescale $z$, and do not change the value of $\vec{p}$.
 The following lemma enables us to control the unstable manifold by rescaling the spectral parameter $z$ before carrying out the scale transformation $S_{\rho}$.
\begin{lemma}\label{lm:bihol}
Let $0 < \rho < 1/2$, $0 < \xi < 1/\sqrt{8 \pi}$, $ \gamma > 0$, $ 0 < \delta \ll \rho \mu$ and $\varepsilon > 0$. Let $H(\underline{w}(\cdot,\cdot)) \in \mathcal{B}( \gamma , \delta , \varepsilon )$ and assume that the map $( \vec{p} , z ) \mapsto w_{0,0}( \vec{p} , z ,  0 , \vec{0} ) \in \mathbb{C}$ is analytic on $U[ \vec{p}^* ]\times D_{ \mu / 2 }$. Then $\mathcal{U}[ w_{0,0} ] \neq \emptyset$ and the map $E_{\rho} : \mathcal{U}[ w_{0,0} ] \to U[\vec{p}^*] \times D_{\mu/2} $ is biholomorphic. Its inverse is denoted by $E_{\rho}^{-1}$.
\end{lemma}
\begin{proof}
The proof is similar to  \cite[Lemma 3.4]{BaChFrSi03_01} and \cite[Lemma 6.1]{HaHe11_01}, except that the models treated in these references describe non moving atoms  for which the map $E_{\rho}$ only depends on $z$. We  give a detailed proof  in our case, since $E_\rho$  also depends on the total momentum $\vec{p}$.
\vspace{2mm}

\noindent \textbf{Step 1}. First, we show that
\begin{align}
\label{inclu}
U[\vec{p}^*] \times D_{\mu \rho / 2 - \delta} \subset  \mathcal{U}[ w_{0,0} ]  \subset   U[\vec{p}^*] \times D_{ \mu  \rho / 2 + \delta }.
\end{align}
Let $(\vec{p},z) \in \mathcal{U}[ w_{0,0} ]$. Since $H(\underline{w}(\cdot,\cdot)) \in  \mathcal{B}( \gamma , \delta , \varepsilon )$, we have that
\begin{equation*}
\vert z \vert \leq \vert z+w_{0,0}(\vec{p},z,0,\vec{0}) \vert + \vert w_{0,0}(\vec{p},z,0,\vec{0}) \vert < \mu \rho / 2 + \delta ,
\end{equation*}
and hence the second inclusion in \eqref{inclu} follows.

Let $(\vec{p},z) \in U[\vec{p}^*] \times D_{\mu \rho / 2 - \delta}$. Using again that $H(\underline{w}(\cdot,\cdot)) \in  \mathcal{B}( \gamma , \delta , \varepsilon )$ gives
\begin{equation*}
\vert w_{0,0}(\vec{p},z,0,\vec{0}) \vert \le \vert w_{0,0}(\vec{p},z,0,\vec{0}) + z \vert + | z | < \mu \rho / 2  ,
\end{equation*}
and hence the first inclusion in \eqref{inclu} is proven. This shows that $\mathcal{U}[ w_{0,0} ] \neq \emptyset$, because $\delta \ll \rho \mu$ and, therefore, $ D_{\mu \rho / 2 - \delta} \neq \emptyset$. Furthermore, $\mathcal{U}[ w_{0,0}(\vec{p}) ] \neq \emptyset$ for all $\vec{p} \in U[\vec{p}^*]$.
\vspace{2mm}

\noindent \textbf{Step 2}.
We prove that $E_\rho : \mathcal{U}[ w_{0,0} ] \to E_\rho( \mathcal{U}[ w_{0,0} ] )$ is locally biholomorphic. To this end, by the inverse function theorem, it suffices to show that $\mathrm{det} \, d E_{\rho}(\vec{p},z) \neq 0$ for any $(\vec{p},z) \in \mathcal{U}[ w_{0,0} ]$. It follows directly from \eqref{Erho} that
\begin{equation*}
\mathrm{det} \, d E_{\rho}(\vec{p},z) = - \rho^{-1} \partial_z w_{0,0}( \vec{p} , z , 0 , \vec{0} ). 
\end{equation*}
Now, let $(\vec{p},z) \in \mathcal{U}[ w_{0,0} ]$. By Step 1, $| z |Ê< \mu \rho / 2 + \delta$, and hence, with $\mathcal{C} := \{ z' \in \mathbb{C} , |z'| = \mu / 3 \}$, Cauchy's formula gives
\begin{equation*}
\big \vert  \partial_z \big(  w_{0,0}(\vec{p} , z  ,0,\vec{0}) +z \big)  \big \vert  \leq \frac{1}{2\pi}  \Big \vert  \int_ { \mathcal{C} } \frac{   w_{0,0}(\vec{p},z' ,0,\vec{0}) +z'}{(z'-z)^2}  dz' \Big \vert  <  \frac{ \delta \mu }{   (\frac{\mu}{3}( 1 - 3 \rho/2 ) - \delta)^2} ,
\end{equation*}
where in the last inequality we used  that $H(\underline{w}(\cdot,\cdot)) \in  \mathcal{B}( \gamma , \delta , \varepsilon )$. For $\rho < 1/2$ and $\delta \ll \rho \mu$, the right-hand side is strictly smaller than $1$, which implies that $\partial_z  w_{0,0}(\vec{p} , z  ,0,\vec{0}) \neq 0$.
\vspace{2mm}

\noindent \textbf{Step 3}. We prove that $E_{\rho} ( \mathcal{U}[ w_{0,0} ]  ) = U[\vec{p}^*] \times D_{\mu/2} $ and that $E_{\rho}$ is injective. Let $( \vec{p}_0 , u _0 ) \in  U[\vec{p}^*] \times D_{\mu /2}$. For any $z$ with $\vert z  \vert = \mu / 2 - \beta$ and $0 < \beta \ll \mu$, we have that $| z - \rho u_0 | \ge ( 1 - \rho ) \mu / 2 - \beta > \delta$, provided that $\delta$ is small enough. This implies that
\begin{equation*}
\vert  (- w_{0,0}(\vec{p}_0,z,0,\vec{0}) - \rho  u_0) -(z- \rho u_0) \vert \leq  \delta <  \vert z  - \rho u_0 \vert.
\end{equation*}
By Rouch\'e's Theorem, the number of zeros of $z \mapsto - w_{0,0}(\vec{p}_0,z,0,\vec{0}) - \rho  u_0$ in $D_{\mu/2-\beta}$ is equal to the number of zeros of $z \mapsto z- \rho  u_0$, which is  equal to 1 if $\beta$ is small enough. Since $\beta > 0$ is arbitrary, there is one (and only one) $z_0$ in $D_{\mu/2}$ such that $-\rho^{-1} w_{0,0}(\vec{p}_0,z_0,0,\vec{0})=u_0$. In other words, $(\vec{p}_0 , u_0 ) = E_\rho ( \vec{p}_0 , z_0 )$.
\end{proof}

\subsubsection{The smooth Feshbach-Schur map} \label{smoothfunc} We choose  a  smooth function $\chi: \mathbb{R} \rightarrow [ 0 , 1 ]$ with compact support included in $[0,1]$, having the property that $\chi(x)=1$ for any $x \in [0,3/4]$.  For any $\rho>0$, we  introduce the smooth function $\chi_{\rho}:  \mathbb{R} \rightarrow [ 0 , 1 ]$, defined by the rescaling
\begin{equation}
\chi_{\rho }(r):=\chi\left( \frac{ r }{ \rho }\right).
\end{equation}
We set
\begin{equation*}
\overline{\chi}_{\rho}(r) := \sqrt{ 1 - \chi_{\rho}^2(r) }.
\end{equation*}
The functional calculus for self-adjoint operators allows us to construct the operators $\chi_{\rho}(H_f)$ and $\overline{\chi}_{\rho}(H_f)$ on the space $\mathcal{H}_f$. The next lemma shows that the pair $(H(\underline{w}(\vec{p},z)),  W_{0,0}( \underline{w}(\vec{p},z) ))$ is a Feshbach pair for $\chi_{\rho}(H_{f})$ for suitable values of the parameters $\rho, \xi, \gamma, \delta, \varepsilon$.
\begin{lemma}
\label{cbo}
Let $0 < \rho < 1/2$, $0 < \xi < 1/\sqrt{8 \pi}$, $0 < \gamma \ll \mu$, $\delta > 0$ and $0 < \varepsilon \ll \rho \mu$.  For all $H(\underline{w}(\cdot,\cdot)) \in \mathcal{B} (\gamma, \delta , \varepsilon)$ and  $(\vec{p},z) \in \mathcal{U} [ w_{0,0} ]$,
\begin{align*}
\big ( H(\underline{w}(\vec{p},z)),  W_{0,0}( \underline{w}(\vec{p},z) ) \big ),
\end{align*}
is a Feshbach pair for $\chi_\rho(H_f)$.
\end{lemma}
\begin{proof}
The proof consists in checking the conditions (a), (b) and (c) of Definition \ref{Fesh}. Condition (a)  is trivial. 
To prove that $W_{0,0}( \underline{w}(\vec{p}, z ) )$ is bounded invertible on $\mathrm{Ran} \, \overline{\chi}_\rho(H_f)$, we use the functional calculus of  Paragraph \ref{funcpar}.  It is sufficient to show that there exists a constant $C>0$  such that $\vert w_{0,0}(\vec{p},z, r , \vec{l} ) \vert \ge C  > 0$ for all $r \in [ \frac{3\rho}{4} , 1 ]$, $| \vec{l} | \le r$, and $(\vec{p},z) \in \mathcal{U}[ w_{0,0} ]$.  To shorten the length of the formulas, we set $t(\vec{p},z,r,\vec{l})=w_{0,0}(\vec{p},z, r , \vec{l} ) -w_{0,0}(\vec{p},z, 0, \vec{0} )$. One has 
\begin{align}
| w_{0,0}( \vec{p},z,r , \vec{l}  ) | & \ge r - | \frac{ \vec{p} }{ m } \cdot \vec{ l } | -Ê\big | t (\vec{p},z , r , \vec{l} )- \big ( r - \frac{\vec{p}}{ m } \cdot \vec{ l } \big ) \big | - |w_{0,0}(\vec{p},z, 0, \vec{0} ) |. \label{eq:aaa}
\end{align}
But
\begin{align*}
& t ( \vec{p},z, r , \vec{l} ) - \big ( r - \frac{ \vec{p}  }{ m } \cdot \vec{ l } \big ) \\
& = \sum_{j=1}^{3}  \int_{0}^{ l_j }   \left( (\partial_{l_j} t )(\vec{p},z,r,  \hat{l}_{j}(l')) + \frac{p_{j}}{m} \right) dl'  + \int_{0}^{r}   \left( (\partial_{r'} t )(\vec{p},z, r',   \vec{0}  ) -1\right) dr',
 \end{align*}
   with $ \hat{l}_{1}(l')=(l', 0 , 0)$,  $\hat{l}_{2}(l')=(l_1,l', 0)$,  $ \hat{l}_{3}(l')=(l_1 ,l_2,l' )$.  Since $H(\underline{w} ( \cdot , \cdot ))$ belongs to the polydisc $ \mathcal{B} ( \gamma  , \delta ,  \varepsilon )$, we deduce that
   $\| t( \vec{p},z , r , \vec{l} ) - ( r -  m^{-1}  \vec{p} \cdot \vec{l} ) \|^{\sharp} \le \gamma$,  and hence the previous relation implies
\begin{align*}
\big | t (\vec{p}, z , r , \vec{l} ) - \big ( r - \frac{ \vec{p} }{ m } \cdot \vec{ l } \big ) \big | \le r \gamma .
\end{align*}   
Since $(\vec{p},z) \in \mathcal{U}[ w_{0,0} ]$,   $| w_{0,0}(\vec{p},z,0,\vec{0}) | < \mu \rho / 2$, and since $\vert \vec{p}-\vec{p}^* \vert<\mu m$, where $\mu=( m-\vert \vec{p}^* \vert)/2m$, we have that $| \vec{p} | < ( 1 - \mu ) m$. Therefore \eqref{eq:aaa} gives
\begin{align}
| w_{0,0}( \vec{p},z ,r, \vec{l} ) | & \ge r \left( \mu - \gamma \right)- \frac{ \mu \rho}{2} . \label{eq:a10}
\end{align}
Since $r \geq 3 \rho/4$, we conclude that, for $\gamma \ll \mu$, $W_{0,0}( \underline{w}( \vec{p},z ))$ is bounded invertible on $\mathrm{Ran} \, \overline{\chi}_\rho(H_f)$ with an inverse bounded by $\mathcal{O} ( ( \mu \rho )^{-1} )$.

To prove that the operator $W_{0,0}( \underline{w}(\vec{p},z ) ) + \overline{\chi}_\rho (H_f) W_{ \ge 1 }( \underline{w}(\vec{p},z ) ) \overline{\chi}_\rho(H_f)$ is bounded invertible on $\mathrm{Ran} \, \overline{\chi}_\rho(H_f)$,  we use  Lemma \ref{opboundl} together with the fact that $H(\underline{w}) \in \mathcal{B} ( \gamma  , \delta ,  \varepsilon )$, which implies that $ \| W_{ \ge 1 }( \underline{w}(\vec{p},z ) )  \| \leq  \varepsilon$. For $ 0 < \varepsilon \ll \rho \mu$, it follows that
\begin{align}
&\big \| \big ( W_{0,0}( \underline{w}(\vec{p},z )  ) \big )^{-1} \overline{\chi}_\rho(H_f)  W_{ \ge 1 }( \underline{w}(\vec{p},z )  ) \big \| <1, \notag \\
& \big \| W_{ \ge 1 }( \underline{w}(\vec{p},z )  ) \big ( W_{0,0}( \underline{w}(\vec{p},z )  ) \big )^{-1} \overline{\chi}_\rho(H_f) \big \| <1, \label{eq:a1}
\end{align}
and hence  $W_{0,0}( \underline{w}(\vec{p},z ) ) + \overline{\chi}_\rho (H_f) W_{ \ge 1 }( \underline{w}(\vec{p},z ) ) \overline{\chi}_\rho(H_f)$ is bounded invertible on $\mathrm{Ran} \, \overline{\chi}_\rho(H_f)$. Condition (c) is a direct consequence of \eqref{eq:a1}.
\end{proof}

\subsubsection{Definition of the renormalization map}
The renormalization map is the composition of the scale transformation, $S_{\rho}$, the smooth Feshbach-Schur map, $F_{\chi_\rho(H_f)}$, and the inverse of the rescaling of the spectral parameter $E_{\rho}$.

\begin{definition}[The renormalization map $\mathcal{R}_\rho$]\label{def_Rrho}
Let $0 < \rho < 1/2$, $0 < \xi < 1/(4\sqrt{8 \pi})$, $0 < \delta ,\varepsilon \ll \rho \mu$ and $0 < \gamma \ll \mu $. The renormalization transformation $\mathcal{R}_\rho : \mathcal{B} ( \gamma,\delta,\varepsilon) \to  \mathcal{W}_{\mathrm{op}}$ is defined by
\begin{equation}\label{eq:def_Rrho}
\mathcal{R}_\rho (H (\underline{w}) )(\vec{p}, \zeta ) := S_\rho \left( F_{ \chi_\rho(H_f) } \left( H \left( \underline{w} \left( E_\rho^{-1}( \vec{p},\zeta )\right)  \right) , W_{0,0} \left( \underline{w} \left(E_\rho^{-1}(\vec{p}, \zeta ) \right) \right) \right) \right) ,
\end{equation}
for  any $( \vec{p} , \zeta ) \in U[ \vec{p}^* ] \times D_{ \mu / 2 }$ and $H(\underline{w} ( \cdot , \cdot )) \in \mathcal{B} ( \gamma,\delta,\varepsilon)$ such that the map $( \vec{p} , z ) \mapsto w_{0,0}( \vec{p} , z ,  0 , \vec{0} )$ is analytic on $U[ \vec{p}^* ]\times D_{ \mu / 2 }$. We set 
\begin{equation}\label{eq:hatw}
H(\underline{\hat{w}}(\vec{p}, \zeta )):=\mathcal{R}_\rho (H (\underline{w}) )(\vec{p}, \zeta ).
\end{equation}
\end{definition}

\begin{remark}
$\quad$
\begin{itemize}
\item The rescaling of the spectral parameter $\zeta$, given by $E_\rho^{-1}( \vec{p},\zeta )$ in Definition \ref{def_Rrho}, allows us to control the expansion of the relevant part along the iteration of the renormalization map. More precisely, if we set 
\begin{equation*}
\tilde{\mathcal{R}}_\rho (H (\underline{w}) )(\vec{p}, z ) := S_\rho \left( F_{ \chi_\rho(H_f) } \left( H \left( \underline{w} ( \vec{p},z )  \right) , W_{0,0} \left( \underline{w} (\vec{p}, z )   \right) \right) \right)
\end{equation*}
and assume that $z$ is sufficiently small to iterate the Feshbach-Schur map and the scale transformation, then the absolute value of the relevant part increases, and after a large number of iterations, it is not possible anymore to apply the Feshbach-Schur map. This phenomenon is illustrated in Figure 2.1; The marginal part coincides with the fixed point manifold $\mathcal{M}_{fp}$, the irrelevant part coincides with the stable manifold $\mathcal{M}_s$, and the relevant part  with the unstable manifold $\mathcal{M}_u$.

\begin{center}
 \begin{tikzpicture}
\label{fig00}

         \path[fill=black!15] 
             (0,0)--(6,2) to[bend left=45] (8,0.5) -- (2, -2.5)  to[bend left=-45] (0,0)  ;

      \draw[-, very thick] (4,1.3) -- (4.0,4.5);
       \draw[-,dotted,very thick] (4,0) -- (4.0,1.3);
        \draw[-, very thick] (0,0) -- (6.5,2.18);

       \draw (3,-1.2) node[below] {$\mathcal{M}_s$};
      \draw (4,3) node[right] {$\mathcal{M}_u$};
      \draw (1,0.5) node[above] {$\mathcal{M}_{fp}$};  
     \fill[black] (7,1) circle (0.5mm);
      \fill[black] (5.5,2.4) circle (0.5mm);
        \draw[-,dashed,thick](7,1)--(7,0.5);
         \draw[-,dashed,thick](5.5,2.4)--(5.5,1.5);
          \draw[-, thick](7,1)  to[bend left=30]  (5.5,2.4);
            \draw[->, thick](5.5,2.4) to[bend left=10] (5.4,4);
             \draw (7,1) node[right] {\small $H (\vec{p})- z\mathds{1}$};
              \draw (5.5,2.4) node[right] {\small $\tilde{\mathcal{R}}^{N}_\rho (H (\vec{p})- z\mathds{1}) $};   
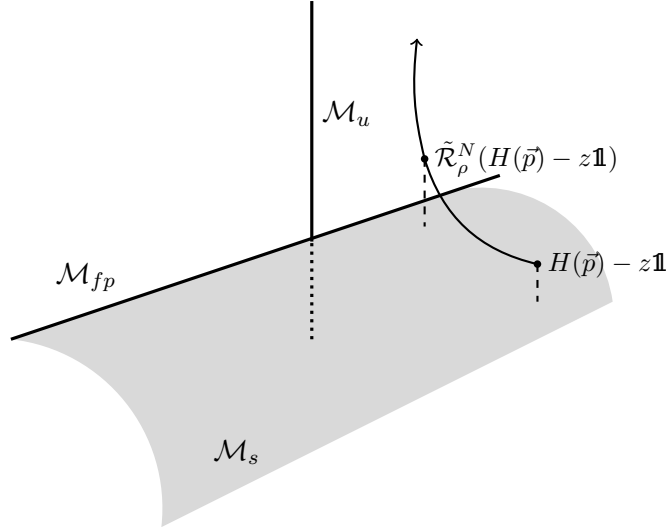
                \draw (5,-2.5) node[below] { \small Figure 2.1. The unstable manifold blows up along the iteration of the renormalization map};
                 \draw (1.2,-2.9) node[below] { \small if the spectral parameter is not rescaled. };

        \end{tikzpicture}

\end{center}

\item We anticipate that the operator-valued function $\mathcal{R}_\rho (H (\underline{w}) )( \cdot,\cdot )$ belongs to $ \mathcal{W}_{\mathrm{op}}$, i.e., can be written as a series of Wick monomials. The proof is sketched in the next subsection. In particular, the sequence of kernels $\underline{\hat{w}}$ defined in \eqref{eq:hatw} is uniquely determined by the procedure we use to write $\mathcal{R}_\rho (H (\underline{w}) )( \cdot,\cdot )$ as a sequence of Wick monomials.
\item The condition that  $(\vec{p},z) \mapsto w_{0,0}(\vec{p},z,0,\vec{0})$ is  analytic on $U[\vec{p}^*] \times D_{\mu /2}$ can be weakened. The proof of Lemma   \ref{lm:bihol} (step 3) indeed shows that the condition  that  $z \mapsto w_{0,0}(\vec{p},z,0,\vec{0})$ is analytic on  $ D_{\mu/2}$ for all $\vec{p} \in U[\vec{p}^*]$  is sufficient to ensure existence of the inverse of the scale transformation $E_{\rho}$.
\end{itemize}
\end{remark}

\subsection{Wick ordering} \label{Wicki} We sketch the proof of the claim that $\mathcal{R}_\rho : \mathcal{B} ( \gamma,\delta,\varepsilon) \to  \mathcal{W}_{\text{op}}$. In fact, if $ \gamma,\delta,\varepsilon$ are ``small enough'', a   stronger results holds: $\mathcal{R}_{\rho}$ is codimension-4 contractive and maps  $ \mathcal{B} ( \gamma,\delta,\varepsilon)$ to $ \mathcal{B} ( \gamma+\varepsilon/2,\varepsilon/2,\varepsilon/2)$. The discussion of this essential property is done in Subsection \ref{itere} and the  exact expression of the new operator $H(\underline{\hat{w}})$ is given in Appendix \ref{Wicko}.  In this paragraph, we focus on the reason why $\mathcal{R}_\rho(H(\underline{w}))$ can be rewritten as a series of Wick monomials. We omit the arguments $(\vec{p},z)$ to shorten the length of the formulas. We have that
\begin{align}
 &F_{ \chi_\rho(H_f) } \left( H \left( \underline{w} \right)  , W_{0,0} \left( \underline{w}  \right)  \right)=  W_{0,0} \left( \underline{w}  \right)  + \chi_{\rho}(H_f)  W_{\geq 1} \left( \underline{w}    \right)  \chi_{\rho}(H_f) \notag \\
 & - \chi_{\rho}(H_f)  W_{\geq 1} \left( \underline{w}   \right)  \overline{\chi}_{\rho}(H_f)  \big ( W_{0,0} \left( \underline{w}  \right)+ \overline{\chi}_{\rho}(H_f) W_{\geq 1} \left( \underline{w}   \right)  \overline{ \chi}_{\rho}(H_f) \big )^{-1} \notag \\
 & \quad \overline{\chi}_{\rho}(H_f) W_{\geq 1} \left( \underline{w}  \right)   \chi_{\rho}(H_f). \label{Wickl}
\end{align}
The two terms on the first line of  the right-hand side of \eqref{Wickl} are already written with Wick monomials. To rewrite  the second line of \eqref{Wickl} with Wick monomials,  we need to expand the inverse of $W_{0,0} \left( \underline{w}  \right)+ \overline{\chi}_{\rho}(H_f) W_{\geq 1} \left( \underline{w}   \right)  \overline{ \chi}_{\rho}(H_f)$ in a convergent Neumann series, and to normal order the product of annihilation and creation operators of  each term of this series. The Neumann series expansion of the last term in \eqref{Wickl} reads
\begin{equation}
\label{sumL}
\sum_{L=2}^{\infty} (-1)^{L-1}  \chi_{\rho}(H_f)  W_{\geq 1} \left( \underline{w}   \right)   \frac{\overline{\chi}_{\rho}(H_f) }{W_{0,0} \left( \underline{w}  \right)}    \left( \overline{\chi}_{\rho}(H_f) W_{\geq 1} \left( \underline{w}   \right)  \frac{ \overline{ \chi}_{\rho}(H_f)}{W_{0,0} \left( \underline{w}  \right)} \right)^{L-2} \overline{\chi}_{\rho}(H_f) W_{\geq 1} \left( \underline{w}  \right)   \chi_{\rho}(H_f).
\end{equation}
 Let $L \in \mathbb{N}$, $L \ge 2$, and consider the $L^{th}$  term in the sum \eqref{sumL}. Each operator $W_{\geq 1} \left( \underline{w}   \right)$ is a series of Wick monomials; We label   the  $L $  operators  $W_{\geq 1}\left( \underline{w}   \right)$ by an index $i$, $i=1,...,L$.  For each   $W_{\geq 1,i}\left( \underline{w}   \right)$,  we  choose one term $W_{M_i,N_i}(\underline{ w})$ in the series defining  $W_{\geq 1,i}\left( \underline{w}   \right)$ and we  normal order the annihilation and creation operators appearing in the product
 \begin{align}
 \chi_{\rho}(H_f)  W_{M_1,N_1} \left( \underline{w}   \right)   \frac{\overline{\chi}_{\rho}(H_f) }{W_{0,0}  \left( \underline{w}  \right)} \prod_{i=2}^{L-1}  & \left( \overline{\chi}_{\rho}(H_f) W_{M_i,N_i} \left( \underline{w}   \right)  \frac{ \overline{ \chi}_{\rho}(H_f)}{W_{0,0} \left( \underline{w}  \right)} \right) \notag \\
& \quad  \overline{\chi}_{\rho}(H_f) W_{M_{L},N_{L}} \left( \underline{w}  \right)   \chi_{\rho}(H_f).  \label{prewick}
\end{align}
Normal ordering is carried out  with the help of Wick's Theorem and the pull-through formula (whose proofs are standards; see \cite{BaFrSi98_01}).
\begin{lemma}[Wick ordering]
\label{Wo}
 Set $\mathbb{N}_p:=\lbrace 1,...,p \rbrace$,    $b^{+}(\underline{k}_j):=b^{*}(\underline{k}_j)$,   $b^{-}(\underline{k}_j):=b(\underline{k}_j)$, and let $\sigma_i= \pm$, $i=1,...,p$. For any ordered finite subset $\mathcal{P}=\lbrace i_1,...,i_p\rbrace$ of $\mathbb{N}$, $i_1<...<i_p$, we set 
\begin{equation*}
\prod_{i \in \mathcal{P}} b^{\sigma_i} (\underline{k}_i)= b^{\sigma_{i_1}} (\underline{k}_{i_1})... b^{\sigma_{i_p}} (\underline{k}_{i_p}).
\end{equation*}
Then, 
\begin{equation}
\label{wiwi}
\prod_{i=1}^{p} b^{\sigma_i}(\underline{k}_i)= \sum_{\mathcal{P} \subset  \mathbb{N}_p}\langle \Omega \vert  \prod_{i \in \mathbb{N}_p \setminus \mathcal{P} } b^{\sigma_i}(\underline{k}_i)  \Omega \rangle \text{ } :\prod_{j \in \mathcal{P}} b^{\sigma_j}(\underline{k}_j) : ,
\end{equation}
where  $: \cdot :$ denotes the Wick-ordered product of creation and annihilation operators:
\begin{equation*}
 :\prod_{j \in \mathcal{P}} b^{\sigma_j}(\underline{k}_j) :=\prod_{j \in \mathcal{P}, \sigma_j=+} b^{\sigma_j}(\underline{k}_j)  \prod_{l \in \mathcal{P}, \sigma_l=-} b^{\sigma_l}(\underline{k}_l).
\end{equation*}
\end{lemma} 
\begin{lemma}[Pull-through formula]
\label{Pt}
Let $f:\mathbb{R}\times \mathbb{R}^3 \to \mathbb{C}$ be a  measurable function. Then, for a.e. $\vec{k} \in \mathbb{R}^3$,
\begin{align}
b(\vec{k}) f(H_f,\vec{P}_f)&=f(H_f + \vert \vec{k} \vert ,\vec{P}_f +\vec{k})  b(\vec{k}),\\
 f(H_f,\vec{P}_f) b^*(\vec{k})&=  b^{*}(\vec{k}) f(H_f + \vert \vec{k} \vert ,\vec{P}_f +\vec{k}). 
\end{align}
\end{lemma} 
To explain how to use Lemmas \ref{Wo} and \ref{Pt}, we illustrate  Formula \eqref{prewick} with some pictures. We represent  the operators $w_{M_i,N_i}(H_f,\vec{P_f},\underline{k}_1,...,\underline{k}_{M_i},\underline{\tilde{k}}_1,...,\underline{\tilde{k}}_{N_i})$ by   circular domains and the creation/annihilation operators by plain lines. We set 
\begin{equation}
F:=\frac{ \overline{ \chi}^{2}_{\rho}(H_f)}{W_{0,0} \left( \underline{w}  \right)}.
\end{equation}
The operator \eqref{prewick} is an  integral over the momenta, $\underline{k}_{1}^{(M_1)},...,\underline{k}_{L}^{(M_L)},\underline{\tilde{k}}_{1}^{(N_1)},...,\underline{\tilde{k}}_{L}^{(N_L)}$, of the operators represented by the drawing

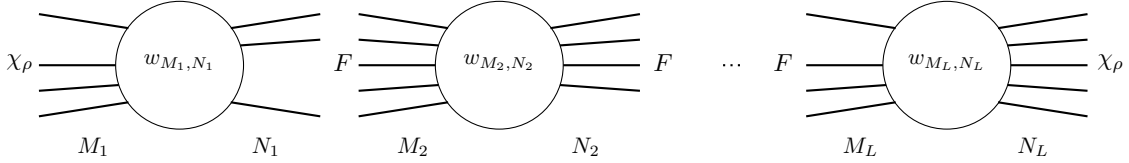
\begin{center}
\begin{tikzpicture}[scale=0.85, every node/.style={transform shape}]
 
   \pgfmathsetmacro{\a}{5}
    \pgfmathsetmacro{\b}{12}

   \draw (0,0) circle (1) ;
   \draw (\a,0) circle (1) ;
   \draw (\b,0) circle (1) ;
         
    \draw[-,thick ] (0.8,0.58)--(2.2,0.8); 
     \draw[-,thick ] (0.95,0.31)--(2.2,0.4); 
      \draw[-,thick ] (0.8,-0.58)--(2.2,-0.8); 
     
      \draw[-,thick] (-1,0)--(-2.2,0);  
    \draw[-,thick] (-0.8,0.58)--(-2.2,0.8); 
      \draw[-,thick ] (-0.8,-0.58)--(-2.2,-0.8); 
      \draw[-,thick ] (-0.95,-0.31)--(-2.2,-0.4);

       \draw[-,thick ](1+\a,0)--(2.2+\a,0);  
    \draw[-,thick ](0.8+\a,0.58)--(2.2+\a,0.8); 
     \draw[-,thick ] (0.95+\a,0.31)--(2.2+\a,0.4); 
     \draw[-,thick ] (0.95+\a,-0.31)--(2.2+\a,-0.4); 
     
      \draw[-,thick] (-1+\a,0)--(-2.2+\a,0);  
    \draw[-,thick] (-0.8+\a,0.58)--(-2.2+\a,0.8); 
     \draw[-,thick ] (-0.95+\a,0.31)--(-2.2+\a,0.4); 
      \draw[-,thick ] (-0.8+\a,-0.58)--(-2.2+\a,-0.8); 
     \draw[-,thick ] (-0.95+\a,-0.31)--(-2.2+\a,-0.4);

      \draw[-,thick, ] (1+\b,0)--(2.2+\b,0);  
    \draw[-,thick ] (0.8+\b,0.58)--(2.2+\b,0.8); 
     \draw[-,thick ] (0.95+\b,0.31)--(2.2+\b,0.4); 
      \draw[-,thick ] (0.8+\b,-0.58)--(2.2+\b,-0.8); 
     \draw[-,thick ] (0.95+\b,-0.31)--(2.2+\b,-0.4); 
     
      \draw[-,thick] (-1+\b,0)--(-2.2+\b,0);  
    \draw[-,thick] (-0.8+\b,0.58)--(-2.2+\b,0.8); 
      \draw[-,thick ] (-0.8+\b,-0.58)--(-2.2+\b,-0.8); 
     \draw[-,thick ] (-0.95+\b,-0.31)--(-2.2+\b,-0.4);

    \draw (-1.35,-1) node[below] {\small$M_1$};
    \draw (1.35,-1) node[below] {\small$N_1$};
     \draw (-1.35+\a,-1) node[below] {\small$M_2$};
    \draw (1.35+\a,-1) node[below] {\small$N_2$};
     \draw (-1.35+\b,-1) node[below] {\small$M_L$};
    \draw (1.35+\b,-1) node[below] {\small$N_L$};

 \draw (0,0.3) node[below] {\small$w_{M_1,N_1}$};
  \draw (\a,0.3) node[below] {\small$w_{M_2,N_2}$};
   \draw (\b,0.3) node[below] {\small$w_{M_L,N_L}$};

     \draw (2.56,0.3) node[below] {$F$};
      \draw (2.56+\a,0.3) node[below] {$F$};
      \draw (-2.56+\b,0.3) node[below] {$F$};
     \draw (-2.5,0.3) node[below] {$\chi_{\rho}$};
     \draw ( 2.56+\b,0.3) node[below] {$\chi_{\rho}$};

 \draw[-,thick,dotted] (8.5,0) -- (8.8,0);


\end{tikzpicture}

\end{center}

\noindent In the picture above, we have set $$\chi_{\rho}:=\chi_{\rho}(H_f)$$ and $$w_{M_i,N_i}:=w_{M_i,N_i}(H_f,\vec{P_f},\underline{k}_1,...,\underline{k}_{M_i},\underline{\tilde{k}}_1,...,\underline{\tilde{k}}_{N_i})$$ to shorten notation.  Plain lines starting from the left hand-side of an operator $w_{M_i,N_i}$ represent creation operators, whereas  plain lines starting from its right-hand side  represent annihilation operators.  We  use  Formula $\eqref{wiwi}$ to normal order the product of  creation and annihilation operators that appears in \eqref{prewick}. This amounts to picking out $m_1/n_1$ creation/annihilation operators over  the $M_1/N_1$ creation/annihilation  operators available from $W_{M_1,N_1}$, .... , $m_L/n_L$ creation/annihilation operators over  the $M_L/N_L$  creation/annihilation operators available from $W_{M_L,N_L}$, and  to contracting all the remaining creation/annihilation operators that appear in \eqref{prewick}.  As the kernels $w_{M_i,N_i}$ are symmetric in $\underline{k}_1,...,\underline{k}_{M_i}$ and  $\underline{\tilde{k}}_1,...,\underline{\tilde{k}}_{N_i}$, the contractions with the same numbers  $m_1,...,m_L/n_1,...,n_L$ give rise to the same contribution. There are  therefore 
$$C^{\underline{M},\underline{N}}_{\underline{m}, \underline{n}}:=\prod_{i=1}^{L} \left( \begin{array}{c} M_i\\m_i \end{array} \right)  \left( \begin{array}{c} N_i\\n_i \end{array} \right)$$ contraction schemes giving rise to the same  operator, where we have set $\underline{M}:=(M_1,...,M_L)$.  Finally, we pull-through the uncontracted $m_1+....+m_L$ creation operators to the left of the operator $\chi_{\rho}(H_f)$ located on the left-hand side of \eqref{prewick}, and the $n_1+....+n_L$ uncontracted annihilation operators to the right of the operator $\chi_{\rho}(H_f)$  located on the right-hand side of \eqref{prewick}.  This causes a shift in the kernel arguments via Formula \eqref{Pt}.  Formula  \eqref{wiwi}  tells us  that the contracted part can be rewritten as the vacuum expectation value of an operator $\mathscr{V}^{\underline{M},\underline{N}}_{\underline{m}, \underline{n}}(\vec{p},z,r,\vec{l},\underline{K}^{(\underline{m},\underline{n})})$; See Appendix \ref{Wicko} for more details.  The Wick-ordered operator corresponding to \eqref{prewick} is therefore a sum over $m_1,...,m_L,n_1,...,n_L$ (with $ 0 \leq m_i \leq M_i$, $ 0 \leq n_i \leq N_i$),  of  operators of the form
\begin{equation}
\label{couscous}
 C^{\underline{M},\underline{N}}_{\underline{m}, \underline{n}} \int d\underline{K}^{(\underline{m},\underline{n})}   b^{*}(\underline{k}^{(\underline{m})}) \text{ }  \langle \Omega \vert    \mathscr{V}^{\underline{M},\underline{N}}_{\underline{m}, \underline{n}} (\vec{p},z,r,\vec{l},\underline{K}^{(\underline{m},\underline{n})})  \Omega  \rangle_{r=H_f, \vec{l}=\vec{P}_f}  \text{ }b(\underline{\tilde{k}}^{(\underline{n})}),
 \end{equation}
 where the operator $ \mathscr{V} ^{\underline{M},\underline{N}}_{\underline{m}, \underline{n}} (\vec{p},z,r,\vec{l},\underline{K}^{(\underline{m},\underline{n})})$ sitting within the expectation value depends on  the operator-valued functions $F$ and $w_{M_i,N_i}$, $i=1,...,L$. The operator in \eqref{couscous} is Wick-ordered. Applying this procedure for each $(M_1,N_1),...,(M_L,N_L)$, and for each $L \in \mathbb{N}$, $L\geq 2$,  \eqref{Wickl} can be rewritten as  a series of Wick monomials. The operator norm convergence of this series and the fact that the kernel $\underline{\hat{w}}$ belongs to $\mathcal{W}$ requires a little more effort to be established. We refer the reader to \cite{BaChFrSi03_01} or \cite{HaHe11_01} for  detailed proofs.
\indent

\section{The first decimation step}
\label{firstep}
Let $0 < \gamma,\delta, \varepsilon< 1$.  Constraining the coupling constant $\lambda_0$, any  polydisc $\mathcal{B}(\gamma,\delta, \varepsilon)$  (see Definition $\ref{discc}$) can be reached from the initial Hamiltonian 
 \begin{equation}
\label{Hp3}
H(\vec{p},z): = \frac{\vec{P}_f^2}{2m} -\frac{ \vec{p}}{m} \cdot \vec{P_f}+  \omega_0 \left(\begin{array}{cc} 1 & 0 \\ 0& 0 \end{array} \right) \notag +\lambda_0  H_I + H_{f}- z \mathds{1},
\end{equation}
by applying to $H(\vec{p},z)$ some isospectral transformations. The aim of this section is to highlight how this can be achieved.
\subsection{The first step: Applying the smooth Feshbach-Schur map to $H(\vec{p},z)$}
\subsubsection{Some notations}
Let $(\vert  \uparrow \rangle, \vert  \downarrow \rangle)$ be the orthonormal  basis of  $\mathbb{C}^2$  in which  
\begin{equation*}
 \left(\begin{array}{cc} \omega_0 & 0 \\ 0& 0 \end{array} \right)= \omega _0 P_{\uparrow}+ 0 \cdot P_{\downarrow}, 
\end{equation*}
where we  introduced the orthogonal projections  $ P_{\uparrow}:= \vert  \uparrow \rangle \langle  \uparrow \vert$ and $ P_{\downarrow}:= \vert  \downarrow\rangle \langle \downarrow \vert$. Let $\chi$ and $\overline{\chi}$ be the smooth functions introduced in  Paragraph \ref{smoothfunc}. Let $\rho_0>0$. We  define two new  bounded operators on $\mathbb{C}^2 \otimes \mathcal{F}_{+}(L^{2}( \underline{ \mathbb{R}}^3 ) )$ by
\begin{equation}
\bm{\chi} := P_{\downarrow} \otimes  \chi_{\rho_{0}}(H_f),   \qquad   \qquad  \overline{\bm{\chi}} := P_{\downarrow} \otimes \overline{\chi}_{\rho_0}(H_f) +  P_{\uparrow}  \otimes  \mathds{1}. 
\end{equation}
The reader can  check that $\bm{\chi}^2 + \overline{\bm{\chi}} ^{2}= \mathds{1}$. We   denote by $H_{0}(\vec{p},z)$ the operator
\begin{equation}
H_{0}(\vec{p},z):=\frac{\vec{P}_f^2}{2m} -\frac{ \vec{p}}{m} \cdot \vec{P_f}+  \omega_0 \left(\begin{array}{cc} 1 & 0 \\ 0& 0 \end{array} \right) \notag  + H_{f}- z \mathds{1}.
\end{equation}

\subsubsection{ The Feshbach pair}
\begin{lemma}
\label{first}
Let $\rho_0<\omega_0$. There exists $\lambda_c > 0$ such that, for all $0\leq \lambda_0 < \lambda_c$, for all  $\vert z \vert < \frac{ \mu \rho_{0}}{2}$,   $(H(\vec{p},z),H_{0}(\vec{p},z) )$ is a Feshbach pair for $\bm{\chi}$.
\end{lemma}
\begin{proof}
 The proof is standard, but we show that Condition (b) of Definition \ref{Fesh} is satisfied, in order to illustrate how the functional calculus outlined in  Paragraph \ref{funcpar} can be used to control the norm of bounded operators.  The restriction of $H_{0}(\vec{p},z)$ to  $\text{Ran}(\overline{\bm{\chi}})$
 can be decomposed into a sum of two operators,
 \begin{equation}
H_{0}(\vec{p},z)_{\mid\text{Ran}(\overline{\bm{\chi}})}= P_{\downarrow} \otimes b_{1}(\vec{p},z,H_f,\vec{P}_f)+  P_{\uparrow} \otimes b_{2}(\vec{p},z,H_f,\vec{P}_f),
 \end{equation}
where 
\begin{align}
b_{1}(\vec{p},z , r,\vec{l})&:= \left(r+\frac{\vec{l}^2}{2m} -  \frac{1}{m}\vec{p} \cdot  \vec{l}   -z \right) \mathds{1}_{r \ge \frac{3 \rho_{0}}{4}} (r),\\
b_{2}(\vec{p},z , r,\vec{l})&:=r+\frac{\vec{l}^2}{2m} -   \frac{1}{m}\vec{p} \cdot  \vec{l} +   \omega_0 -z,
\end{align}
with $b_{1}(\vec{p},z,H_f,\vec{P}_f)$ and $b_{2}(\vec{p},z,H_f,\vec{P}_f)$  defined by the functional calculus of  Paragraph \ref{funcpar}. Using the bound (\ref{essbo}),   $b_{1}(\vec{p},z,H_f,\vec{P}_f)$  is bounded invertible, if, for any $\vert z \vert < \frac{ \mu \rho_{0}}{2}$, 
 \begin{equation*}
 \underset{r \geq \frac{3 \rho_{0}}{4}, \vert \vec{l} \vert \leq r}{\sup} \left(  \frac{1}{\vert r +\frac{\vec{l}^2}{2m} - \frac{1}{m}\vec{p} \cdot \vec{l}-z  \vert}\right)< \infty.
 \end{equation*} 
This is indeed the case, as
 \begin{align*}
 \Re \left( r +\frac{\vec{l}^2}{2m} - \frac{1}{m} \vec{p}\cdot \vec{l}-z \right)  &\geq r\left( 1 -  \frac{1}{m }\vert  \vec{p} \vert  \right) - \vert \Re(z)\vert >   \frac{ \mu \rho_{0} }{4},
 \end{align*}
and we see that the inverse of $b_{1}(\vec{p},z,H_f,\vec{P}_f)$ is bounded by a constant of order $(\mu \rho_0)^{-1}$. Similarly, $b_{2}(\vec{p},z,H_f,\vec{P}_f)$ is bounded invertible and its inverse is bounded by a constant of order $ (\omega_0 - \mu \rho_0/2)^{-1}$, which implies that the restriction of $H_{0}(\vec{p},z)$ to  $\text{Ran}(\overline{\bm{\chi}})$ is bounded invertible. The bounded invertibility of the restriction of the operator
\begin{equation*}
H_{\overline{\bm{\chi}}}(\vec{p},z):= H_0(\vec{p},z) + \lambda_0 \overline{\bm \chi} H_I   \overline{\bm \chi}
\end{equation*}
to $\text{Ran}(\overline{\bm{\chi}})$ is shown with a Neumann expansion and the standard estimate
\begin{equation}
\label{2.22}
\big \| (H_f+ \rho_{0})^{-\frac12} H_I (H_f+ \rho_{0})^{-\frac12} \big \| \leq C \rho_{0}^{-\frac12}.
\end{equation}
The Neumann expansion for $H^{-1}_{\overline{\bm{\chi}}}(\vec{p},z)$  is formally equal to 
\begin{equation}
\label{Neuman}
\left(  H_0(\vec{p},z)  \right)_{\text{Ran}(\overline{\bm{\chi}})}^{-1} \sum_{n=0}^{\infty} (-\lambda_0)^{n}  \left[\overline{ \bm {\chi}}  H_I  \overline{ \bm {\chi}}  \left(  H_0(\vec{p},z)   \right)_{\text{Ran}(\overline{\bm{\chi}})}^{-1}\right]^{n}.
\end{equation}
Introducing the operator $(H_f+ \rho_{0})^{-\frac12}$ on the left and on the right of each operator $H_I$, we obtain  the bound 
\begin{equation}
\label{37}
 \|  \left[ H_{\overline{\bm{\chi}}}(\vec{p},z) \right]^{-1}_{ \text{Ran}(\overline{\bm{\chi}})}   \|   \le \, \rho_{0}^{-1}    \sum_{n=0}^{\infty} \left(C  \lambda_{0} \rho_{0}^{-\frac{1}{2}}  \right)^{n} \Big \| \left[\frac{H_f + \rho_{0}}{  H_0(\vec{p},z) } \right]_{  \text{Ran}(\overline{\bm{\chi}})} \Big \|^{n+1}  .
\end{equation}
Again, the functional calculus of section \ref{out} shows that $\|  [(H_f + \rho_{0})(  H_0(\vec{p},z) )^{-1}]_{  \text{Ran}(\overline{\bm{\chi}})}\| $ is bounded by 
\begin{equation*}
\underset{r \geq 3 \rho_{0}/4, \vert \vec{l} \vert \leq r }{\sup} \frac{ r+ \rho_{0}}{ \vert r + \frac{\vec{l}^2}{2m} - \frac{1}{m}\vec{p} \cdot \vec{l} -z\vert } +  \underset{r \geq0, \vert \vec{l} \vert \leq r}{\sup} \frac{ r+ \rho_{0}}{ \vert r + \frac{\vec{l}^2}{2m} - \frac{1}{m} \vec{p}  \cdot \vec{l}  + \omega_0-z \vert }=\mathcal{O}(\mu^{-1}).
\end{equation*}  
Therefore, there exists a  positive constant $\lambda_c$ of order $ \mu \rho_{0}^{\frac12} $, such that, for any $0 \leq \lambda_0<\lambda_c$,  the Neumann series in (\ref{Neuman}) converges in norm.  
   \end{proof}
  
  \subsection{The second step: Reaching a  polydisc $\mathcal{B}(\gamma,\delta,\varepsilon)$} 
 As the Feshbach-Schur map is isospectral, the study of the invertibility of the operator $F_{\bm{\chi}}(H(\vec{p},z),H_{0}(\vec{p},z))$ restricted on $P_{\downarrow} \mathbb{C}^2 \otimes \mathds{1}_{H_f \leq \rho_0}\mathcal{H}_f \supset \text{Ran}(\bm{\chi}) $ should provide   important insight about the spectrum of $H(\vec{p})$ near the origin of the complex plane. The restriction of  $F_{\bm{\chi}}(H(\vec{p},z),H_{0}(\vec{p},z))$ to $P_{\downarrow} \mathbb{C}^2 \otimes \mathds{1}_{H_f \leq \rho_0}\mathcal{H}_f$ is equal to
  \begin{equation}
  \label{expre}
P_{\downarrow} \otimes  \Big \langle \big( H_f + \frac{\vec{P}_{f}^{2}}{2m} - \frac{\vec{p}}{m}   \cdot  \vec{P}_f +   \lambda_0 \bm {\chi}   H_I  \bm {\chi}  - \lambda_{0}^{2}  \bm {\chi}  H_I  \overline{\bm {\chi}}   (H_{ \overline{\bm {\chi} }}(\vec{p},z))^{-1} \overline{\bm {\chi}}  H_I  \bm {\chi}  -z \mathds{1} \big) \mathds{1}_{H_f \leq \rho_0} \Big \rangle_{\downarrow},
  \end{equation}
 where $\langle A \rangle_{\downarrow} \in \mathcal{B}( \mathcal{H}_f)$ is the bounded operator associated to the quadratic form  
 \begin{equation*}
\langle \psi  \vert  \langle A \rangle_{\downarrow}    \phi \rangle := \langle \downarrow  \otimes \psi \vert  A (\downarrow  \otimes  \phi) \rangle,  \qquad \forall \psi,\phi \in \mathcal{H}_f, 
 \end{equation*}
  for any operator $A$ in $\mathcal{B}(\mathbb{C}^2 \otimes \mathcal{H}_f)$. If we apply the  scale transformation $S_{\rho_0}$ defined in Subsection \ref{rescaling} to (\ref{expre}) and Wick-order the resulting operator following the procedure sketched in Section \ref{Wicki},  we conclude that the next result holds.
 \begin{lemma}
\label{polyd}
Let $0 < \rho_0 <\omega_0$, $0 < \gamma , \delta , \varepsilon < 1$, and $\rho_0^{3/2} < \xi < 1/\sqrt{8 \pi}$. There exists $\lambda_c > 0$ such that, for all $0 \leq  \lambda_0  < \lambda_c$, for all   $\vert z \vert < \frac{ \mu}{2}$,
\begin{align*}
& \mathcal{S}_{\rho_{0}} \big( [ F_{\bm{\chi}}(H(\vec{p}, \rho_{0} z),H_{0}(\vec{p},  \rho_{0} z)) ]_{P_{\downarrow} \mathbb{C}^2 \otimes \mathds{1}_{H_f \leq \rho_0}\mathcal{H}_f} \big) =: P_{\downarrow} \otimes H( \underline{ w}^{(0)}(\vec{p}, z )), \\
& \text{with} \quad H(\underline{ w}^{(0)}( \cdot ,\cdot)) \in \mathcal{B}(\frac{ \sqrt{3} \rho_{0}}{m}+\gamma,\delta,\varepsilon).
\end{align*} 
\end{lemma} 
\noindent The proof is given in Appendix \ref{polydlem}. In particular, we  show there that the critical value $\lambda_c$ can be chosen such that $\lambda_c \le  C  \min ( \gamma , \delta,\varepsilon) \mu^2 \rho_0^{1/2}$ for some positive constant $C$ independent of the problem parameters.
\vspace{1mm}

 \subsection{Analyticity of the effective Hamiltonian $H(w^{(0)}(\vec{p},z))$  in $\vec{p}$ and $z$}
\begin{lemma}
\label{anainit}
If the assumptions of Lemma \ref{polyd} are satisfied, then $( \vec{p} , z ) \mapsto H( \underline{ w}^{(0)}(\vec{p},   z ))$ defined in Lemma \ref{polyd} is analytic on $U[\vec{p}^*] \times D_{\mu/2}$.
\end{lemma}
\begin{proof}
The proof  follows from $\eqref{expre}$ and the Neumann expansion for  $ \left[ H_{\overline{\bm{\chi}}}(\vec{p},z) \right]^{-1}_{ \text{Ran}(\overline{\bm{\chi}})} $  in the proof of Lemma \ref{first}. Since $( \vec{p} , z ) \mapsto (H_f + \frac{\vec{P}_{f}^{2}}{2m} - \frac{\vec{p}}{m}   \cdot  \vec{P}_f   - \rho_0 z \mathds{1}) \mathds{1}_{H_f \leq 1} \in \mathcal{B}( \mathcal{H}_{ \mathrm{red} } )$ is analytic  on $U[\vec{p}^*] \times D_{\mu/2}$, and since  $H_I$ is independent of $(\vec{p},z)$, it remains to show that 
\begin{equation}
\bm {\chi}  H_I  \overline{\bm {\chi}}     \left[ H_{\overline{\bm{\chi}}}(\vec{p}, \rho_0 z) \right]^{-1}_{ \text{Ran}(\overline{\bm{\chi}})}   \overline{\bm {\chi}}  H_I  \bm {\chi} \label{eq:aa1}
\end{equation}
is analytic in $\vec{p}$ and $z$.  Introducing $(H_f+ \rho_0)^{-1/2}$ on the left and on the right of $H_I$ in \eqref{Neuman}, we have seen  in the proof of Lemma \ref{first} that the Neumann expansion for \eqref{eq:aa1} converges uniformly in $(\vec{p},z) \in  U[\vec{p}^*] \times D_{\mu/2}$; see \eqref{37}.  As $(H_f+ \rho_0)^{-1/2} H_I(H_f+ \rho_0)^{-1/2}$ is bounded and independent of $ \vec{p}$ and $z$, we only need to show that the operator   $$ \left[ \frac{H_{f}+ \rho_0}  {H_{0}(\vec{p}, \rho_0z)} \right]_{ \text{Ran}(\overline{\bm{\chi}})}$$  is analytic  on $U[\vec{p}^*] \times D_{\mu/2}$.  By the  equivalence between weak and strong analyticity stated in Theorem \ref{hahn}, and the polarization identity for bilinear forms, it is sufficient  to check that $$\langle   \downarrow  \otimes \psi \vert \left[ (H_{f}+ \rho_0)  (H_{0}(\vec{p}, \rho_0z))^{-1} \right]_{ \text{Ran}(\overline{\bm{\chi}})}   \downarrow  \otimes \psi \rangle$$ and $$\langle   \uparrow  \otimes \psi \vert \left[ (H_{f}+ \rho_0)  (H_{0}(\vec{p}, \rho_0z))^{-1} \right]_{ \text{Ran}(\overline{\bm{\chi}})}  \uparrow  \otimes \psi \rangle$$ are analytic in $\vec{p}$ and $z$ for any $  \psi   \in \mathcal{H}_f$.  The functions 
\begin{equation*}
(\vec{p},z,r,\vec{l}) \mapsto \frac{r + \rho_0}{\frac{\vec{l}^2}{2m} - \frac{\vec{p}}{m} \cdot \vec{l} +r -\rho_0z}, \qquad (\vec{p},z,r,\vec{l}) \mapsto \frac{r + \rho_0}{\frac{\vec{l}^2}{2m} - \frac{\vec{p}}{m} \cdot \vec{l} +r +\omega_0 - \rho_0z} ,
\end{equation*}
are analytic in $(\vec{p},z)$ for any  $(r,\vec{l}) \in  \mathbb{R}_+ \times \mathbb{R}^3$ with $r \geq 3 \rho_0/4$ and $\vert \vec{l} \vert \leq r$, and for any $(r,\vec{l}) \in  \mathbb{R}_+ \times \mathbb{R}^3$  with $\vert \vec{l} \vert \leq r$, respectively. Using the  functional calculus of  Paragraph \ref{funcpar}, we   deduce that  $$\langle   \downarrow  \otimes \psi \vert \left[ (H_{f}+ \rho_0)  (H_{0}(\vec{p}, \rho_0z))^{-1} \right]_{ \text{Ran}(\overline{\bm{\chi}})}   \downarrow  \otimes \psi \rangle$$ and $$\langle   \uparrow  \otimes \psi \vert \left[ (H_{f}+ \rho_0)  (H_{0}(\vec{p}, \rho_0z))^{-1} \right]_{ \text{Ran}(\overline{\bm{\chi}})}   \uparrow  \otimes \psi \rangle$$  are analytic on $U[\vec{p}^*] \times D_{\mu/2}$. This is a direct consequence of  Morera's theorem for several complex variables; See Theorem \ref{morera} in Appendix \ref{Hoo}.  For a similar and detailed argumentation, the reader can consult the proof of Lemma \ref{compo} in Section \ref{presan}.
\end{proof}

 \indent
 
 \section{Iteration of the renormalization map and Preservation of analyticity in $\vec{p}$ and $z$}
 The renormalization map $\mathcal{R}_{\rho}$ can be iterated  arbitrarily many times if, initially,  it is  restricted to a sufficiently small polydisc $\mathcal{B}(\gamma, \delta, \varepsilon)$ of operator valued functions. The purpose of this section is to explain under which conditions this iteration can be done and  to  show  that the analyticity of the operator $ H(\underline{w}(\vec{p},z))$ in $(\vec{p},z)$  is preserved under iterations of the renormalization map.

\subsection{Results on Analyticity}
\label{presan}
We denote by $\mathcal{B}^{\text{an}}(\gamma, \delta, \varepsilon)$ the subset  of   $\mathcal{B}(\gamma, \delta, \varepsilon)$ composed of  the analytic maps $U[\vec{p}^*] \times D_{\mu/2} \ni (\vec{p},z) \mapsto H(\underline{w}(\vec{p},z)) \in \mathcal{B}(\mathcal{H}_{\text{red}})$. Lemma \ref{preserved} states that the renormalization map preserves analyticity.
\begin{lemma}
\label{preserved}
Let $0 < \rho < 1/2$, $0< \xi < 1/ (4\sqrt{8 \pi})$,  $0<\gamma \ll  \mu $, and $0< \delta,\varepsilon\ll \rho \mu $.  If  $ H(\underline{w}(\cdot,\cdot)) \in \mathcal{B} ^{\text{an}}( \gamma,\delta,\varepsilon)$, then the $\mathcal{B}(\mathcal{H}_{\text{red}})$-valued function  $H(\underline{\hat{w}}(\vec{p}, \zeta )) = \mathcal{R}_\rho (H (\underline{w}) )(\vec{p}, \zeta )$ is  analytic on $U[\vec{p}^*] \times D_{\mu/2}$.
\end{lemma}
\begin{remark}
$\quad$
\begin{itemize}
\item Different strategies  lead to Lemma \ref{preserved}. The  preservation of analyticity can be proven for kernels: If $(\vec{p},z) \mapsto \underline{w}(\vec{p},z) \in \mathcal{W}^{\sharp}$ is analytic on $U[\vec{p}^*] \times D_{\mu/2}$, one can show that $(\vec{p},z) \mapsto \underline{\hat{w}}(\vec{p},z) \in \mathcal{W}^{\sharp}$ is analytic. To do so, it suffices to consider the explicit expression of  the kernel $\underline{\hat{w}}$, as given in Appendix \ref{Wicko}. One shows that all  the terms appearing in the series \eqref{whatmn} are analytic.  The uniform convergence   in $\vec{p}$ and $z$ of the series in  \eqref{whatmn}  is then sufficient to establish the result. Such a procedure has been used in \cite{HaHe11_01} to show the analyticity of the ground state energy with respect to the coupling constant in the spin boson model.  Here, we prefer to work on the operator level and  we  adopt an approach similar to \cite{GrHa09_01}.
\item We recall that, to pass from the operator $\mathcal{R}_\rho (H (\underline{w}) )(\vec{p}, \zeta )$ given by the expression \eqref{eq:def_Rrho} to its Wick-ordered expression $H(\underline{\hat{w}}(\vec{p}, \zeta ))$, it suffices to follow the procedure described in Paragraph \ref{Wicki} (and detailed in  Appendix \ref{Wicko}), based on the combination of the pull-through formula and Wick-ordering. This uniquely determines the sequence of kernels $\underline{\hat{w}}( \vec{p} , \zeta )$. Naturally, since the two operators coincide, proving analyticity of $\mathcal{R}_\rho (H (\underline{w}) )(\vec{p}, \zeta )$ is equivalent to proving analyticity of $H(\underline{\hat{w}}(\vec{p}, \zeta ))$.
\end{itemize}
\end{remark}
To prove Lemma \ref{preserved}, we first state Lemma \ref{compo} which shows that  the analyticity  of   $(\vec{p},z) \mapsto H(\underline{w}(\vec{p},z)) \in \mathcal{W}_{\text{op}}$   implies the analyticity of its ``components''  in the  following sense.

\begin{lemma}
\label{compo}
Assume that the $\mathcal{B}(\mathcal{H}_{\text{red}})$-valued function $H(\underline{w}(\cdot,\cdot)) \in \mathcal{B}(\gamma, \delta, \varepsilon)$ is analytic on $U[\vec{p}^*] \times D_{\mu/2}$.  Then $(\vec{p},z) \mapsto w_{0,0}(\vec{p},z,0,\vec{0})  \in \mathbb{C} $, $(\vec{p},z) \mapsto W_{0,0}(\underline{w}(\vec{p},z)) \in \mathcal{B}(\mathcal{H}_{\text{red}}) $ and $(\vec{p},z) \mapsto W_{\geq 1}(\underline{w}(\vec{p},z)) \in \mathcal{B}(\mathcal{H}_{\text{red}})$ are analytic on  $U[\vec{p}^*] \times D_{\mu/2}$.
\end{lemma}
\vspace{2mm}

\begin{proof}[Proof of Lemma \ref{preserved}] The proof follows  from  Lemmas \ref{lm:bihol} and \ref{compo}.  We remind the reader that
\begin{equation*}
\mathcal{R}_\rho (H (\underline{w}) )(\vec{p}, \zeta ) = S_\rho \left( F_{ \chi_\rho(H_f) } \left( H \left( \underline{w} \left( E_\rho^{-1}( \vec{p},\zeta )\right)  \right) , W_{0,0} \left( \underline{w} \left(E_\rho^{-1}(\vec{p}, \zeta ) \right) \right) \right) \right).
\end{equation*}
We have seen in Lemma \ref{lm:bihol} that $E_\rho^{-1}( \vec{p},z )$ is holomorphic on $U[\vec{p}^*] \times D_{\mu/2}$.  The composition of two analytic maps being analytic,  $ H \left( \underline{w} \left( E_\rho^{-1}( \vec{p},\zeta )\right) \right)$ and $ W_{0,0} \left( \underline{w} \left(E_\rho^{-1}(\vec{p}, \zeta ) \right)   \right) $ are analytic on $U[\vec{p}^*] \times D_{\mu/2}$.  We now remark that  the smooth Feshbach-Schur map preserves analyticity. Indeed, if $H(\vec{p},z)$ and $T(\vec{p},z)$ are analytic $\mathcal{B}(\mathcal{H}_{\text{red}})$-valued functions, then, with $W \equiv H - T$,
\begin{equation*}
F_{\chi}(H(\vec{p},z) , T(\vec{p},z))=H_{\chi}(\vec{p},z)-\chi W(\vec{p},z) \overline{\chi} [H_{\overline{\chi}}(\vec{p},z)]^{-1}_{ \text{Ran}(\overline{\bm{\chi}})}  \overline{\chi}  W(\vec{p},z) \chi
\end{equation*}
is analytic as  a product of bounded analytic operators, provided that the Feshbach-Schur map is well-defined.  As the scale transformation $S_{\rho}$ preserves analyticity, we deduce that $\mathcal{R}_\rho (H (\underline{w}) )(\vec{p}, \zeta ) $ is analytic on $U[\vec{p}^*] \times D_{\mu/2}$. 
\end{proof}
\begin{proof}[Proof of Lemma \ref{compo}]    We first prove that $(\vec{p},z) \mapsto w_{0,0}(\vec{p},z, r ,\vec{l}) \in \mathbb{C}$ is analytic for any fixed $(r,\vec{l}) \in \mathcal{B}$.   In order to ``extract'' $w_{0,0}$ from  $H(\underline{w}(\cdot,\cdot))$,  we consider a sequence of  smooth functions whose squared   modulus  converges weakly to a Dirac distribution.  Let $\eta$ be a smooth function with compact support, such that $ \int_{\mathbb{R}^{3}} \vert \eta(\vec{x}) \vert^2 d^{3}x =1.$ For any $\vec{l}$ in the unit ball of $ \mathbb{R}^3$, for any $n \in \mathbb{N}$, we set
 \begin{equation}
 \eta_{n, \vec{l}}(\vec{k}):= n^{3/2} \eta\left( n (\vec{k}- \vec{l}) \right).
 \end{equation}
 Let $\vec{l}, \vec{l}'$ be in the unit ball of $\mathbb{R}^3$. We introduce the two-particle state $\Psi_{n}:= b^{*}( \eta_{n, \vec{l}}) b^{*}( \eta_{n, \vec{l}'})   \Omega $ and consider the sequence of holomorphic functions $(g_n)_{n \in \mathbb{N}}$  defined by
 \begin{equation}
 g_n(\vec{p},z):= \langle \Psi_{n} \vert  \mathds{1}_{H_f \leq 1}    H(\underline{w}(\vec{p},z))   \mathds{1}_{H_f \leq 1}  \Psi_{n} \rangle,
\end{equation}
 for all  $n \in \mathbb{N}$ and all $(\vec{p},z) \in U[\vec{p}^*] \times D_{\mu/2}$. 
We show in Appendix \ref{lemacom}  that $(g_n)_n$ converges uniformly  on $U[\vec{p}^*] \times D_{\mu/2}$ to  a multiple of $w_{0,0}(\vec{p},z,\vert \vec{l} \vert + \vert \vec{l}' \vert,\vec{l}+ \vec{l}')$. This condition is sufficient to ensure that $(\vec{p},z) \mapsto w_{0,0}(\vec{p},z, r,\vec{l} )$ is analytic for any fixed $(r,\vec{l}) \in \mathcal{B}$; see  Lemma \ref{7.1}.

We now show that  $(\vec{p},z) \mapsto W_{0,0}(\underline{w}(\vec{p},z)) \in \mathcal{B}(\mathcal{H}_{\text{red}})$ is  analytic on  $U[\vec{p}^*] \times D_{\mu/2}$. To do so, we make use of the equivalence between weak and strong analyticity stated in Theorem \ref{hahn}, and we show that  $(\vec{p},z) \mapsto W_{0,0}(\underline{w}(\vec{p},z)) \in \mathcal{B}(\mathcal{H}_{\text{red}})$ is weakly analytic. Thanks to the polarization identity for bilinear forms,  it is sufficient to prove that  $\langle \Psi  \vert    W_{0,0}(\underline{w}(\vec{p},z))    \Psi \rangle$ is analytic for any $   \Psi  \in \mathcal{H}_{\text{red}}$.  Let $  \Psi  \in \mathcal{H}_{\text{red}}$. We write $  \Psi  $ as a sequence of $n$-photon state functions $(\psi^{(n)})_{n \in \mathbb{N}}$;   By the functional calculus  outlined in Paragraph \ref{funcpar}, we have that
\begin{equation}
\label{serr}
\begin{split}
\langle \Psi  \vert    W_{0,0}(\underline{w}(\vec{p},z))   \Psi \rangle&=\sum_{n=1}^{\infty} \int_{\mathcal{D}(n) } \prod_{i=1}^{n}d\underline{k}_i  \text{ }   \vert  \psi^{(n)}(\underline{k}_1,...,\underline{k}_n)  \vert^{2}  w_{0,0}(\vec{p},z,\Sigma(\underline{k}^{(n)}),\vec{\Sigma}(\underline{k}^{(n)}))\\
& \qquad + w_{0,0}(\vec{p},z,0,\vec{0}) \vert \psi^{(0)} \vert^2,
\end{split}
\end{equation}
where $$\mathcal{D}(n):=\lbrace (\underline{k}_1,...,\underline{k}_n) \in \underline{B}_{1}^{n} \mid  \vert \vec{k_1}  \vert+...+  \vert \vec{k}_n  \vert \leq 1 \rbrace.$$
We use  the theorem of Morera for several complex variables  (see Theorem \ref{morera} in Appendix \ref{Hoo}) to show the analyticity of the map  $(\vec{p},z) \mapsto \langle \Psi  \vert    W_{0,0}(\underline{w}(\vec{p},z))   \Psi \rangle \in \mathbb{C}$.  Let $V_i:=l_1 \times \Delta_i \times l_3 \times l_4 $, $(i=1,...,4)$, be a  surface included in the domain $U[\vec{p}^*] \times D_{\mu/2}$ that satisfies the criteria of Theorem \ref{morera}.  $H(\underline{w}) \in \mathcal{B}(\gamma,\delta,\varepsilon)$, and there exists a constant $C>0$ such that  $ \vert  w_{0,0}(\vec{p},z,r,\vec{l}) \vert \leq C$, for all $(\vec{p},z) \in U[\vec{p}^*] \times D_{\mu/2}$ and $(r,\vec{l}) \in \mathcal{B}$. As  
\begin{equation*}
\sum_{n=0}^{\infty} \int_{\mathcal{D}(n)} \prod_{i=1}^{n}d\underline{k}_i \text{ }  \vert  \psi(\underline{k}_1,...,\underline{k}_n)  \vert^{2} = \| \Psi \|^2<\infty,
\end{equation*}
we deduce that  the function $(\vec{p},z) \mapsto \langle \Psi  \vert    W_{0,0}(\underline{w}(\vec{p},z))  \Psi \rangle$ is continuous; Furthermore, by integrating \eqref{serr} over $\partial V_i$,  the sum over $n$  and the  integral over $\partial V_i$ can be exchanged. Fubini's theorem  implies that 
\begin{align}
&\int_{\partial V_i}     \langle \Psi  \vert    W_{0,0}(\underline{w}) \Psi \rangle = \int_{ \partial V_i} \text{ } w_{0,0}(\cdot ,\cdot,0,\vec{0}) \vert \psi^{(0)} \vert^2  \notag \\
&+\sum_{n=1}^{\infty} \int_{\underline{\mathbb{R}}^n}  \prod_{i=1}^{n}d\underline{k}_i    \vert  \psi^{(n)}(\underline{k}_1,...,\underline{k}_n)  \vert^{2}  \int_{ \partial V_i}   \text{ }   w_{0,0}(\cdot,\cdot,\Sigma(\underline{k}^{(n)}),\vec{\Sigma}(\underline{k}^{(n)})) . \label{marre}
\end{align}
The analyticity of $U[\vec{p}^*] \times D_{\mu/2}  \ni (\vec{p},z) \mapsto w_{0,0}(\vec{p},z,\Sigma(\underline{k}^{(n)}),\vec{\Sigma}(\underline{k}^{(n)}))$ for  any $(\Sigma(\underline{k}^{(n)}),\vec{\Sigma}(\underline{k}^{(n)}))\in \mathcal{B}$ shows that the right-hand side of \eqref{marre} is equal to zero.   Hence,  Morera's theorem implies that  $(\vec{p},z) \mapsto \langle \Psi  \vert    W_{0,0}(\underline{w}(\vec{p},z))  \vert \Psi \rangle$ is analytic for any $\Psi$ in $\mathcal{H}_{\text{red}}$, and we conclude that  $(\vec{p},z) \mapsto     W_{0,0}(\underline{w}(\vec{p},z)) \in \mathcal{B}(\mathcal{H}_{\text{red}})$ is analytic on $U[\vec{p}^*] \times D_{\mu/2}$. Obviously, since $W_{\geq 1}(\underline{w}(\vec{p},z)) = H(\underline{w}(\vec{p},z)) - W_{0,0}(\underline{w}(\vec{p},z))$,  $U[\vec{p}^*] \times D_{\mu/2}  \ni (\vec{p},z) \mapsto W_{\geq 1}(\underline{w}(\vec{p},z))  \in \mathcal{B}(\mathcal{H}_{\text{red}})$ is also analytic.
\end{proof}

\subsection{Iterating the renormalization map}
In Subsection \ref{presan} it has been shown that the renormalization map preserves analyticity. We now investigate under which condition the renormalization map can be iterated.
\subsubsection{Codimension-4 contractivity}
 The renormalization map is codimension-4 contractive, in the sense of the following lemma.
\begin{lemma}
\label{contr}
Let $0<\rho \ll \mu$, $0<\xi<1/ (4\sqrt{8 \pi})$,  $0<\gamma\ll \mu$,  $0<\delta \ll \rho \mu$,  and $0<\varepsilon \ll \rho^2 \mu^2$. Then,
 \begin{equation}
 \mathcal{R}_{\rho}: \mathcal{B}^{\text{an}}(\gamma, \delta, \varepsilon) \rightarrow \mathcal{B}^{\text{an}}(\gamma + \frac{\varepsilon}{2},  \frac{\varepsilon}{2}, \frac{\varepsilon}{2}).
 \end{equation}
\end{lemma}
  
 The proof of Lemma \ref{contr}  is similar to the proof of  \cite[Theorem 3.8]{BaChFrSi03_01}. We sketch its main idea in the rest of this paragraph.  For the sake of completeness, a  detailed exposition is given in Appendix \ref{AppD}. If  $\rho$ and $\varepsilon$ are sufficiently small, we expect  from  the scaling properties of the kernels in (2.24) that the irrelevant part $W_{\geq 1}(\underline{\hat{w}})$ of the effective operator $H(\underline{\hat{w}}):=\mathcal{R}_{\rho}(H(\underline{w}))$ will have a smaller  norm   than $W_{\geq 1}(\underline{w})$. To  check it, we need to  bound the Wick-ordered series $W_{\geq 1}(\underline{\hat{w}})$. $W_{\geq 1}(\underline{\hat{w}})$ is a sum over $L \geq 1$,  followed by a sum over $M_1,\dots,M_L/N_1,\dots,N_L$ and $m_1,\dots,m_L/n_1,\dots,n_L$ ($ \sum_i m_i +n_i \geq 1$, $0 \leq m_i \leq M_i$, $0 \leq n_i \leq N_i$),  of operators of the form
\begin{equation}
\label{couscous2}
 C^{\underline{M},\underline{N}}_{\underline{m},\underline{n}} \int  d\underline{K}^{(\underline{m},\underline{n})}    b^{*}(\underline{k}^{(\underline{m})}) \text{ }  \langle \Omega \vert \mathscr{V}^{\underline{M},\underline{N}}_{\underline{m},\underline{n}}(\vec{p},z,r,\vec{l},\underline{K}^{(\underline{m},\underline{n})})   \Omega  \rangle_{r=H_f,\vec{l}=\vec{P}_f}  \text{ }b(\underline{\tilde{k}}^{(\underline{n})}).
 \end{equation}
 The norm of  the operator $\mathscr{V}^{\underline{M},\underline{N}}_{\underline{m},\underline{n}}(\vec{p},z,r,\vec{l},\underline{K}^{(\underline{m},\underline{n})})$ in \eqref{couscous2}  can be bounded by the product of  the norms of  the kernels $w_{M_i,N_i}$ with the norms of the  operator valued functions $F$.   After rescaling,  \eqref{couscous2} is  bounded in norm by 
 \begin{equation}
 \label{mau}
  C^{\underline{M},\underline{N}}_{\underline{m},\underline{n}} \rho^{-1} \left( \prod_{i=1}^{L}  \varepsilon (\sqrt{8 \pi }\xi)^{M_i+N_i} \rho^{2(m_i+n_i)} \right) \left(\frac{C}{\rho \mu} \right)^{L-1}.
 \end{equation}
In \eqref{mau},  the term $(C/\rho \mu)^{L-1}$ comes from the  norm of the rescaled operators $F$, and the factors $\varepsilon \xi^{M_i+N_i}$ arise because $\vert \vert \underline{w}_{\geq 1} \vert \vert^{\sharp} \leq \varepsilon$. It is  not difficult to realize that the sum over $L$, $M_1,\dots,M_L/N_1,\dots,N_L$ and $m_1,\dots,m_L/n_1,\dots,n_L$, of  $\eqref{mau}$,   is  bounded if $\rho$ and $\varepsilon$ are small enough. In particular, it becomes smaller than $\varepsilon/2$ for  sufficiently small $\rho,\varepsilon$. 

The relevant part does not contract under renormalization. Indeed,
\begin{equation}
W_{0,0}(\underline{\hat{w}}(\vec{p},\zeta))= S_{\rho} W_{0,0} (\underline{w}(E_{\rho}^{-1}(\vec{p},\zeta))) + \tilde{W}  (\underline{w}(E_{\rho}^{-1}(\vec{p},\zeta))),
\end{equation}
where $\tilde{W}  (\underline{w}(E_{\rho}^{-1}(\vec{p},\zeta)))$ is the sum over  all Wick-ordered  operators of the  form  \eqref{mau} with $m_1=...=m_L=n_1=\dots=n_L=0$. This term can be bounded by $\varepsilon/2$ if  $\rho$ and $\varepsilon$ are sufficiently  small. However, the norm of  $S_{\rho} W_{0,0} (\underline{w}(E_{\rho}^{-1}(\vec{p},\zeta))) $  does not change much. The reader can consult Appendix \ref{AppD} for more detailed calculations.

 \subsubsection{Iteration of the renormalization map}
\label{itere}
Let $\rho, \xi , \gamma, \delta,\varepsilon$  satisfying the constraints of Lemma \ref{contr}. Let $0 < \rho_0 <  \min( \xi^{2/3},\omega_0)$.
Lemmas \ref{polyd} and \ref{anainit} show that there exists a constant $\lambda_c >0$, such that, for all $ 0 \le  \lambda_0  < \lambda_c$,  the operator-valued function $U[\vec{p}^*] \times D_{\mu/2} \ni (\vec{p},z) \mapsto H(\underline{w}^{(0)}(\vec{p},z))$ obtained from $(\vec{p},z) \mapsto H(\vec{p},z)$ after a  rescaling of $z$ by $\rho_0$, a  smooth Feshbach-Schur transformation, and  a scale transformation, $S_{\rho_0}$, lies in  $\mathcal{B}^{\text{an}}( \gamma  ,\varepsilon   ,  \varepsilon )$. The codimension-4 contractivity of Lemma \ref{contr} implies that 
 \begin{equation}\label{eq:a14}
\mathcal{R}_{\rho}: \mathcal{B}^{\text{an}} \left( \gamma , \varepsilon, \varepsilon \right) \rightarrow  \mathcal{B}^{\text{an}} \left( \gamma +  \frac{\varepsilon }{2},  \frac{\varepsilon }{2}, \frac{\varepsilon }{2}\right).
\end{equation}
As $\varepsilon \ll  \rho^2 \mu^2 <\mu$, $\gamma + \varepsilon \ll \mu$, we deduce that  the renormalization map can be iterated indefinitely. For any $n\geq 0$, we set
\begin{equation}
H(\underline{w}^{(n)}):=\mathcal{R}^{n}_{\rho}(H(\underline{w}^{(0)}))  \in \mathcal{B}^{\text{an}} \left( \gamma + \varepsilon r_n,  \frac{\varepsilon}{2^{n}}, \frac{\varepsilon}{2^{n}}\right),
\end{equation}
where $r_n= 1-2^{-n}$ for $n\geq 0$. Thanks to  the isospectrality of the renormalization map, we  keep track of  the   eigenvalue $0$  of the effective Hamiltonians $H(\underline{w}^{(n)})$ in order to find an eigenvalue $z_{\infty}(\vec{p})$ of the initial operator $H(\vec{p})$. We set, for any $n \in \mathbb{N} \cup \lbrace 0 \rbrace$, 
\begin{equation}
\mathcal{U}^{(n)}:=\lbrace(\vec{p},z ) \in U[\vec{p}^*] \times D_{\mu/2}  \mid \text{ } \vert  w_{0,0}^{(n)}(\vec{p},z,0,\vec{0} )  \vert < \frac{ \mu \rho}{2} \rbrace,
\end{equation}
 and we  introduce the maps $E_{\rho,n}: \mathcal{U}^{(n)} \to U[\vec{p}^*] \times D_{\mu/2}$ defined by  
$$   E_{\rho,n}(\vec{p},z):=(\vec{p},-\rho^{-1} w^{(n)}_{0,0}(\vec{p},z,0,\vec{0})),$$
 for all $(\vec{p},z) \in U[\vec{p}^*] \times D_{\mu/2}$.  We explain with some pictures   why   the sets $\mathcal{U}^{(n)}$ and the maps $E_{\rho,n}$ are relevant to  find  the ground state eigenvalue of the initial Hamiltonian $H(\vec{p})$. It is more convenient to work with $\vec{p} \in U[\vec{p}^*]$ fixed and to look at the sets
  \begin{equation}
\mathcal{U}^{(n)}[\vec{p}]:=\lbrace z  \in   D_{\mu/2} \mid \text{ } \vert  w_{0,0}^{(n)}(\vec{p},z,0,\vec{0} )  \vert <\frac{ \mu \rho}{2} \rbrace.
\end{equation}
The sets  $\mathcal{U}^{(n)}[\vec{p}]$ are non-empty for all $\vec{p} \in U[\vec{p}^*]$; see \eqref{inclu}. $\mathcal{U}^{(n)}$ is the union over $\vec{p} \in  U[\vec{p}^*]$  of all sets $ \lbrace \vec{p} \rbrace \times \mathcal{U}^{(n)}[\vec{p}]$. For any $n\geq 0$,  $\mathcal{U}^{(n)}[\vec{p}] \subset  D_{\mu/2} $ and $ \rho^{-1} w_{0,0}^{(n)}(\vec{p},\mathcal{U}^{(n)}[\vec{p}] ,0,\vec{0} ) \subset D_{\mu/2}$.
 
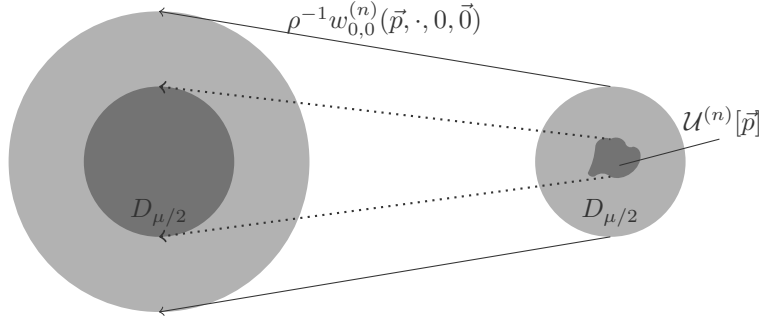
\begin{center}
\begin{tikzpicture}

      \fill[ color=gray!60] (0,0) circle (2) ;
      \fill[fill=black!55] (0,0) circle (1) ;

      \fill[color=gray!60] (6, 0) circle (1) ;
     
            \color{black!80}

    \fill[black!55]     (6,0.3)to[bend left=45] (6.2,0.25)  to[bend left=-60] (6.3,0.2)   to[bend left=40] (6.4,0.1)  to[bend left=60] (6,-0.2) to[bend left=-40] (5.8,-0.19) to[bend left=60] (5.7,-0.15) to[bend left=-10] (5.8,0.2) to[bend left=60] (6,0.3);

 \draw[<-] (0,2)--(6,1);
 \draw[<-] (0,-2)--(6,-1);
 
  \draw[<-,thick,dotted] (0,1)--(6,0.3);
 \draw[<-,thick,dotted] (0,-1)--(6,-0.2);
 
    \draw (6,-1) node[above] {\small $ D_{\mu/2}$};
   \draw (0,-1) node[above] {\small$ D_{\mu/2}$};
     \draw (7.5,0.2) node[above] { \small $\mathcal{ U } ^{(n)}[\vec{p}]$};
      \draw[-] (6.12,-0.05)--(7.45,0.3);
     \draw (3,1.5) node[above] {\small$ \rho^{-1} w_{0,0}^{(n)}(\vec{p},\cdot, 0,\vec{0})$};

\end{tikzpicture}

\end{center}

\noindent The isospectrality  of the renormalization map shows that $\dim[\ker{H(\underline{w}^{(n+1)}(\vec{p},z))}] \neq 0$ iff  $$\dim[\ker{H(\underline{w}^{(0)}(E_{\rho,0}^{-1} \circ...\circ E_{\rho,n}^{-1} (\vec{p},z)))]} \neq 0.$$ If we want $z$ to be located  inside the disc $D_{\mu/2}$ for all $n \geq 0$, the set of  initial spectral values $\Pi_{z} (E_{\rho,0}^{-1} \circ...\circ E_{\rho,n}^{-1} (\vec{p}, D_{\mu/2})) \subset D_{\mu/2}$ has to shrink with $n$. $\Pi_{z} $ is the projection along  the $z$-component in $\mathbb{C}^3 \times \mathbb{C}$, that is
\begin{equation*}
\Pi_{z}(\vec{p}',z'):=z',
\end{equation*}
for any $(\vec{p}',z') \in \mathbb{C}^3 \times \mathbb{C}$.
 
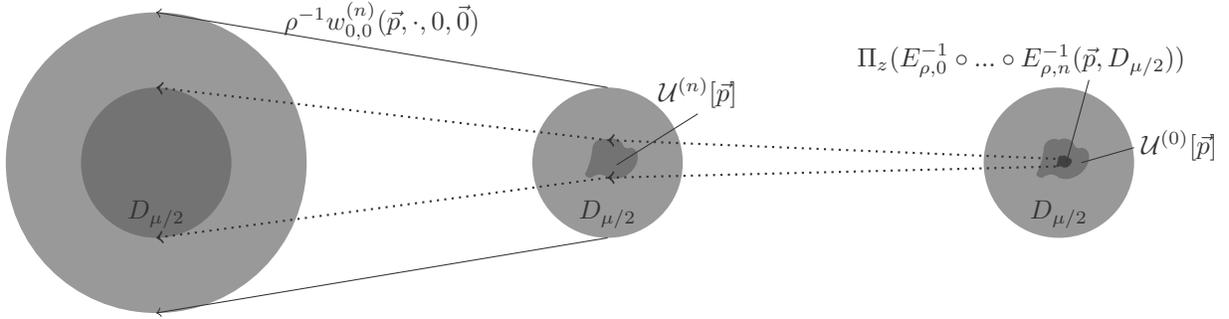
\begin{center}
\begin{tikzpicture}

      \fill[ color=gray!80] (0,0) circle (2) ;
      \fill[fill=black!55] (0,0) circle (1) ;

      \fill[color=gray!80] (6, 0) circle (1) ;
      
          \fill[color=gray!80] (12, 0) circle (1) ;

            \color{black!80}
           
   \fill[black!55]     (6,0.3)to[bend left=45] (6.2,0.25)  to[bend left=-60] (6.3,0.2)   to[bend left=40] (6.4,0.1)  to[bend left=60] (6,-0.2) to[bend left=-40] (5.8,-0.19) to[bend left=60] (5.7,-0.15) to[bend left=-10] (5.8,0.2) to[bend left=60] (6,0.3);

 \fill[black!55, shift={(6 cm,0)}]     (6,0.3)to[bend left=45] (6.2,0.25)  to[bend left=-60] (6.3,0.2)   to[bend left=40] (6.4,0.1)  to[bend left=60] (6,-0.2) to[bend left=-40] (5.8,-0.19) to[bend left=60] (5.7,-0.15) to[bend left=-10] (5.8,0.2) to[bend left=60] (6,0.3);

 \fill[black!75, shift={(10.25 cm,0)},scale=0.3]     (6,0.3)to[bend left=45] (6.2,0.25)  to[bend left=-60] (6.3,0.2)   to[bend left=40] (6.4,0.1)  to[bend left=60] (6,-0.2) to[bend left=-40] (5.8,-0.19) to[bend left=60] (5.7,-0.15) to[bend left=-10] (5.8,0.2) to[bend left=60] (6,0.3);

 \draw[<-] (0,2)--(6,1);
 \draw[<-] (0,-2)--(6,-1);
  \draw[-] (7.2,0.65)-- (6.12,-0.03); 
  \draw[<-,thick,dotted] (0,1)--(6,0.3);
 \draw[<-,thick,dotted] (0,-1)--(6,-0.2);
 
    \draw (6,-1) node[above] {\small $ D_{\mu/2}$};
   \draw (0,-1) node[above] {\small$ D_{\mu/2}$};
     \draw (7.2,0.6) node[above] { \small $\mathcal{ U }^{(n)}[\vec{p}]$};
     \draw (3,1.5) node[above] {\small$ \rho^{-1} w_{0,0}^{(n)}(\vec{p},\cdot, 0,\vec{0})$};
       \draw (12,-1) node[above] {\small $ D_{\mu/2}$};
   \draw (13.6,-0.1) node[above] { \small $\mathcal{ U }^{(0)}[\vec{p}]$};
    \draw (11.5,1) node[above] { \small $ \Pi_{z} ( E_{\rho,0}^{-1} \circ ...\circ  E_{\rho,n}^{-1} (\vec{p},D_{\mu/2}))$};
      
       \draw[<-,thick,dotted] (6,0.3)--(12,0.05);
 \draw[<-,thick,dotted] (6,-0.2)--(12,-0.05);

               \draw[-] (12.4,1.1)--(12.1,0.0);
                \draw[-] (13.,0.2)--(12.3,0.0);
      
\end{tikzpicture}

\end{center}

\noindent In the limit where  $n$ tends to infinity, we expect  the set $\Pi_{z} (E_{\rho,0}^{-1} \circ...\circ E_{\rho,n}^{-1} (\vec{p}, D_{\mu/2}))$ to shrink to a point located near $0$, which will turn out to be the eigenvalue of $H(\vec{p})$ rescaled by a factor $\rho_{0}^{-1}$.  We  define  the  sequence  $( e_{0,n})_{n}$ of complex-valued functions  on  $U[\vec{p}^*] \times D_{\mu/2}$ by  
\begin{equation}
\label{emn1}
 e_{ 0,n}(\vec{p},z):= \Pi_{z} (E_{\rho,0}^{-1} \circ...\circ E_{\rho,n}^{-1} (\vec{p}, z)),
\end{equation}
 for all $n \geq 0$ and for all $(\vec{p},z) \in U[\vec{p}^*] \times D_{\mu/2}$.
\vspace{2mm}

\begin{lemma}
\label{44}
Let  $0<\rho\ll \mu$, $0<\xi<1/ (4\sqrt{8 \pi})$,  $0<\gamma\ll \mu$,  $0<\delta \ll \rho \mu$,  and $0<\varepsilon \ll \rho^2 \mu^2$.  The sequence of functions $ (e_{0, n})_n$ converges uniformly  on $U[\vec{p}^*] \times D_{\mu/2}$ to a function $e_{ 0,\infty}$.
$e_{ 0,\infty}(\vec{p},z )$  is independent of $z$.
\end{lemma}
\vspace{2mm}

\begin{proof} 
We  define  the  sequence  $( e_{(m,n)})_{m,n}$ of functions  on  $U[\vec{p}^*] \times D_{\mu/2}$,   $n \geq m$,  by  
\begin{equation}
\label{emn}
e_{(m,n)}(\vec{p},z):= \Pi_{z} (E_{\rho,m}^{-1} \circ...\circ E_{\rho,n}^{-1} (\vec{p}, z)).
\end{equation}
Let  $m \geq 0$. We   show that the  sequence $(e_{(m,n)}(\vec{p},z))_n$ converges uniformly  in $(\vec{p},z)$  to some function $e_{m, \infty}(\vec{p},z)$.  We have seen in the proof of Lemma \ref{lm:bihol} that $\vert \partial_z w^{(m)}_{0,0}(\vec{p},z,0,\vec{0}) + 1 \vert<C <1$, for all  $(\vec{p},z) \in \mathcal{U}^{(m)}$. The constant $C$ does not depend on $m$, because $H(\underline{w}^{(m)}) \in \mathcal{B}^{\text{an}}(\gamma+ \varepsilon, \varepsilon,\varepsilon)$ for all $m \geq 0$. We denote by $h^{(m)}$ the complex-valued analytic function defined  by  $ (\vec{p},h^{(m)}(\vec{p},z)):=E^{-1}_{\rho,m}(\vec{p},z)$, for any $(\vec{p},z) \in U[\vec{p}^*] \times D_{\mu/2}$.  $\vert \partial_z h^{(m)}(\vec{p},z) \vert<  \rho/( 1-C)$ for all $(\vec{p},z) \in U[\vec{p}^*] \times D_{\mu/2}$.  If $\rho$ is  sufficiently small, $\rho/( 1-C)<1$. Let $ n,k \geq 0$, $n \geq m$,  and $(\vec{p},z) \in U[\vec{p}^*] \times D_{\mu/2}$. Then, $e_{(m,n)} =h^{(m)} \circ ... \circ h^{(n)}$, and 
\begin{align*}
\vert e_{m,n}(\vec{p},z)- e_{m,n+k}(\vec{p},z) \vert &< \frac{\rho}{1-C} \vert  e_{m+1,n}(\vec{p},z)- e_{m+1,n+k}(\vec{p},z) \vert < ... \\
&<  \left( \frac{\rho}{1-C} \right)^{n-m} \vert  e_{n,n}(\vec{p},z)- e_{n,n+k}(\vec{p},z) \vert < \left( \frac{\rho}{1-C} \right)^{n-m}\mu. 
\end{align*}
  The sequence $ (e_{m,n}(\vec{p},z))_n$ is Cauchy and converges to $e_{m, \infty}(\vec{p},z)$. The convergence is uniform in $(\vec{p},z)$, as $\vert e_{m,n}(\vec{p},z)- e_{m,\infty}(\vec{p},z) \vert \leq (\rho/(1-C))^{n-m} \mu$ for any $(\vec{p},z) \in U[\vec{p}^*] \times D_{\mu/2}$. Let $z,z'\in D_{\mu/2}$ and $\vec{p} \in U[\vec{p}^*]$. Then
$$ \vert e_{m,n}(\vec{p},z)- e_{m,n}(\vec{p},z') \vert < \left( \frac{\rho}{1-C} \right)^{n-m} \vert z-   z' \vert,$$
  and we deduce that $e_{m,\infty}(\vec{p},z)$ does not depend on $z$. As $H(\underline{w}^{(m)})  \in \mathcal{B}^{\text{an}} \left( \gamma+  \varepsilon r_m,  \frac{\varepsilon}{2^{m}}, \frac{\varepsilon}{2^{m}}\right)$,  
\begin{equation*}
\vert e_{m,\infty}(\vec{p},z)- \rho e_{m+1,\infty}(\vec{p},z) \vert \leq 2^{-m} \varepsilon.
\end{equation*}
We deduce that
\begin{equation*}
\vert e_{0,\infty}(\vec{p},0) \vert \leq \sum_{j=0}^{\infty} \rho^{j} \vert e_{j,\infty}(\vec{p},0)- \rho e_{j+1,\infty}(\vec{p},0) \vert \leq \frac{2\varepsilon}{2- \rho}, 
\end{equation*}
 for any $ \vec{p} \in U[\vec{p}^*]$. Therefore, $e_{0,\infty}(\vec{p},z)=e_{0,\infty}(\vec{p})$ belongs to a small disk of radius $2\varepsilon$ centered at the origin of the complex plane.\end{proof}
\vspace{2mm}

As  $n$ tends to infinity, the perturbation $W_{\geq 1}(\underline{w}^{(n)})$ tends to zero, and we expect the sequence of operators $H(\underline{w}^{(n)}(\vec{p},e_{n,\infty}(\vec{p},0)))$ to  converge uniformly to $\alpha(\vec{p}) H_f + \vec{\beta}(\vec{p}) \cdot \vec{P}_f $.  The operator $ \alpha(\vec{p}) H_f + \vec{\beta}(\vec{p}) \cdot \vec{P}_f$ has an eigenvalue $0$ with associated eigenvector $  \Omega$,   and  we expect  the complex sequence $( \rho_0 e_{0,n}(\vec{p},0))_n$  to converge to an eigenvalue of the initial ``Hamiltonian'' $H(\vec{p})$. 
\vspace{2mm}

\begin{lemma}
\label{fixp}
Let  $0<\rho\ll \mu$, $0<\xi<1/ (4\sqrt{8 \pi})$,  $0<\gamma\ll \mu$,  $0<\delta \ll \rho \mu$,    $0<\varepsilon \ll \rho^2 \mu^2$, and $\vec{p} \in U[\vec{p}^*]$.  Then $H(\underline{w}^{(n)}(\vec{p},e_{n,\infty}(\vec{p},0)))$ converges  in norm to an operator $H_{\text{Fix}}(\vec{p}):=\alpha(\vec{p}) H_f +  \vec{\beta}(\vec{p}) \cdot \vec{P}_f $, where $\alpha(\vec{p})   \in \mathbb{C}$ and  $\vec{\beta}(\vec{p}) \in \mathbb{C}^3$.  If $\vec{p} \in \mathbb{R}^3 \cap U[\vec{p}^*]$, $\alpha(\vec{p})  \in \mathbb{R}$ and  $\vec{\beta}(\vec{p}) \in \mathbb{R}^3$.
\end{lemma}
\vspace{2mm}

\begin{proof}
Let $\vec{p} \in U[\vec{p}^*] $. We  show that the sequence of kernels  $(w_{0,0}^{(n)}(\vec{p},e_{n,\infty}(\vec{p},0),r,\vec{l}))_n$ converges uniformly on $\mathcal{B}$  to  $\alpha(\vec{p}) r +  \vec{\beta}(\vec{p}) \cdot \vec{l}$. We introduce the function
\begin{equation*}
T_{n}(\vec{p},r,\vec{l}):=w_{0,0}^{(n)}(\vec{p},e_{n,\infty}(\vec{p},0),r,\vec{l})- w_{0,0}^{(n)}(\vec{p},e_{n,\infty}(\vec{p},0),0,\vec{0}),
\end{equation*}
defined for any $(r,\vec{l}) \in \mathcal{B}$.
We show that the continuous  functions $(r,\vec{l}) \mapsto (\partial_r T_{n})(\vec{p},r,\vec{l}) \in \mathbb{C}$ and $(r,\vec{l}) \mapsto (\partial_{l_j} T_{n})(\vec{p},r,\vec{l})  \in \mathbb{C}$ converge to constants $\alpha(\vec{p})$ and $\beta_j(\vec{p})$ uniformly on $\mathcal{B}$, respectively.  We set 
\begin{equation*}
\Delta T_{n}(\vec{p},r,\vec{l}):=T_{n}(\vec{p},r,\vec{l})-\rho^{-1}T_{n-1}( \vec{p},\rho r, \rho \vec{l}).
\end{equation*}
For any $m,n \in \mathbb{N}$ with $n>m$,
\begin{equation}
\label{premeq}
T_{n} (\vec{p},r,\vec{l})- \rho^{m-n}T_{m}(\vec{p},\rho^{n-m}r, \rho^{n-m}\vec{l})=\sum_{i=m+1}^{n} \rho^{i-n} \Delta T_{i} (\vec{p},\rho^{n-i}r,\rho^{n-i}\vec{l}).
\end{equation}
From the exact formula \eqref{w_0,0^eff} for the kernels $w_{0,0}^{(m)}$, we deduce that
\begin{align*}
& \Delta T_{n}(\vec{p},r,\vec{l}) \\
&=\rho^{-1} \sum_{L=2}^{\infty} (-1)^{L-1} \underset{\underset{p_i+q_i \geq 1}{\underline{p} , \underline{q}}}{\sum}  \left( V^{(n)}_{\underline{0,p,0,q}} (\vec{p}, e_{n-1,\infty}(\vec{p},0),r,\vec{l}) - V^{(n)}_{\underline{0,p,0,q}} (\vec{p}, e_{n-1,\infty}(\vec{p},0),0,\vec{0}) \right).
\end{align*}
The proof of Lemma  \ref{contr} (see Appendix \ref{AppD}) shows that  $\Delta T_{n}(\vec{p},r,\vec{l})$ is differentiable with respect to the variables $r$ and $l_j$, for $j=1,\dots,3$, as the series of partial derivatives converge uniformly on $\mathcal{B}$. Furthermore,  $ \| \partial_r \Delta T_{n}(\vec{p}) \|_{\infty} \leq  \varepsilon/2^{n-1}$ and  $\| \partial_{l_j} \Delta T_{n}(\vec{p}) \|_{\infty} \leq \varepsilon/2^{n-1}$. Differentiating \eqref{premeq} with respect to $r$, we have that 
\begin{equation*}
 \vert  (\partial_r T_{n}) (\vec{p},r,\vec{l})- (\partial_r T_{m})(\vec{p},\rho^{n-m}r, \rho^{n-m}\vec{l}) \vert \leq  \frac{\varepsilon}{2^{m-1}},
\end{equation*}
 for all $(r,\vec{l}) \in \mathcal{B}$ and $n>m$. In particular, for $(r,\vec{l}) =(0,\vec{0})$,
\begin{equation*}
 \vert  (\partial_r T_{n}) (\vec{p}, 0,\vec{0})- (\partial_r T_{m})(\vec{p},0,  \vec{0}) \vert \leq  \frac{\varepsilon}{2^{m-1}}.
 \end{equation*}
 $( (\partial_r T_{n}) (\vec{p}, 0,\vec{0}))_n$ is a complex Cauchy sequence and converges to a complex number that we denote by $\alpha(\vec{p})$. The same result holds for the partial derivatives with respect to $l_j$, and we denote their limits by $\beta_j(\vec{p})$.  $\alpha(\vec{p})$ and $\beta_j(\vec{p}) \in \mathbb{R}$ if $\vec{p} \in \mathbb{R}^3 \cap U[\vec{p}^*]$. We show that the sequence of  functions $(T_{n})_n$ converges to $\alpha(\vec{p})r + \vec{\beta}(\vec{p}) \cdot \vec{l}$ uniformly on $\mathcal{B}$. Let $\eta>0$. There exists   $N (\eta) \in \mathbb{N}$ such that, for any $n \geq N (\eta)$, $\vert  (\partial_r T_{n}) (\vec{p}, 0,\vec{0}) -\alpha(\vec{p}) \vert < \eta$ and  $\frac{\varepsilon}{2^{n-1}} <\eta$. Then, for any $n> N(\eta)$, for any $(r,\vec{l}) \in \mathcal{B}$,
 \begin{align*}
 \vert &  (\partial_r T_{n}) (\vec{p},r,\vec{l})- \alpha(\vec{p})   \vert \leq \vert     (\partial_r T_{n}) (\vec{p},r,\vec{l})- (\partial_r T_{N(\eta)})(\vec{p},\rho^{n-N(\eta)}r, \rho^{n-N(\eta)}\vec{l})  \vert \\
 &+   \vert  (\partial_r T_{N(\eta)})(\vec{p},0, \vec{0})  -\alpha(\vec{p})  \vert+  \vert (\partial_r T_{N(\eta)})(\vec{p},\rho^{n-N(\eta)}r, \rho^{n-N(\eta)}\vec{l}) - (\partial_r T_{N(\eta)})(\vec{p},0, \vec{0})    \vert\\
& <2\eta+  \vert (\partial_r T_{N(\eta)})(\vec{p},\rho^{n-N(\eta)}r, \rho^{n-N(\eta)}\vec{l}) - (\partial_r T_{N(\eta)})(\vec{p},0, \vec{0})    \vert.
 \end{align*}
 The function $\partial_r T_{N(\eta)}$ is continuous  in $(0,\vec{0})$ and there exists $M(\eta)>N(\eta)$, such that , for any $n> M(\eta)$, for any $(r,\vec{l}) \in \mathcal{B}$,
 \begin{equation*}
  \vert  (\partial_r T_{n}) (\vec{p},r,\vec{l})- \alpha(\vec{p})   \vert <3 \eta.
 \end{equation*}
 We deduce that $(\partial_r T_{n})_n$ converges uniformly on $\mathcal{B}$ to the constant $\alpha(\vec{p})$. A similar result holds for the partial derivatives with respect to $l_j$. As $T_{n} (\vec{p},0,\vec{0})=0$,   $T_{n} (\vec{p},r,\vec{l})$ converges uniformly to $\alpha(\vec{p})r + \vec{\beta}(\vec{p}) \cdot \vec{l}$  on $\mathcal{B}$.  $w_{0,0}^{(n)}(\vec{p},e_{n,\infty}(\vec{p},0),0,\vec{0})= - \rho e_{n+1,\infty}(\vec{p},0)$ converges to zero, and therefore, $w_{0,0}^{(n)}(\vec{p},e_{n,\infty}(\vec{p},0),r,\vec{l})$ converges uniformly on $\mathcal{B}$ to $\alpha(\vec{p})r + \vec{\beta}(\vec{p}) \cdot \vec{l}$.
\end{proof}
\vspace{2mm}

\begin{remark}
   $\vert \alpha(\vec{p}) -1 \vert \leq  \gamma + \varepsilon \ll \mu$, and   $\vert \vec{\beta}(\vec{p}) + \vec{p}/m \vert   \leq  \gamma + \varepsilon \ll \mu$. Therefore, the ``effective'' momentum $-m \vec{\beta}(\vec{p})$ stays close to  $\vec{p}$ and  the renormalization map can be iterated indefinitely.\end{remark}

\noindent We set 
\begin{equation}
z_{\infty}(\vec{p}):= \rho_0 e_{0,\infty}(\vec{p},0).
\end{equation} 
\vspace{2mm}
The next lemma states that   $z_{\infty}(\vec{p})$ is an eigenvalue of $H(\vec{p})$.
\begin{lemma}
\label{ee}
  $z_{\infty}(\vec{p})$ is a non-degenerate eigenvalue of  $H(\vec{p})$. If  $\vec{  p} \in \mathbb{R}^3 \cap U[\vec{p}^*]$,  $z_{\infty}(\vec{p})$ lies  at the bottom of the spectrum of $H( \vec{p})$, i.e.
\begin{equation}
z_{\infty}(\vec{p})= \inf \sigma( H(\vec{p})).
\end{equation}
\end{lemma}
\vspace{2mm}

The vector $ \Omega $ is  non-zero  eigenvector of $H_{\text{Fix}}$ with eigenvalue $0$.  The Feshbach-Schur theorem \ref{th1} asserts  that we can  find an eigenvector  of $H(\vec{p})$ corresponding to    $z_{\infty}(\vec{p})$ with the help of the auxiliary operators $Q_{\chi}$  introduced in \eqref{QQ}. We consider the  sequence of vectors
\begin{equation}
 \Psi_{n}(\vec{p}):=Q_{-1}(\vec{p}) \Gamma_{ \rho_0}^{*}  Q_{0}(\vec{p}) \Gamma_{ \rho}^{*}\cdots\Gamma_{ \rho}^{*} Q_{n}(\vec{p}) (    \downarrow   \otimes    \Omega) ,
 \end{equation}
where
\begin{equation}
\label{Qm}
Q_{n}(\vec{p}):=Q_{\chi_{\rho}(H_f)} \left( H (\underline{w}^{(n)}(\vec{p},e_{(n,\infty)}(\vec{p}))), W_{0,0}(w_{0,0}^{(n)}(\vec{p},e_{(n,\infty)}(\vec{p}))) \right),
\end{equation}
and
\begin{equation}
\label{Q-1}
Q_{(-1)}(\vec{p}):=Q_{\bm{\chi}}  \left( H (\vec{p},z_{\infty}(\vec{p})),  H (\vec{p},z_{\infty}(\vec{p})) - \lambda_0 H_I\right).
\end{equation}
\vspace{2mm}

\begin{lemma}
\label{447}
Under the conditions of Lemma \ref{44}, the limit
\label{psis}
\begin{equation}
\Psi_{ \infty}(\vec{p}):= \underset{n \rightarrow \infty}{\lim} \Psi_{n}(\vec{p}) 
\end{equation}
exists, is non zero, and belongs to  $\mathrm{ker}[ H(\vec{p},z_{\infty}(\vec{p}))]$. The convergence of  $\Psi_{n}(\vec{p}) $ to $\Psi_{ \infty}(\vec{p})$ is uniform in $\vec{p} \in U[\vec{p}^*]$. \end{lemma}
\vspace{2mm}

 The detailed proof  of the assertion -- ``$\rho_{0} e_{0,\infty}(\vec{p})$ is an eigenvalue of $H(\vec{p})$'' -- in Lemma \ref{ee} is similar to  \cite[Theorem 3.12]{BaChFrSi03_01}, and we refer the reader to this paper for details. The proof of the non-degeneracy  of $z_{\infty}(\vec{p})$ is a direct consequence of  \cite[Theorem 2.1]{HaHe12_01}. Furthermore, the fact that $z_{\infty}(\vec{p})$ lies at the bottom of the spectrum if $\vec{p} \in U[\vec{p}^*] \cap \mathbb{R}^3$, is, in all points, similar to the proof of \cite[Theorem 2.1 (i)]{HaHe11_01}. The proof of Lemma \ref{psis} is carried out  in detail in  \cite{ BaChFrSi03_01}.

\section{Proof of  Theorem \ref{main1} }
\label{fin}

\subsection{Proof of the main theorem}
 The proof of Theorem \ref{main1}  follows  by  gathering   the results of Lemmas \ref{anainit}, \ref{preserved},  \ref{contr},   \ref{ee} and \ref{447}.  Lemmas \ref{anainit} and  \ref{contr} show that the map $(\vec{  p},z) \mapsto H(\underline{w}^{(n)}(\vec{p},z) )\in \mathcal{B}(\mathcal{H}_{\text{red}})$ is analytic on $U[\vec{p}^*] \times D_{\mu/2}$ for all $n \in \mathbb{N} \cup \lbrace 0 \rbrace$. We deduce from  Lemmas \ref{compo} and  \ref{lm:bihol}  that  $E^{-1}_{\rho,n} (\vec{ p},z)$ is analytic on $U[\vec{p}^*] \times D_{\mu/2}$, for any $n \geq 0$. As a composition of analytic maps, $e_{(0,n)}(\cdot,0)$ defined in \eqref{emn1} is analytic on $U[\vec{p}^*]$ for any $n \geq 0$. The uniform convergence of the sequence  of functions  $(e_{(0,n)}(\cdot,0))_n$ to $e_{(0,\infty)}$ on  $U[\vec{p}^*]$ shows that  $e_{(0,\infty)}$, and therefore $z_{\infty}$, are analytic on $U[\vec{p}^*]$. 

The analyticity of $\vec{  p} \mapsto \Psi_{ \infty}(\vec{  p})$ is proven in a similar way. To shorten   notations, we set $H^{(n)}(\vec{p}):=H (\underline{w}^{(n)}(\vec{p},e_{(n,\infty)}(\vec{p})))$ and $W_{\geq 1}^{(n)}(\vec{p}):=W_{\geq 1}(\underline{w}^{(n)}(\vec{p},e_{(n,\infty)}(\vec{p})))$,  for any $n\geq 0$.  Lemmas  \ref{compo}  and \ref{contr} imply that $\vec{p} \mapsto H^{(n)}(\vec{p}) \in \mathcal{B}(\mathcal{H}_{\text{red}})$ and   $\vec{p} \mapsto W_{\geq 1}^{(n)}(\vec{p}) \in \mathcal{B}(\mathcal{H}_{\text{red}})$ are analytic on $U[\vec{p}^*]$.  From  \eqref{Qm} and the definition of $Q_{\chi}$ in \eqref{QQ},  we deduce that 
\begin{equation*}
  Q_{n}(\vec{p})= \chi_{\rho}(H_f) - \overline{\chi}_{\rho} (H_f)  [ H^{(n)}(\vec{p})]^{-1}_{\text{Ran}(\overline{\chi }(H_f))} \overline{\chi}_{\rho} (H_f)   W_{\geq 1}^{(n)}(\vec{p})   \chi_{\rho}(H_f) \in \mathcal{B}(\mathcal{H}_{\text{red}})
\end{equation*}
 is analytic on $U[\vec{p}^*]$,  for any $n \geq 0$,  being a product of bounded analytic operators. Furthermore,
   \begin{equation*}
Q_{-1}(\vec{p})= \bm{\chi} - \lambda_0 \overline{\bm{\chi}}  [ H_{\overline{\bm{\chi}}}  (\vec{p},z_{\infty}(\vec{p}))]_{\text{Ran}(\overline{\bm{\chi}})}^{-1}   \overline{\bm{\chi}}    H_I \bm{\chi}  \in \mathcal{B}(\mathcal{H})
\end{equation*}
is  analytic in $\vec{p}$. Indeed, the Neumann expansion of  $\overline{\bm{\chi}} [ H_{\overline{\bm{\chi}}}  (\vec{p},z_{\infty}(\vec{p}))]_{\text{Ran}(\overline{\bm{\chi}})}^{-1}  (H_{f}+ \rho_0)^{1/2} \overline{\bm{\chi}}$ converges uniformly on $U[\vec{p}^*]$ (see the proof of Lemma \ref{first}), and  the  bounded operator $[(H_f+\rho_0)   (H_{0}  (\vec{p},z_{\infty}(\vec{p})))]_{\text{Ran}(\overline{\bm{\chi}})}^{-1} $ is analytic on $U[\vec{p}^*]$; See  the proof of Lemma \ref{anainit}.  Therefore, $\Psi_{n}(\vec{ p} )$ is analytic on $U[\vec{p}^*]$, for any $n \geq 0$. The uniform convergence of $(\Psi_{n})_{n}$ to $\Psi_{ \infty}$ on $U[\vec{p}^*]$ completes the proof.

\begin{appendix}

\section{Lemmas about holomorphic functions}
\label{Hoo}
In this appendix, we gather some useful theorems for  Banach space-valued and complex-valued analytic functions of several  variables. The reader is referred to \cite{Vlad}, \cite{Ru}, \cite{Kato} and \cite{BarlDrag} for more detailed expositions. Let $U$ be a connected open set in $\mathbb{C}^n$.
\begin{definition} [Analytic function]
Let $n \geq 1$ and let $X$ be a Banach space. A function $f: U  \subset \mathbb{C}^{n} \rightarrow X$ is  said to be analytic on $U$ (or strongly analytic) if, for any $u=(u_1,\dots ,u_n) \in U$, there exists a neighborhood $\mathcal{V}(u)\subset U$ of $u$,  such that 
\begin{equation}
f(z)= \sum_{\alpha_1,\dots,\alpha_n \geq 0} x_{\alpha_1\dots \alpha_n} (z_1-u_1)^{\alpha_1} \cdots(z_n-u_n)^{\alpha_n} ,
\end{equation}
for any $z=(z_1,\dots,z_n) \in \mathcal{V}(u)$, and where  $ x_{\alpha_1\dots\alpha_n} \in X$.
\end{definition}
It is possible to show that strong and weak analyticity are equivalent for Banach space-valued holomorphic functions (see e.g. \cite{Ru} for functions of one variable, and \cite{BarlDrag} for a generalization to several complex variables).
\begin{theorem}
\label{hahn}
Let  X  be a Banach space, and $X'$ its dual. Let $f:U \subset \mathbb{C}^{n} \to  X$ be a function. Then the following assertions are equivalent:
\begin{enumerate}[(i)]
\item $f$ is analytic on $U$;
\item For any $\Phi \in X'$, $\Phi \circ f$ is  analytic on $U$.
\end{enumerate}
\end{theorem}

Theorem \ref{hahn} is  useful,  since we can use the standard results for complex-valued analytic functions  to investigate the analyticity of Banach space-valued  functions. In particular, for functions $f: U  \subset \mathbb{C}^{n} \rightarrow B(\mathcal{H})$ - where $B(\mathcal{H})$ is the Banach space of bounded operators on the Hilbert space $\mathcal{H}$-  it is sufficient to check that $\langle \phi \vert f(\cdot)  \psi \rangle$ is analytic for any $ \phi, \psi   \in \mathcal{H}$. 
A theorem due to Hartogs says that a complex-valued function of several variables is analytic if and only if it is holomorphic with respect to each variable individually. This theorem establishes the equivalence between analyticity and complex differentiability (or holomorphy) for complex-valued functions of several variables. To find out whether a complex-valued function of  several variables is holomorphic on $U$ or not, we will use the generalization of Cauchy-Poincar\'e and Morera's theorems to several complex variables (see \cite{Vlad}):

\begin{theorem} [Cauchy-Poincar\'e]
If $f:U\subset \mathbb{C}^n  \rightarrow \mathbb{C}$ is holomorphic, then, for any $(n+1)-$dimensional bounded surface $V \subset U$ with class  $\mathrm{C}^1$ boundary such that $\partial V $is an $n$-dimensional piecewise-smooth surface, 
\begin{equation}
\int_{\partial V} f =0. 
\end{equation}
\end{theorem}
\vspace{1mm}

\begin{theorem}[Morera]
\label{morera}
 Let $f: U \subset \mathbb{C}^n \to \mathbb{C}$ be a continuous  function. We assume that
\begin{equation}
\int_{\partial V_k} f =0
\end{equation}
 for any   arbitrary surface $V_k=l_1 \times... \times l_{k-1} \times \Delta_k \times l_{k+1} \times... \times l_n$ included in $U$, where the $l_j$'s are class $\mathrm{C}^1$  curves with ends $z'_j$ and $z''_j$ in the $z_j$-plane, and $\Delta_k$ is a closed bounded simply connected domain in the $z_k$-plane with a piecewise-smooth boundary, $k=1,...,n$.  Then $f$ is holomorphic on $U$.
\end{theorem}

Finally, we state a well known convergence theorem for sequences of holomorphic functions.
\begin{lemma}
\label{conve}
Let $X$ a Banach space, and $(f_k)_{k \in \mathbb{N}}$ a sequence of holomorphic functions  $f_k:U \rightarrow X$, which converges  to a function $f$ uniformly on any compact subset $K \subset U$. Then $f$ is holomorphic on $U$.
\end{lemma}

\section{Wick ordered expression of the renormalized operator}
\label{Wicko}
 In Subsection \ref{Wicki}, we explained shortly why $\mathcal{R}_{\rho} (H(\underline{w}))$ can  be rewritten as a series of Wick monomials. In this appendix,  we give the exact expression of the new kernel $\hat{ \underline{w} }$ and complete the intuitive picture of Subsection \ref{Wicki}.  We use the notations of Paragraph \ref{notation}, and, to avoid a large increase of the length of our formulas,   we introduce
 \begin{align}
&r_{i} := \sum_{j=1}^{i-1} \Sigma [ \underline{\tilde{k}}_j^{(n_j)} ] +\sum_{j=i+1}^L \Sigma [ \underline{k}_j^{(m_j)} ] , \quad \vec{l}_{i} := \sum_{j=1}^{i-1} \vec{\Sigma} [ \underline{\tilde{k}}_j^{(n_j)} ] +\sum_{j=i+1}^L \vec{\Sigma} [ \underline{k}_j^{(m_j)} ], \notag \\
&\tilde{r}_{i} := \sum_{j=1}^{i} \Sigma [ \underline{\tilde{k}}_j^{(n_j)} ] +\sum_{j=i+1}^L \Sigma [ \underline{k}_j^{(m_j)} ] , \quad \vec{\tilde{l}}_{i} := \sum_{j=1}^{i} \vec{\Sigma} [ \underline{\tilde{k}}_j^{(n_j)} ] +\sum_{j=i+1}^L \vec{\Sigma} [ \underline{k}_j^{(m_j)} ]. \label{def_mu_L}
\end{align}
We remind the reader that
 \begin{equation}
 \label{Fi}
 F(\vec{p},z , r,\vec{l}):= \frac{ \overline{\chi}_{\rho}^{2} (r) }{ w_{0,0}(\vec{p},z,r, \vec{l})}.
 \end{equation} 

In order to get the expression of the new kernels, we have to investigate how the arguments of the kernels $w_{M_i,N_i}$ are modified by the pull-through of the annihilation and creation operators.  We set $p_i:=M_i-m_i$, $q_i:=N_i-n_i$  to follow the  notations already present in the literature.  Once we have contracted   $p_1+....+p_L$ creation  operators with  $q_1+...+q_L$ annihilation operators and pulled through the $m_1+...+m_L$ uncontracted creation operators to the left and the $n_1+...+n_L$ uncontracted  annihilation operators to the right, the  third and the fourth arguments of the  operators $w_{m_i+p_i,n_i+q_i}(\vec{p},z,H_f,\vec{P}_f,\underline{K}^{(m_i+p_i,n_i+q_i)}_{i})$  and $F(\vec{p},z,H_f,\vec{P}_f)$ are shifted by  a certain amount, which depends on their position $i$. Namely, they are modified by the uncontracted  creation operators that were originally sitting on their right hand-side and have been pulled-through to the left, and the uncontracted  annihilation operators that were originally sitting on their left hand-side, and have been pulled-through to the right. This leads to a shift 
$$H_f \rightarrow H_f+ \sum_{j=i+1}^{L} \Sigma[\underline{k}^{(m_j)}] +   \sum_{j=1}^{i-1} \Sigma[\underline{\tilde{k}}^{(n_j)}]= H_f + r_i$$ for the  third argument of $w_{m_i+p_i,n_i+q_i}$, and a shift 
$$H_f \rightarrow H_f+ \sum_{j=i+1}^{L} \Sigma[\underline{k}^{(m_j)}] +   \sum_{j=1}^{i} \Sigma[\underline{\tilde{k}}^{(n_j)}]=H_f + \tilde{r}_i$$ for the third argument of the operator-valued function  $F(\vec{p},z,H_f,\vec{P}_f)$ sitting directly on the right of $w_{m_i+p_i,n_i+q_i}(\vec{p},z,H_f,\vec{P}_f)$. The shifts for $\vec{P}_f$ are given by the same formulas, with $\Sigma$ replaced by $\vec{\Sigma}$.  Thanks to Formula \eqref{wiwi}, the contracted part  can be rewritten as the vacum expectation value of an operator $\mathscr{V}_{\underline{m},\underline{n}}^{\underline{m+p},\underline{n+q}}(\vec{p},z,r,\vec{l},\underline{K}^{(\underline{m},\underline{n})})$, whose expression after rescaling is   given  inside the brackets of Equation \eqref{VVm}.  
 \vspace{2mm}
 
 \noindent We are ready to state the exact expression of the new kernel $\underline{\hat{w}}=(\hat{w}_{M,N}) $ which satisfies $H(\underline{\hat{w}})= \mathcal{R}_{\rho} (H(\underline{w}))$. We introduce
  \begin{align} 
W_{p,q}^{m,n}(\vec{p}, z , r , \vec{l} , \underline{K}^{(m,n)}) := \mathds{1}_{H_f \leq 1} \int_{\underline{B}_{1}^{p+q}}  & b^{*}  (\underline{x}^{p})  w_{m+p, n+q}(\vec{p},z , H_f + r , \vec{P}_f + \vec{l} , \underline{k}^{(m)},\underline{x}^{(p)}, \underline{\tilde{k}}^{(n)}, \underline{\tilde{x}}^{(q)})  \notag \\
& b  (\underline{\tilde{x}}^{q}) d\underline{X}^{(p,q)}   \mathds{1}_{H_f \leq 1}. \label{Wii}
\end{align}
We also define the function
   \begin{align}
  & V_{\underline{m,p,n,q}} ( \vec{p},z,r,\vec{l},\underline{K}^{(M,N)}) :=   \chi_{1}( r+ \tilde{r}_0) \chi_{1}( r+ \tilde{r}_L) \notag \\
  & \quad \Big \langle \Omega \big \vert  \prod_{i=1}^{L-1} \Big(  W_{p_i,q_i}^{m_i,n_i} (\vec{p},z,\rho(r+r_i),\rho(\vec{l}+ \vec{l}_i ), \rho \underline{K}_{i}^{(m_i,n_i)})  F(\vec{p}, z , H_f+ \rho(r+ \tilde{r}_i), \vec{P}_f+  \rho(\vec{l} + \vec{\tilde{l}}_i)) \Big) \notag \\
  & \quad \qquad W_{p_L,q_L}^{m_L,n_L} (\vec{p},z,\rho(r+r_L),\rho(\vec{l}+ \vec{l}_L), \rho \underline{K}_{L}^{(m_L,n_L)})  \Omega \Big \rangle,   \label{VVm}
  \end{align}
  \vspace{1mm}
  
  \noindent  where we have set $\underline{m} := (m_1 , \cdots , m_L )$, $\underline{n} := (n_1 , \cdots , n_L )$, $\underline{p} := (p_1 , \cdots , p_L )$, $\underline{q} := (q_1 , \cdots , q_L )$, $\underline{m,p,n,q}:=(m_1,p_1,n_1,q_1, \cdots , m_L,p_L,n_L,q_L)$, $$\underline{K}_{i}^{(m_i,n_i)}:=(\underline{k}_{i}^{(m_i)}, \underline{\tilde{k}}_{i}^{(n_i)}) \quad  \text{ and } \quad \underline{k}_{i}^{(m_i)}:=(\underline{k}_{m_1+...+m_{i-1}+1},...,\underline{k}_{m_1+...+m_{i}})$$ to simplify notation. We introduce  the constants \indent
\begin{equation}
C_{\underline{m},\underline{n}}^{\underline{m+p},\underline{n+q}} :=\prod_{i=1}^{L} \left( \begin{array}{c} m_i+p_i\\p_i \end{array}\right)  \left( \begin{array}{c} n_i+q_i\\q_i \end{array}\right).
 \end{equation}
Setting $(\vec{p},z): = E_\rho^{-1}(\vec{p}, \zeta )$,  the sequence $(\hat{w}_{M,N})$ is given by the following expressions:  For $M+N \ge 1$,
  \begin{align}
&  \hat{w}_{M,N}[ \vec{p},\zeta, r, \vec{l},\underline{K}^{(M,N)}] \notag \\
  &=     \sum_{L=1}^{\infty} (-1)^{L-1}  \rho^{\frac{3}{2}(M+N)-1} \underset{ \tiny \begin{array}{c} \underline{m, p,n,q}  \\ m_1+...+m_L=M, \\ n_1+...+n_L=N \\  m_i+n_i+p_i+q_i \geq 1 \end{array}}{\sum}   C_{\underline{m},\underline{n}}^{\underline{m+p},\underline{n+q}}  \text{ } V^{\mathrm{sym}}_{\underline{m,p,n,q}} (\vec{p},z,r,\vec{l},\underline{K}^{(M,N)}).   \label{whatmn}
  \end{align}
The notation $ f^{\mathrm{sym}}\left(\underline{K}^{(M,N)} \right)$ appearing in \eqref{whatmn} denotes the symmetrization of $f$ w.r.t. the variables $\underline{k}^{(M)}$ and $\underline{\tilde{k}} \phantom{}^{(N)}$, that is
\begin{equation}
\begin{split}
&f^{\mathrm{sym}}\left(\underline{K}^{(M,N)} \right) := \frac{1}{M!N!}\sum_{\pi\in S_M}\sum_{\tilde{\pi}\in S_N} f\left( \underline{k}_{\pi(1)} , \dots , \underline{k}_{\pi(M)} , \underline{\tilde{k}}_{\tilde{\pi}(1)} , \dots , \underline{\tilde{k}}_{\tilde{\pi}(N)} \right).
\end{split}
\end{equation} 
\vspace{1mm}

\noindent Finally, for $(M,N) = (0,0)$,
 \begin{equation}
\label{w_0,0^eff}
 \hat{w}_{0,0}(\vec{p},\zeta,r,\vec{l})= \rho^{-1} w_{0,0}(\vec{p},z, \rho r , \rho \vec{l})+ \rho^{-1} \sum_{L=2}^{\infty} (-1)^{L-1} \underset{\underset{p_i+q_i \geq 1}{\underline{p} , \underline{q}}}{\sum}  V_{\underline{0,p,0,q}} (\vec{p},z,r,\vec{l}),
 \end{equation} 
where 
\begin{align}
& V_{\underline{0,p,0,q}} (\vec{p},z,r,\vec{l}) \notag \\
&:= \chi_{1}^{2}(r) \Big \langle \Omega \big \vert  \prod_{i=1}^{L-1} \Big(  W^{0,0}_{p_i,q_i} (\vec{p}, z,\rho r,\rho \vec{l} )   F (\vec{p},z , H_f+ \rho r, \vec{P}_f+ \rho \vec{l}) \Big)  W^{0,0}_{p_L,q_L} (\vec{p}, z,\rho r,\rho \vec{l} )    \text{ }\Omega \Big  \rangle. \label{def_DL_2}
\end{align}
\indent

\section{Proof of Lemma  \ref{polyd}}
\label{polydlem}
  To render the  structure of the proof of Lemma \ref{polyd} clearer, we subdivide it  into three steps.  All the positive  constants of order one appearing in the estimates  are denoted by  the capital letter $C$ if they are independent of the parameters $\rho_0$, $\lambda_0$ and $\mu$.
\subsection{Step 1: Explicit  expressions for the kernels}
We start from \eqref{expre}.  We set
 \begin{align*}
H^{(0)}(\vec{p},z):= \mathcal{S}_{\rho_{0}} \Big ( \big \langle & \big( H_f + \frac{\vec{P}_{f}^{2}}{2m} - \frac{\vec{p}}{m}   \cdot  \vec{P}_f +   \lambda_0 \bm {\chi}   H_I  \bm {\chi} \\[4pt]
& - \lambda_{0}^{2}  \bm {\chi}  H_I  \overline{\bm {\chi}}   (H_{ \overline{\bm {\chi} }}(\vec{p}, \rho_0 z))^{-1} \overline{\bm {\chi}}  H_I  \bm {\chi}  - \rho_0 z \mathds{1}  \big)  \mathds{1}_{H_f \leq \rho_0} \big \rangle_{\downarrow} \Big).
\end{align*}
One has that
\begin{align}
 H^{(0)}(\vec{p},z) = &  \hat{w}_{0,0}(\vec{p},z,H_f,\vec{P}_f) \mathds{1}_{H_f \leq 1} \notag \\[4pt]
 &+   \lambda_0 \mathcal{S}_{\rho_{0}} \left(  \big \langle \bm {\chi}   H_I  \bm {\chi}\big \rangle_{\downarrow}\right)  - \lambda_{0}^{2}  \mathcal{S}_{\rho_{0}} \left( \big \langle  \bm {\chi}  H_I  \overline{\bm {\chi}}   (H_{ \overline{\bm {\chi} }}(\vec{p},\rho_0 z))^{-1} \overline{\bm {\chi}}  H_I  \bm {\chi}  \big \rangle_{\downarrow} \right),\label{rescler}
\end{align}
where  $\hat{w}_{0,0}(\vec{p},z,r,\vec{l}):=r + \rho_{0} \frac{\vec{l}^2}{2m} - \frac{1}{m}  \vec{p} \cdot  \vec{l} - \rho_0 z$. An obvious calculation shows that $$\underset{(\vec{p},z) \in U[\vec{p}^*] \times D_{\mu/2}}{\sup} \|\hat{w}_{0,0}(\vec{p}, z , r , \vec{l})-\hat{w}_{0,0}(\vec{p}, z , 0 , \vec{0})  - ( r -  m^{-1} \vec{p} \cdot \vec{l} ) \|^{\sharp} \leq \frac{\sqrt{3} \rho_{0}}{m},$$ where we used that $\vert \vec{l}  \vert \leq 1$. We remind the reader that
\begin{equation*}
 \| \hat{w}_{0,0} \|^{\sharp} := | \hat{w}_{0,0}( 0 , \vec{0} ) | + \| \partial_{r}  \hat{w}_{0,0} \|_{\infty} + \sum_{i=1}^{3} \| \partial_{l_i}  \hat{w}_{0,0} \|_\infty . \end{equation*}
  The second term on the right hand side of \eqref{rescler} is  already cast into a sum of  Wick monomials. A straightforward calculation making use of the Pull-through formula to restrict the integration range  implies that 
\begin{equation}\label{eq:a2}
 \lambda_0 \mathcal{S}_{\rho_{0}} \left(  \big \langle \bm {\chi}   H_I  \bm {\chi}\big \rangle_{\downarrow}\right)  =- i  \lambda_0   \rho_{0}   \int_{\underline{B}_{1}} \vert \vec{k} \vert^{\frac12}  \vec{\epsilon} (\underline{k}) \cdot \vec{e}_z \text{ } \chi(H_f) \left( b(\underline{k}) - h.c. \right)  \chi(H_f) \, d \underline{k} 
\end{equation}
The associated kernels are independent of $\vec{p},z$ and $\vec{l}$, and are given by 
\begin{equation}\label{eq:a3}
\hat{w}_{1,0}( r, \underline{k})= -  \hat{w}_{0,1}( r, \underline{k})=  i \lambda_0  \rho_{0} \vert \vec{k} \vert^{1/2}  \text{ } \vec{\epsilon} (\underline{k}) \cdot \vec{e}_z \text{ } \chi(r+ \vert k \vert)  \chi(r). 
\end{equation}
 Their norm  $\| \cdot \|_{\frac{1}{2}}$ is of order $ \lambda_0  \rho_{0}$. The analysis of  the third term on the right-hand side of  \eqref{rescler}  requires more work.  The procedure is similar to the discussion of    Paragraph \ref{Wicki}. We write down the Neumann expansion for  $   \left[ H_{\overline{\bm{\chi}}}(\vec{p}, \rho_0 z) \right]^{-1}_{ \text{Ran}(\overline{\bm{\chi}})}$ and  Wick-order  it to rewrite $$  \mathcal{S}_{\rho_{0}} \left( \big \langle  \bm {\chi}  H_I  \overline{\bm {\chi}}    \left[ H_{\overline{\bm{\chi}}}(\vec{p}, \rho_0 z) \right]^{-1}_{ \text{Ran}(\overline{\bm{\chi}})} H_I  \bm {\chi}  \big \rangle_{\downarrow} \right)$$ as a series of Wick monomials.   We use the formulas  in Appendix \ref{Wicko} and get that 
 \begin{equation}
- \lambda_{0}^{2}   \mathcal{S}_{\rho_{0}} \left( \big \langle  \bm {\chi}  H_I  \overline{\bm {\chi}}    \left[ H_{\overline{\bm{\chi}}}(\vec{p}, \rho_0 z) \right]^{-1}_{ \text{Ran}(\overline{\bm{\chi}})} H_I  \bm {\chi}  \big \rangle_{\downarrow} \right)=H(\underline{\tilde{w}}(\vec{p},z)),
\end{equation}
where, for $M,N \ge 0$,
\begin{align}
& \tilde{w}_{M,N}(\vec{p},z,r,\vec{l},\underline{K}^{(M,N)}) \notag \\
&= \sum_{L=2}^{\infty} (-1)^{L-1} \rho_{0}^{\frac{3}{2}(M+N)-1} \underset{ \tiny \begin{array}{c} \underline{m,p,n,q} , \\ m_1+\cdots+m_L=M, \\ n_1+\cdots+n_L=N \\ m_i+n_i+p_i+q_i =1 \end{array}}{\sum}   V^{\text{sym}}_{\underline{m,p,n,q}}(\vec{p},\rho_0 z,r, \vec{l}, \underline{K}^{(M,N)})  , \label{compliq}
\end{align}
with  
\begin{align}
V_{\underline{m,p,n,q}} &( \vec{p},\rho_0 z,r,\vec{l},\underline{K}^{(M,N)}) :=  \chi_{1}( r+ \tilde{r}_0)   \chi_{1}( r+ \tilde{r}_L)  \Big \langle \downarrow \otimes \Omega \Big \vert  \prod_{i=1}^{L-1} \Big(  W_{p_i,q_i}^{m_i,n_i} ( \rho_0 \underline{K}_{i}^{(m_i,n_i)})  \notag \\
   Ê&   F(\vec{p}, \rho_0 z , H_f+ \rho_0(r+ \tilde{r}_i), \vec{P}_f+  \rho_0(\vec{l} + \vec{\tilde{l}}_i)) \Big)  W_{p_L,q_L}^{m_L,n_L} ( \rho_0 \underline{K}_{L}^{(m_L,n_L)})   \text{ }  \downarrow \otimes\Omega \Big \rangle .\label{Vmpnq}
\end{align}
\vspace{2mm}

\noindent The function $F$ is given by 
\begin{equation}
\label{5.24}
F(\vec{p}, z , r, \vec{l})= P_{\downarrow} \otimes b^{-1}_{1}(\vec{p},z , r ,  \vec{l}) \overline{\chi}_{\rho_0}^{2} (r) + {P}_{\uparrow} \otimes  b^{-1}_{2}(\vec{p},z , r ,  \vec{l}),
\end{equation}
with
\begin{align}
\label{B1}
b_1(\vec{p},z,r, \vec{l})&= \Big ( r + \frac{\vec{l}^2}{2m} - \frac{\vec{ p}} {m} \cdot \vec{l} -z \Big ) 
\mathds{1}_{r \geq 3 \rho_0/4} \\
\label{B2}
b_2(\vec{p},z,r, \vec{l}) &= r + \frac{\vec{l}^2}{2m}  + \omega_0- \frac{\vec{p}}{m} \cdot \vec{l} -z.
\end{align}
\vspace{2mm}

\noindent The sum in \eqref{compliq} is only carried out for $m_i+p_i+n_i+q_i=1$ and the operators $W_{p_i,q_i}^{m_i,n_i}$  do not depend on $(\vec{p},z)$. In \eqref{Vmpnq},
\begin{align}
 & W_{p_i ,q_i}^{m_i,n_i} (  \underline{K} ^{(m_i ,n_i )} ) \notag \\
 &=  \int_{\underline{B}_{1}} b^{*} ( \underline{x}^{(p_i)}) \text{ } w_{m_i+p_i,n_i+q_i} (  \underline{k}^{(m_i)},\underline{x}^{(p_i)},\underline{\tilde{k}}^{(n_i)},\underline{\tilde{x}}^{(q_i)} )  \text{ } b ( \underline{\tilde{x}}^{(q_i)}) d \underline{X}^{(p_i,q_i)}.  \label{Wmp}
\end{align}
 As $H_I$ is linear in annihilation and creation operators, the only non-zero   kernels  are the $\mathbb{C}^2$-valued functions $w_{0,1}$ and  $w_{1,0}$, given by
\begin{equation}
w_{0,1} ( \underline{k} ) =- w_{1,0} ( \underline{k} )=i  \lambda_0 \vert \vec{k} \vert^{1/2}  \vec{\epsilon} (\underline{k})  \cdot \vec{\sigma}  \text{ }\mathds{1}_{ \vert \vec{k} \vert \leq 1}.
\end{equation}
Therefore, $W_{p ,q}^{m,n} $ is non-zero if, and only if, 
\begin{equation*}
m+p=0, \text{ }n+q=1 \qquad \text{or} \qquad m+p=1, \text{ }n+q=0.
\end{equation*}
If $m+n=1$,
\begin{equation}
\label{W00mn}
W_{0,0}^{m,n} (\rho_{0} \underline{k} ) =  (-1)^{m} i    \lambda_{0} \rho_{0}^{1/2} \vert \vec{k} \vert^{1/2}   \vec{\epsilon} (\underline{k})  \cdot \vec{\sigma} \,  \mathds{1}_{|\vec{k}| \le 1}( \vec{k} ).
\end{equation}
If $m+n=0$,
\begin{equation}
W_{1,0}^{0,0}= -i \lambda_0 \int_{ \underline{B}_1} \vert \vec{k} \vert^{1/2}  \text{ }  \vec{\epsilon}(\underline{k}) \cdot \vec{\sigma} \text{ }  b^{*} (\underline{k}) d \underline{k}, 
\end{equation}
and
\begin{equation}
W_{0,1}^{0,0} = i \lambda_0 \int _{ \underline{B}_1} \vert \vec{k} \vert^{1/2}  \text{ }  \vec{\epsilon}(\underline{k}) \cdot \vec{\sigma} \text{ } b (\underline{k}) d \underline{k}.
\end{equation}
Note that $\mathds{1}_{|\vec{k}| \le 1}( \vec{k} )$ appears in \eqref{W00mn}, because the term $\chi(r+\tilde{r}_0)$ in \eqref{Vmpnq} forces $\vert \vec{k} \vert$ to be smaller than  one. 

\subsection{Step 2: Bounds on the  norm of  kernels}
We construct the sequence of kernels $\underline{\hat{w}}=(\hat{w}_{M,N})$ by  setting $\hat{w}_{0,0}(\vec{p},z,r,\vec{l})=r + \rho_{0} \frac{\vec{l}^2}{2m} - \frac{1}{m}  \vec{p} \cdot  \vec{l} - \rho_0 z$,   $\hat{w}_{1,0}$ and $\hat{w}_{0,1}$ as given in \eqref{eq:a3}, and $\hat{w}_{M,N}=0$ for $M+N>1$. 
$H^{(0)}(\vec{p},z)=H((\underline{\hat{w}} + \underline{\tilde{w}})(\vec{p},z))$, with $\tilde{w}_{M,N}$ as in \eqref{compliq}. In order to show that the Wick-ordered operator  $H(\underline{\hat{w}} + \underline{\tilde{w}})$ belongs to  a polydisc $\mathcal{B}(\gamma,\delta,\varepsilon)$, we need to investigate the convergence of the series in \eqref{compliq}. The assertion of the lemma   follows if this series  converges uniformly  in norm $\|\cdot\|^{\sharp}$ on $U[\vec{p}^*] \times D_{\mu/2}$, and, can be bounded by  an arbitrary small constant after an adequate tuning of the coupling  $\lambda_0$.  We  first bound $\| V_{\underline{m,p,n,q}} \|^{\sharp}$. To avoid a long proof, we only explicitly bound $\| V_{\underline{m,p,n,q}} \|_{\frac{1}{2}}$. The reader can check that the bounds are similar for the partial derivatives. Our proof works  if the ultraviolet cut-off function in the interacting Hamiltonian, $\mathds{1}(\vert \vec{k} \vert \leq 1)$, is replaced by an arbitrary cut-off function in $L^2 (\mathbb{R}^3)$ that  equals one in a neighborhood of  the origin. To shorten notations, we set $R_i := H_f + \rho_0 ( r + \tilde r_i )$, and $\vec{L}_i := \vec{P}_f + \rho_0 ( \vec{l} + \vec{\tilde{l}}_i )$.  We introduce $(H_f+ \rho_{0})^{-1/2}$ to the right and to the left of $W_{0,1}^{0,0}$ and $W_{1,0}^{0,0}$ in (\ref{Vmpnq}). Then,
\begin{align}
\label{Vmpnq2}
& V_{\underline{m,p,n,q}}(\vec{p}, \rho_0 z, r, \vec{l},\underline{K}^{(M,N)}) =  \rho_{0} \chi_{1}( r+ \tilde{r}_0)   \chi_{1}( r+ \tilde{r}_L)  \big \langle \downarrow \otimes \Omega  \big \vert  \notag \\
& \qquad \prod_{i=1}^{L-1} \Big[ (H_f+ \rho_{0})^{-\frac{1}{2}} W_{p_i,q_i}^{m_i,n_i} (\rho_{0} \underline{K}_{i}^{(m_i,n_i)})  (H_f + \rho_{0})^{-\frac{1}{2}} (H_f+ \rho_{0}) F(\vec{p},   \rho_0 z , R_i , \vec{L}_i ) \Big]  \notag \\
& \qquad (H_f + \rho_{0})^{-\frac{1}{2}}    W_{p_L,q_L}^{m_L,n_L} (\rho_{0} \underline{K}_{L}^{(m_L,n_L)})   (H_f + \rho_{0})^{-\frac{1}{2}}  \text{ } \downarrow \otimes \Omega  \big \rangle.
\end{align}
\vspace{2mm}

\noindent For $i \neq  L$, we give an upper bound for
\begin{equation}
A_{p_i,q_i}^{m_i,n_i}:= \|  (H_f+ \rho_{0})^{-\frac{1}{2}} W_{p_i,q_i}^{m_i,n_i}(\rho_{0} \underline{K}_{i}^{(m_i,n_i)})  (H_f+\rho_{0})^{-\frac{1}{2}}   (H_f+ \rho_{0}) F(\vec{p}, \rho_0 z , R_i , \vec{L}_i)  \|,
\end{equation}
 $m_i+n_i+p_i+q_i=1$. There are two different cases:
\begin{enumerate}[(i)]
\item  $m_i+n_i=1$. Then, 
\begin{align}
\label{5.45}
A_{0,0}^{m_i,n_i} & \leq  \lambda_0   \| F(\vec{p},  \rho_0 z, R_i , \vec{L}_i ) \| \vert \rho_{0} \vec{k}_i \vert^{1/2} \| \vec{\epsilon}_{\lambda_i} (\vec{k}_i) \cdot \vec{\sigma } \| \mathds{1}_{|\vec{k}_i| \le 1}( \vec{k}_i ) \notag \\
& \leq C    \frac{ \lambda_{0}}{\mu} | \vec{k}_i |^{1/2}  \mathds{1}_{|\vec{k}_i| \le 1}( \vec{k}_i ) \rho_{0}^{-1/2}. 
\end{align}

\item  $m_i+n_i=0$. Then,
\begin{equation}
\label{5.46}
A_{p_i,q_i}^{0,0} \leq  C \lambda_0 \rho_{0}^{-\frac12}   \| (H_f + \rho_{0}) F(\vec{p},   \rho_0 z , R_i , \vec{L}_i ) \|. 
\end{equation}
\end{enumerate}
We use the functional calculus outlined in Paragraph \ref{funcpar} to bound $\| (H_f + \rho_{0}) F(\vec{p},  \rho_0 z , R_i , \vec{L}_i ) \|$. It  is bounded by the sum of the constants $C_1$ and $C_2$, where
\begin{align*}
& C_1:=  \underset{r' \geq \frac{3 \rho_{0}}{4} ,  \vert \vec{l}' \vert \leq r'}{\sup} \left \vert \frac{ r' + \rho_{0}}{ r' + \rho_0( r + \tilde r_i ) + \frac{ \displaystyle  \vert \vec{l}' + \rho_0 ( \vec{l} + \vec{\tilde{l}}_i ) \vert^{2}}{2} - \frac{\vec{p}}{m} \cdot ( \vec{l}' + \rho_0 ( \vec{l} + \vec{\tilde{l}}_i )) -  \rho_0 z} \right \vert, \\
& C_2:=  \underset{r' \geq 0,  \vert \vec{l}' \vert \leq r'}{\sup} \left \vert \frac{  r' + \rho_{0}}{ \displaystyle   r' + \rho_0( r + \tilde r_i ) + \omega_0+ \frac{\vert \vec{l}' + \rho_0 ( \vec{l} + \vec{\tilde{l}}_i ) \vert^{2}}{2} - \frac{\vec{p}}{m} \cdot (\vec{l}' + \rho_0 ( \vec{l} + \vec{\tilde{l}}_i ) ) -  \rho_0 z} \right \vert .
\end{align*}
Using the assumptions   $\vert z \vert <   \mu   /2$ together with $| \vec{l} | \le r$ and $| \vec{\tilde{l}}_i | \le \tilde{r}_i$, we obtain that
\begin{equation*}
C_1 \leq  \underset{r' \geq \frac{3 \rho_{0}}{4} \text{ } }{\sup} \frac{ r' + \rho_{0}}{  \mu  r'   - \frac{ \mu \rho_{0}}{2}}, \qquad C_2 \leq  \underset{r' \geq 0 \text{ }}{\sup} \frac{ r' + \rho_{0}}{  \mu  r'  + \omega_0 - \mu \frac{\rho_{0}}{2}}.
\end{equation*}
 It follows  that  there is a constant $C>1$, such that $A_{p_i,q_i}^{0,0} \leq C \mu^{-1} \lambda_{0} \rho_{0}^{-\frac12}  $, and the norm $\|\cdot \|_{\frac{1}{2}}$  of $V_{\underline{m,p,n,q}} (\vec{p}, \rho_0 z , r, \vec{l}, \underline{K}^{(M,N)} )$ is   bounded by
\begin{align*}
&  \|V_{\underline{m,p,n,q}} (\vec{p}, \rho_0 z) \|_{\frac{1}{2}} \leq C^{L+1}   \left(\frac{ \lambda_{0}}{\mu \rho_{0}^{1/2}} \right)^{L} \rho_{0}.
\end{align*}
Similar calculations for the partial derivatives show that 
\begin{align*}
\big \| \partial_r V_{\underline{m,p,n,q}} (\vec{p}, \rho_0 z) \big \|_{\frac{1}{2}} &\leq  \mu^{-1}(L+1) C^{L+1} \left(\frac{ \lambda_{0}}{\mu \rho_{0}^{1/2}} \right)^{L} \rho_{0} ,\\
\big \| \partial_{l_k} V_{\underline{m,p,n,q}} (\vec{p}, \rho_0 z) \big \|_{\frac{1}{2}} &\leq \mu^{-1} (L-1) C^{L+1}   \left(\frac{ \lambda_{0}}{\mu \rho_{0}^{1/2}} \right)^{L} \rho_{0},
\end{align*}
which implies that
\begin{equation}
\label{bbbb}
\rho_{0}^{\frac{3}{2}(M+N)-1} \big \| V_{\underline{m,p,n,q}} (\vec{p}, \rho_0 z) \big \|^{\sharp} \leq 5 \mu^{-1} (L+1) C^{L+1}   \rho_{0}^{\frac32 (M+N)} \left ( \mu^{-1}\rho_{0}^{-\frac12}  \lambda_{0} \right )^{L}.
\end{equation}
\vspace{1mm}

\subsection{Step 3: Bound on the norms $\|\cdot \|_{\xi}^{\sharp}$}

\noindent By Plugging (\ref{bbbb}) into \eqref{compliq}, we finally obtain  that
\begin{align*}
\| \underline{\tilde{w}}(\vec{p},z) \|_{\xi , \ge 1}^{\sharp} & \leq 5  \mu^{-1} \sum_{L=2}^{\infty} (L+1) C^{L+1} \left ( 4 \mu^{-1} \rho_{0}^{-\frac12}  \lambda_{0} \right )^{L}  \left(  \left(\sum_{m  \geq 0} \left(   \rho_{0}^{\frac32}  \xi^{-1} \right)^{m } \right)^{2}-1 \right) .
\end{align*}
If we assume that  $\xi >   \rho_{0}^{ \frac32}   $, we get the bound
\begin{equation}
\label{B20}
\| \underline{\tilde{w}}(\vec{p},z) \|_{\xi , \ge 1}^{\sharp} \leq 10 \mu^{-1}  C \frac{\rho_{0}^{3/2} \xi^{-1}}{(1- \rho_{0}^{3/2} \xi^{-1} )^2}  \sum_{L=2}^{\infty} (L+1) \left(4 \mu^{-1} C  \lambda_{0}  \rho_{0}^{-\frac12}   \right)^{L}.
\end{equation}
If $0\leq \lambda_0< \lambda_c(\mu)$, with
\begin{equation*}
\lambda_c(\mu):=   \frac{ \mu \rho_{0}^{\frac12} }{4C},
\end{equation*}
the series in \eqref{B20} converges.  Setting $\alpha:=\lambda_0/\lambda_c<1$, we find that
\begin{equation}
\label{cb}
\| \underline{\tilde{w}} (\vec{p},z) \|_{ \xi , \ge 1}^{\sharp} \leq  10 \mu^{-1} C \frac{\rho_{0}^{3/2} \xi^{-1}}{(1- \rho_{0}^{3/2} \xi^{-1} )^2}   \frac{ 3 \alpha^2}{(1-\alpha)^2}.
\end{equation}
The right-hand side of \eqref{cb} can be made arbitrary small by tuning  $\lambda_0$. To conclude the proof, it is sufficient to check that $\| \tilde{w} _{0,0} (\vec{p},z) \|^{\sharp}$ can also be made  as small as we wish by an appropriate tuning of $\lambda_0$. To this end, we have to give an upper bound for  the norm of  
\begin{equation*}
   \rho_{0}^{-1}  \sum_{L=2}^{\infty} (-1)^{L-1}  \underset{\tiny  \begin{array}{c} \underline{p} , \underline{q} \\  p_i+q_i=1 \end{array}}{\sum} V_{\underline{0,p,0,q}} (\vec{p},\rho_0 z,r, \vec{l}).
\end{equation*}
 Some easy calculations lead  to 
 \begin{equation}
 \label{ineq33}
\big \|  V_{\underline{0,p,0,q}} (\vec{p}, \rho_0 z) \big \|^{\sharp} \leq 5 \mu^{-1} (L+1) C^{L+1}   \rho_{0} ( \lambda_{0} \mu^{-1} \rho_{0}^{-\frac12}  )^{L}.
\end{equation}
We deduce from (\ref{ineq33}) that
\begin{equation*}
\Big \| \rho_{0}^{-1}  \sum_{L=2}^{\infty} (-1)^{L-1}  \underset{ \tiny \begin{array}{c} \underline{p} , \underline{q} \\  p_i+q_i=1 \end{array}}{\sum} V_{\underline{0,p,0,q}} (\vec{p}, \rho_0 z)  \Big \|^{\sharp} \leq  5 \mu^{-1} C \sum_{L=2}^{\infty} (L+1)    (2 \mu^{-1} \lambda_{0} \rho_{0}^{-\frac12}   )^{L}.
 \end{equation*}
For $\lambda_{0} < \mu \rho_{0}^{1/2}/2$, the series on the right-hand side converges, and   can be made as small as we wish by an appropriate  tuning  of $\lambda_0$.  All the estimates are uniform in $(\vec{p},z) \in U[\vec{p}^*] \times D_{\mu/2}$, and  we deduce that $H(\underline{w}^{(0)})$  can belong to any polydisc $\mathcal{B}(\frac{ \sqrt{3} \rho_{0}}{m}+\gamma,\delta,\varepsilon)$ if  the coupling constant $\lambda_0$ is small enough.

\section{Completion of the proof of Lemma \ref{compo}}
\label{lemacom}
 \subsection{Proof of the uniform convergence in $(\vec{p},z) \in U[\vec{p}^*] \times D_{\mu/2}$}
We carry out the steps that were sketched in Lemma \ref{compo}.  
As  $ \Psi_{n}$ is a two-particle state, $g_{n}(\vec{p},z)$ is the sum of three functions, $u_{n}(\vec{p},z)$, $v_{n}(\vec{p},z)$ and $w_{n}(\vec{p},z)$, defined by
 \begin{align*}
 u_{n}(\vec{p},z) &:= \langle \Psi_{n} \vert  \mathds{1}_{H_f \leq 1}    w_{0,0}( \vec{p},z, H_f,\vec{P}_f)   \mathds{1}_{H_f \leq 1}  \Psi_{n} \rangle,\\[4pt]
 v_{n}(\vec{p},z) &:= \langle \Psi_{n} \vert  \mathds{1}_{H_f \leq 1}    W_{1,1}(\underline{w}(\vec{p},z))   \mathds{1}_{H_f \leq 1}  \Psi_{n} \rangle,\\[4pt]
w_{n}(\vec{p},z) &:=  \langle \Psi_{n} \vert  \mathds{1}_{H_f \leq 1}    W_{2,2}(\underline{w}(\vec{p},z))   \mathds{1}_{H_f \leq 1}  \Psi_{n} \rangle.
\end{align*}
We first look at the function $u_n$. We introduce the set $$\mathcal{D}(n;p):=\lbrace (\underline{k}_1,...,\underline{k}_n,\underline{k}_{n+1},...,\underline{k}_{n+p}) \in \underline{B}_{1}^{n+p} \mid  \vert \vec{k_1}  \vert+...+  \vert \vec{k}_n  \vert \leq 1,  \vert \vec{k}_{n+1}  \vert+ ...+ \vert \vec{k}_{n+p}  \vert \leq 1 \rbrace$$
and use the notation $\mathcal{D}(n):=\mathcal{D}(n;0)$.
An easy application of the pull-through formula and the canonical commutation relations shows that
 \begin{equation}
 \label{77}
 \begin{split}
 u_{n}(\vec{p},z)&=\int_{ \mathcal{D}(2)} d \underline{k}_1  d\underline{k}_2    \text{ }  w_{0,0}(\vec{p},z, \vert \vec{k}_1 \vert +  \vert \vec{k}_2 \vert,  \vec{k}_1+ \vec{k}_2)  \text{ } \overline{\eta}_{n,\vec{l}}(\underline{k}_1) \overline{\eta}_{n,\vec{l}'}(\underline{k}_2)  \eta_{n,\vec{l}}(\underline{k}_2) \overline{\eta}_{n,\vec{l}'}(\underline{k}_1) \\
 &+  \int_{\mathcal{D}(2)}  d \underline{k}_1  d\underline{k}_2    \text{ }  w_{0,0}(\vec{p},z, \vert \vec{k}_1 \vert +  \vert \vec{k}_2 \vert,  \vec{k}_1+ \vec{k}_2)  \text{ } \vert \eta_{n,\vec{l}}(\underline{k}_1) \vert^{2}  \vert \eta_{n,\vec{l}'}(\underline{k}_2) \vert^{2}.
 \end{split} 
 \end{equation}
We assume that $\vec{l} \neq \vec{l}'$.  We show that the first term on the right-hand side of  \eqref{77} tends to zero when $n$ tends to infinity, uniformly in $(\vec{p},z)$.  We make the change of variables  $\vec{k'}_1:=n(\vec{k}_1- \vec{l})$, $\vec{k'}_2:= n (\vec{k}_2- \vec{l})$.  Let $R>0$ such that $\text{Supp}(\eta) \subset B_{R}$. As $H(\underline{w}) \in \mathcal{B}(\gamma,\delta,\varepsilon)$, there exists a constant $C(\gamma)>0$ so that the first term on the right-hand side of  \eqref{77}  is   bounded by 
\begin{equation}
\label{integrand}
 C(\gamma) \int_{ \underline{B}_{R}^{2}}  d \underline{k}'_1  d\underline{k}'_2 \text{ }  \vert \overline{\eta}(\vec{k'}_1)  \overline{\eta}(\vec{k'}_2+ n(\vec{l}-\vec{l'})) \eta(\vec{k}_2) \eta(\vec{k'}_1+ n(\vec{l}-\vec{l'}))\vert,
\end{equation}
for any $(\vec{p},z)\in U[\vec{p}^*] \times D_{\mu/2}$. The function $\eta$ has a compact support. Therefore, the integrand in equation \eqref{integrand} goes point-wise to zero for any $(\underline{k}'_1,\underline{k}'_2)$. Furthermore, the integrand in \eqref{integrand} is bounded because $\eta$ is smooth. Lebesgue dominated convergence theorem  implies that \eqref{integrand} goes to zero as $n$ tends  to infinity. In the case where $\vec{l}=\vec{l}'$, the two terms on the right-hand side of   \eqref{77} are similar, and it is therefore sufficient to analyze the second term. This term converges  to $ 4 w_{0,0}(\vec{p},z,\vert \vec{l} \vert + \vert \vec{l'} \vert ,\vec{l} +\vec{l'})$,uniformly in $(\vec{p},z) \in U[\vec{p}^*] \times D_{\mu/2}$, as  a direct consequence of Lebesgue convergence theorem. Indeed, the supremum of $\big \vert  \vec{\nabla}_{r,\vec{l}}  \text{ } w_{0,0}(\vec{p},z,r,\vec{l})  \big \vert$  over  $\mathcal{B}$ is uniformly bounded in $(\vec{p},z)$ and  $\mathcal{B}$ is convex. Therefore, there exists a constant $C>0$ such that $$\vert w_{0,0}(\vec{p},z,r,\vec{l})-w_{0,0}(\vec{p},z,r',\vec{l}') \vert \leq C \vert (r,\vec{l})-(r',\vec{l}') \vert$$ for all $(\vec{p},z) \in U[\vec{p}^*] \times D_{\mu/2}$ and all $ (r,\vec{l}),(r',\vec{l}') \in \mathcal{B}$. The uniform convergence in $(\vec{p},z)$ follows.  We deduce  that $ u_{n}(\vec{p},z)$ converges uniformly to $4   w_{0,0}(\vec{p},z,\vert \vec{l} \vert + \vert \vec{l'} \vert ,\vec{l} +\vec{l'}) +  4   w_{0,0}(\vec{p},z, \vert 2 \vec{l} \vert,  2 \vec{l}) \delta_{\vec{l}, \vec{l'}}$ on $U[\vec{p}^*] \times D_{\mu/2}$. 
\vspace{2mm}

We now show that the sequence of functions $(v_n)_n$ converges to zero uniformly in $(\vec{p},z) \in U[\vec{p}^*] \times D_{\mu/2}$. We remind the reader that
 \begin{align*}
  v_{n}(\vec{p},z) &:= \langle \Psi_{n} \vert  \mathds{1}_{H_f \leq 1}    W_{1,1}(\underline{w}(\vec{p},z))   \mathds{1}_{H_f \leq 1}  \Psi_{n} \rangle,\\
\end{align*}
 where $\Psi_{n}:= b^{*}( \eta_{n, \vec{l}}) b^{*}( \eta_{n, \vec{l}'}) \vert \Omega \rangle$. One has that 
 \begin{equation}
 \label{II2}
 \begin{split}
 v_{n}(\vec{p},z) &=\underset{ \mathcal{D}(2; 2)\times \underline{B}_{1}^{2}}{\int} d\underline{K}  \text{ }\overline{\eta}_{n, \vec{l}}(\underline{k}_1) \overline{\eta}_{n, \vec{l}'}(\underline{k}_2)  \eta_{n, \vec{l}}(\underline{k}_3)    \eta_{n, \vec{l}'}(\underline{k}_4) \text{ } \langle \Omega \vert b(\underline{k}_1) b(\underline{k}_2) b^{*}(\underline{k}_5)\\
 & \qquad \qquad \qquad \qquad  w_{1,1}(\vec{p},z,H_f,\vec{P}_f,\vec{k}_5,\vec{k}_6)   b(\underline{k}_6) b^{*}(\underline{k}_3) b^{*}(\underline{k}_4)   \Omega \rangle,
 \end{split}
 \end{equation}
 
 \noindent  where we have set  $d\underline{K}:=\prod_{i=1}^{6} d \underline{k}_i$. The pull-through formula, together with the canonical commutation relations, imply that $ v_{n}(\vec{p},z)= \sum_{i=1}^{4} v_{n;i}(\vec{p},z)$, where 
   \begin{align*}
 v_{n;1}(\vec{p},z)&= \int_{\mathcal{D}(2;1)}  d\underline{k}_1  d\underline{k}_2  d\underline{k}_3  \text{ }\overline{\eta}_{n, \vec{l}} (\underline{k}_1) \overline{\eta}_{n, \vec{l'}} (\underline{k}_2)  \eta_{n, \vec{l}}(\underline{k}_3)     \eta_{n, \vec{l'}}(\underline{k}_2)      \text{ }  w_{1,1}(\vec{p},z,\vert \vec{k}_2\vert, \vec{k}_2,\vec{k}_1,\vec{k}_3) \mathds{1}_{\vert \vec{k}_3 \vert + \vert \vec{k}_2 \vert \leq 1},  \\[4pt]
 v_{n;2}(\vec{p},z)&= \int_{\mathcal{D}(2;1)}  d\underline{k}_1  d\underline{k}_2  d\underline{k}_3  \text{ }\overline{\eta}_{n, \vec{l}} (\underline{k}_1)  \overline{\eta}_{n, \vec{l'}}(\underline{k}_2)  \eta_{n, \vec{l}}(\underline{k}_2)     \eta_{n, \vec{l'}}(\underline{k}_3)     \text{ }  w_{1,1}(\vec{p},z,\vert \vec{k}_2\vert, \vec{k}_2,\vec{k}_1,\vec{k}_3) \mathds{1}_{\vert \vec{k}_3 \vert + \vert \vec{k}_2 \vert \leq 1} ,  \\[4pt]
  v_{n;3}(\vec{p},z)&= \int_{\mathcal{D}(2;1)} d\underline{k}_1  d\underline{k}_2  d\underline{k}_3 \text{ } \overline{\eta}_{n, \vec{l}}(\underline{k}_1)  \overline{\eta}_{n, \vec{l'}}(\underline{k}_2)  \eta_{n, \vec{l}}(\underline{k}_3)     \eta_{n, \vec{l'}}(\underline{k}_1)   \text{ }   w_{1,1}(\vec{p},z,\vert \vec{k}_1\vert, \vec{k}_1,\vec{k}_2,\vec{k}_3)  \mathds{1}_{\vert \vec{k}_3 \vert + \vert \vec{k}_1 \vert \leq 1},  \\[4pt]
 v_{n;4}(\vec{p},z)&=\int_{\mathcal{D}(2;1)} d\underline{k}_1  d\underline{k}_2  d\underline{k}_3 \text{ } \overline{\eta}_{n, \vec{l}}(\underline{k}_1)  \overline{\eta}_{n, \vec{l'}}(\underline{k}_2)   \eta_{n, \vec{l}}(\underline{k}_1)     \eta_{n, \vec{l'}}(\underline{k}_3)  \text{ }  w_{1,1}(\vec{p},z,\vert \vec{k}_1\vert, \vec{k}_1,\vec{k}_2,\vec{k}_3)   \mathds{1}_{\vert \vec{k}_3 \vert + \vert \vec{k}_1 \vert \leq 1}.
 \end{align*}
\vspace{3mm}

We  use that $\eta$ has a compact support included in $ B_{R}(0)$, $R>0$. We only detail the upper bounds for $v_{n;1}$ and $v_{n;2}$. The reader can check that the bounds for $v_{n;3}$ and $v_{n;4}$ are similar. After  a well-suited change of variables, we find that, for $n$  large enough,
  \begin{align*}
 v_{n;1}(\vec{p},z)&=  n^{-3} \int_{\underline{B}_{R}^{3} }  d\underline{k}'_1  d\underline{k}'_2  d\underline{k}'_3  \text{ }\overline{\eta} (\underline{k}'_1)   \eta (\underline{k'}_3)     \vert \eta (\underline{k}'_2) \vert^{2}      \text{ } f_n(\vec{p},z,\underline{k}'_2;\underline{k}'_1,\underline{k}'_3),  \\
 v_{n;2}(\vec{p},z)&=    n^{-3}  \int_{\underline{B}_{R}^{3} }  d\underline{k}'_1  d\underline{k}'_2  d\underline{k}'_3  \text{ }\overline{\eta}  (\underline{k}'_1)  \overline{\eta}(\underline{k}'_2)  \eta(\underline{k}'_2 + n(\vec{l'}- \vec{l}))     \eta (\underline{k}'_3) \text{ } g_n(\vec{p},z,\underline{k}'_2;\underline{k}'_1,\underline{k}'_3),
 \end{align*}
 where we have  set 
\begin{equation*}
\begin{split}
f_n(\vec{p},z,\underline{k};\underline{k}',\underline{k}'')&=  w_{1,1}(\vec{p},z,\vert   n^{-1} \vec{k} + \vec{l'}\vert,  n^{-1} \vec{k}+  \vec{l'},  n^{-1} \vec{k'}  + \vec{l},  n^{-1} \vec{k''} + \vec{l}),\\
g_n(\vec{p},z,\underline{k};\underline{k}',\underline{k}'')&=  w_{1,1}(\vec{p},z,\vert   n^{-1} \vec{k} + \vec{l'}\vert,  n^{-1} \vec{k}+  \vec{l'},  n^{-1} \vec{k'}  + \vec{l},  n^{-1} \vec{k''} + \vec{l}').
\end{split}
\end{equation*}
$H(\underline{w})$ belongs to $\mathcal{B}(\gamma,\delta,\varepsilon)$. It follows that $\|w_{1,1}(\vec{p},z) \|_{\frac{1}{2}}\leq \varepsilon$ for any $(\vec{p},z)\in U[\vec{p}^*] \times D_{\mu/2}$, and 
  \begin{align*}
  \vert  v_{n;1}(\vec{p},z) \vert&  \leq  C n^{-3} \int_{\underline{B}_{R}^{3} }  d\underline{k}'_1  d\underline{k}'_2  d\underline{k}'_3  \text{ }\vert \overline{\eta} (\underline{k}'_1)   \eta (\underline{k'}_3)  \vert \text{ }   \vert \eta (\underline{k}'_2) \vert^{2},\\
   \vert  v_{n;2}(\vec{p},z) \vert  &\leq C n^{-3}  \int_{\underline{B}_{R}^{3} }  d\underline{k}'_1  d\underline{k}'_2  d\underline{k}'_3  \text{ } \vert \overline{\eta}  (\underline{k}'_1)  \overline{\eta}(\underline{k}'_2)  \eta(\underline{k}'_2 + n(\vec{l'}- \vec{l}))     \eta (\underline{k}'_3) \vert.
  \end{align*}

 The upper bounds go to zero when $n$ goes to infinity, uniformly in $\vec{p}$ and $z$. A similar procedure would show that $(w_n)_n$ converges uniformly to zero.

 \vspace{2mm}
 \subsection{Knowing $f(\vert \vec{x} \vert + \vert \vec{y} \vert, \vec{x} +  \vec{y})$ is sufficient to know $f(r, \vec{l})$ for any $(r,\vec{l}) \in \mathcal{B}$.}
We remind the reader that $\mathcal{B}$ has been defined in \eqref{BB}.
\begin{lemma}
 \label{7.1}
Let  $f: \mathcal{B} \rightarrow \mathbb{C}$  be a function. We suppose that for any $\vec{x}, \vec{y} \in B_{1}(0)$ such that $\vert \vec{x} \vert  + \vert  \vec{y} \vert \leq 1$, the value of $f(\vert \vec{x} \vert + \vert \vec{y} \vert, \vec{x} +  \vec{y})$ is known. Then the value of $f$ is known for any point of $\mathcal{B}$. 
 \end{lemma}
 
 \begin{proof}
 Let $(r_0, \vec{x}_0) \in \mathcal{B}$, with $\vec{x}_0 \neq \vec{0}$ and $ \vert    \vec{x}_0 \vert <r_0$. We  can always choose a vector $\vec{x} \in \mathbb{R}^{3}$ orthogonal to $\vec{x}_0$, such that 
 \begin{equation*}
 \vert    \vec{x} \vert  =\frac{1}{2} \left( r_0- \frac{\vert \vec{x}_0 \vert^2}{r_0} \right).
 \end{equation*}
 Then $ 0<\vert    \vec{x} \vert <r_0/2 \leq 1/2$. We consider the vector
 \begin{equation*}
  \vec{y}   : = \vec{x}_0-\vec{x}.
 \end{equation*}
 As $\vec{x}$ and $\vec{x}_0$ are orthogonal, 
 \begin{equation*}
 \vert \vec{y} \vert^2 =  \vert    \vec{x}_0 \vert^2 +\vert    \vec{x} \vert^2 =\frac{1}{4} \left( r_0- \frac{\vert \vec{x}_0 \vert^2}{r_0} \right)^2 + \vert    \vec{x}_0 \vert^2= \frac{1}{4} \left( r_0+ \frac{\vert \vec{x}_0 \vert^2}{r_0} \right)^2.
 \end{equation*}
 Therefore, $\vert \vec{y} \vert<r_0 \leq 1$, and 
 \begin{equation*}
 \vert    \vec{x} \vert +\vert    \vec{y} \vert=r_0.
 \end{equation*}
 We  have found two vectors $(\vec{x}, \vec{y})$ in the unit ball of $\mathbb{R}^3$ such that $\vec{x}+ \vec{y}=\vec{x}_0$ and $\vert \vec{x} \vert + \vert \vec{y}\vert=r_0$. Consequently, $f(r_0, \vec{x}_0)=f(\vert    \vec{x} \vert +\vert    \vec{y} \vert,\vec{x}+ \vec{y})$. The case $\vec{x}_0=\vec{0}$ is trivial, as we can always choose $\vec{x} \in B_1(0)$ of norm $\vert \vec{x} \vert =r_0/2$. Then, $\vert \vec{x} \vert +\vert -\vec{x} \vert=r_0.$ 
 \end{proof}

\vspace{3mm}

\section{Proof of Lemma \ref{contr}}
\label{AppD}
We prove Lemma \ref{contr} with the help of the explicit formulas given in Appendix \ref{Wicko}. This section is cut into three subsections.  We start with a lemma that bounds  the function $F$ defined in \eqref{Fi} and its derivatives. We then show that the norm $\|\cdot\|^{\sharp}$ of the function $V_{\underline{m,p,n,q}}$ -- see \eqref{Vmpnq} -- is finite and  uniformly bounded in $(\vec{p},z)$. Finally, we prove Lemma \ref{contr}.
\subsection{A lemma to bound  the kernels}
 \begin{lemma}
  \label{bou1}
There exists a constant $C_{\chi}>0$ such that, for any $0 < \rho < 1/2$, $0 < \gamma \ll \mu$, $\delta , \varepsilon > 0$, $H(\underline{w})( \cdot , \cdot ) \in \mathcal{B}(\gamma, \delta,\varepsilon)$, and $(\vec{p},z,r,\vec{l}) \in \mathcal{U}[ w_{0,0} ] \times \mathcal{B}$, the function $F$ defined in \eqref{Fi} satisfies the upper bounds:
  \begin{align}
&  \vert F (\vec{p},z,r,\vec{l}) \vert \leq \frac{C_{\chi}}{\mu \rho},  \quad   \vert  \partial_r F (\vec{p},z,r,\vec{l}) \vert  \leq \frac{C_{\chi}}{(\mu \rho)^2} ,  \quad  \vert  \partial_{l_j} F (\vec{p},z,r,\vec{l}) \vert  \leq \frac{C_{\chi}}{(\mu \rho)^2}. \label{bFi}
  \end{align}
  \end{lemma}
    \begin{proof}
 One has that
    \begin{equation}
    \label{eqF}
    \begin{split}
    \partial_{r} F(\vec{p},z,r,\vec{l}) &= 2 \rho^{-1}   \frac{  \overline{\chi}' (\rho^{-1}Êr ) \overline{\chi}_{\rho}(r)}{ w_{0,0}(\vec{p},z,r,\vec{l})} -  \frac{ (\partial_r w_{0,0})(\vec{p},z,r,\vec{l}) \text{ }  \overline{\chi}_{\rho}^{2}(r)  }{ w_{0,0}^{2}(\vec{p},z,r,\vec{l})},
    \end{split}
    \end{equation} 
  and
    \begin{equation}
    \partial_{l_j}  F(\vec{p},z,r,\vec{l})=-   \frac{ (\partial_{l_j} w_{0,0})(\vec{p},z,r,\vec{l}) \text{ }  \overline{\chi}_{\rho}^{2}(r )  }{ w_{0,0}^{2}(\vec{p},z,r,\vec{l})}.
    \end{equation}
For all $r \in [ \frac34 \rho , 1 ]$, the bound (\ref{eq:a10}) holds for the function $w_{0,0}(\vec{p},z,r,\vec{l})$, and hence $| F(\vec{p},z,r,\vec{l}) | \le C_\chi ( \rho \mu )^{-1}$. The other bounds also follow directly from  \eqref{eq:a10}. \end{proof}

\subsection{A bound for the norm of  $V$}
\vspace{2mm}

 \begin{lemma}
 \label{ctkim}
 Let $L \in \mathbb{N}$, and $\underline{m,p,n,q} \in \mathbb{N}^{4L}$ with $\sum_{i=1}^{L}m_i=M$,  $\sum_{i=1}^{L}n_i=N$.   $V_{\underline{m,p,n,q}}$ defined in \eqref{VVm} belongs to $\mathcal{W}_{M,N}$ and we have that
 \begin{equation}
 \label{ccca}
 \rho^{\frac{3}{2}(M+N)-1} \big \|  V_{\underline{m,p,n,q}}(\vec{p},z) \big \|^{\sharp} \leq 5   (L+1) \mu^{-L} C_{\chi}^{L+1} \rho^{ 2(M+N) -L} \prod_{i=1}^{L}     (\sqrt{8 \pi})^{p_i+q_i} \big \| w_{m_i+p_i,n_i+q_i}(\vec{p},z)  \big \| ^{\sharp} ,
 \end{equation}
 \normalsize
 for all $( \vec{p} , z ) \in U [ \vec{p}^* ] \times D_{ \mu / 2 }$.
 \end{lemma}
\vspace{2mm}

 \begin{proof}
To shorten notations, we set
\begin{align*}
 \tilde{W}_i &:= W^{m_i,n_i}_{p_i,q_i}(\vec{p},z,\rho(r+r_i), \rho(\vec{l} + \vec{l}_i), \rho \underline{K}_{i}^{(m_i,n_i)} ), \\[4pt]
 \tilde{W}'_{i,r} &:= ( \partial_r W^{m_i,n_i}_{p_i,q_i} ) (\vec{p},z,\rho(r+r_i), \rho(\vec{l} + \vec{l}_i), \rho \underline{K}_{i}^{(m_i,n_i)} ), \\[4pt]
  \tilde{W}'_{i,l_j} & := ( \partial_{l_j} W^{m_i,n_i}_{p_i,q_i} ) (\vec{p},z,\rho(r+r_i), \rho(\vec{l} + \vec{l}_i), \rho \underline{K}_{i}^{(m_i,n_i)} ).
\end{align*}
We remind the reader that the operators $ W^{m_i,n_i}_{p_i,q_i}$ have been defined in \eqref{Wii}.  $\partial_rW^{m_i,n_i}_{p_i,q_i}$ and $\partial_{l_j}W^{m_i,n_i}_{p_i,q_i}$  have the same form as $ W^{m_i,n_i}_{p_i,q_i}$, with the  kernel  $w_{m_i+p_i,n_i+q_i}$ replaced by  $\partial_r w_{m_i+p_i,n_i+q_i}$ and $\partial_{l_j} w_{m_i+p_i,n_i+q_i}$, respectively. The  formulas (\ref{VVm}) and (\ref{bFi}) imply that 
 \begin{align*}
& \big \vert   V_{\underline{m,p,n,q}} (\vec{p},z,r, \vec{l}, \underline{K}^{(M,N)}) \big  \vert   \leq  C_{\chi}^{L+1}  \left(\rho \mu \right)^{ -L+1}     \prod_{i=1}^{L} \| \tilde{W}_i \| ,\\
& \big \vert    \partial_r  V_{\underline{m,p,n,q}} (\vec{p},z,r, \vec{l}, \underline{K}^{(M,N)})\big \vert     \leq (L+1) \mu^{-1} C_{\chi}^{ L+1}   \left(\rho \mu \right)^{ -L+1}   \prod_{i=1}^{L}   \Big( \| \tilde{W}_i \| +  \rho \big \| \tilde{W}'_{i,r}  \big \|  \Big)   , \\
& \big  \vert  \partial_{l_j}  V_{\underline{m,p,n,q}} (\vec{p},z,r, \vec{l}, \underline{K}^{(M,N)}) \big \vert    \leq (L+1)   \mu^{-1}  C_{\chi}^{ L+1}  \left(\rho \mu \right)^{-L+1}  \prod_{i=1}^{L}   \Big( \| \tilde{W}_i \| + \rho \big \| \tilde{W}'_{i,l_j} \big \| \Big).
  \end{align*}

Together with Lemma \ref{lm:Hbounded_injective}, the definition of the norm $\| \cdot \|^{\sharp}$ (see \eqref{sharp}--\eqref{def_norme_f_m,n}) implies that   
 \begin{align*}
& \big \|   V_{\underline{m,p,n,q}} (\vec{p},z)  \big \|_{\frac{1}{2}}   \leq  C_{\chi}^{L+1}  \left(\rho \mu \right)^{ -L+1}     \prod_{i=1}^{L}   (\sqrt{8 \pi})^{ p_i+q_i} \big  \| w_{m_i+p_i,n_i+q_i}(\vec{p},z)   \big \|^{\sharp} \rho^{\frac{M+N}{2}},\\[4pt]
& \big \| \partial_r  V_{\underline{m,p,n,q}} (\vec{p},z)\big \|_{\frac{1}{2}}   \leq (L+1) \mu^{-1} C_{\chi}^{ L+1}   \left(\rho \mu \right)^{ -L+1}        \prod_{i=1}^{L}   (\sqrt{8 \pi})^{ p_i+q_i} \big    \| w_{m_i+p_i,n_i+q_i}(\vec{p},z)   \big \|^{\sharp}  \rho^{\frac{M+N}{2}}, \\[4pt]
&  \big \|  \partial_{l_j}  V_{\underline{m,p,n,q}} (\vec{p},z) \big  \|_{\frac{1}{2}}     \leq (L+1)   \mu^{-1}  C_{\chi}^{ L+1}  \left(\rho \mu \right)^{-L+1}       \prod_{i=1}^{L}   (\sqrt{8 \pi})^{ p_i+q_i}  \big  \| w_{m_i+p_i,n_i+q_i}(\vec{p},z)   \big \|^{\sharp} \rho^{\frac{M+N}{2}}.
  \end{align*}
  \normalsize
 \end{proof}

\subsection{Proof of Lemma \ref{contr}}
 \begin{proof}  The rest of the proof of Theorem \ref{contr} is  similar to  \cite[Theorem 3.8]{BaChFrSi03_01}. Let $H ( \underline{w} ( \cdot , \cdot ) ) \in \mathcal{B}( \gamma , \delta , \varepsilon )$ and let $\underline{\hat{w}}$ be defined by \eqref{whatmn}. We have to show that $H ( \underline{\hat{w}} ( \cdot , \cdot ) ) \in \mathcal{B}( \gamma + \varepsilon / 2 , \varepsilon / 2 , \varepsilon / 2 )$.  It follows from \eqref{whatmn} that 
\begin{equation*}
\big \| \hat{w}_{M,N}( \vec{p},\zeta) \big \|^{\sharp}\leq     \sum_{L=1}^{\infty}    \rho^{\frac{3}{2}(M+N)-1} \underset{ \tiny \begin{array}{c} \underline{m, p,n,q} ,  \\ m_1+...+m_L=M, \\ n_1+...+n_L=N , \\  m_i+n_i+p_i+q_i \geq 1 \end{array}}{\sum}   C_{\underline{m},\underline{n}}^{\underline{m+p},\underline{n+q}}     \big \|V_{\underline{m,p,n,q}} (\vec{p},z) \big \|^{\sharp} ,
   \end{equation*}
where, recall, $( \vec{p}  , z ) = E_\rho^{-1} ( \vec{p} , \zeta )$. Since $H(\underline{w})  \in \mathcal{B}( \gamma , \delta , \varepsilon)$, we have that
\begin{equation*}
\big \| w_{m_i+p_i,n_i+q_i}(\vec{  p},z)  \big \|^{\sharp}  \leq  \varepsilon \xi^{m_i+p_i+n_i+q_i} 
\end{equation*}
for any  $m_i+p_i+n_i+q_i \ge 1$. Plugging this bound into \eqref{ccca}, we get that 
\begin{equation*}
\xi^{-(M+N)} \big \| \hat{w}_{M,N}( \vec{p},\zeta) \big \|^{\sharp} \leq   5 C_{\chi}  (2\rho)^{2(M+N)}   \sum_{L=1}^{\infty}    (L+1) \left( \frac{   \varepsilon  C_{\chi}}{\mu \rho} \right)^{L} \underset{ \tiny \begin{array}{c} \underline{m, p,n,q} , \\ m_1+...+m_L=M, \\ n_1+...+n_L=N , \\  m_i+n_i+p_i+q_i \geq 1 \end{array}}{\sum}    \prod_{i=1}^{L} \\ f(m_i,p_i,n_i,q_i),
\end{equation*}
where
\begin{equation*}
 f(m_i,p_i,n_i,q_i)=    ( 2 \sqrt{8 \pi}  \xi)^{   p_i+q_i}  \left(\frac{1}{2}\right)^{m_i+n_i}. 
\end{equation*}
\vspace{1mm}

\noindent  It holds that 
\begin{equation*}
\begin{split}
  \sum_{M+N \leq K}  &  \underset{ \tiny \begin{array}{c} \underline{m, p,n,q} , \\ m_1+...+m_L=M, \\ n_1+...+n_L=N , \\  m_i+n_i+p_i+q_i \geq 1 \end{array}}{\sum}    \prod_{i=1}^{L}  f(m_i,p_i,n_i,q_i)\leq   \sum_{M,N \leq K} \underset{ \tiny \begin{array}{c} \underline{m, p,n,q} , \\ m_1+...+m_L=M, \\ n_1+...+n_L=N , \\  m_i+n_i+p_i+q_i \geq 1 \end{array}}{\sum}    \prod_{i=1}^{L}  f(m_i,p_i,n_i,q_i)  \\
&     \leq     \prod_{i=1}^{L}  \Big( \underset{ \tiny \begin{array}{c}  p_i,q_i   \\ m_i  \leq K, \\ n_i \leq K , \\  m_i+n_i+p_i+q_i \geq 1 \end{array}}{\sum}   f(m_i,p_i,n_i,q_i)   \Big) 
\end{split}
\end{equation*}
for any $K \in \mathbb{N}$.  Since $2 \sqrt{8 \pi}  \xi<1/2$ by assumption, $ f(m_i,p_i,n_i,q_i) <  \left(\frac{1}{2}\right)^{m_i+n_i+p_i+q_i}$, and  taking the limit $K \rightarrow \infty$, we finally get that 
\begin{equation}
\label{cleq2}
\big \| \hat{w}( \vec{p},\zeta) \big \|^{\sharp}_{\xi,\geq 1} \leq   20 C_{\chi}   \rho^2       \sum_{L=1}^{\infty}    (L+1) \left( \frac{    16 \text{ }   \varepsilon C_{\chi}}{\mu \rho} \right)^{L}.
\end{equation}
Choosing $\varepsilon \ll \rho \mu$, we obtain
\begin{equation}
\label{cfini}
\big \| \hat{w}( \vec{p},\zeta) \big \|^{\sharp}_{\xi , \ge 1} \leq  C \frac{  \rho  \varepsilon  }{  \mu } . 
\end{equation}
\eqref{cfini}  is  smaller than $\varepsilon/2$  if $\rho \ll \mu$.
\vspace{2mm}

 Next, writing 
$\hat{w}(\vec{p}, \zeta , r , \vec{l} )=: \hat{t} (\vec{p}, \zeta , r , \vec{l} ) +\hat{w}(\vec{p}, \zeta , 0 , \vec{0})$, we have to show that for any $(\vec{p},\zeta) \in U[\vec{p}^*] \times D_{\mu/2}$,
\begin{align*}
  \big \| \hat{t}(\vec{p}, \zeta , r , \vec{l} ) - ( r - m^{-1}  \vec{p} \cdot \vec{l} ) \big \|^{\sharp} \le \gamma + \varepsilon / 2 ,
 \end{align*}
and that $|\hat{w}(\vec{p}, \zeta , 0 , \vec{0}) + \zeta | \le \varepsilon / 2$. We have that
\begin{equation}
 \label{hatw0}
\hat{w}(\vec{p}, \zeta , r , \vec{l} ) = \rho^{-1} w_{0,0}(\vec{p},z, \rho r , \rho \vec{l}) + \rho^{-1} \sum_{L=2}^{\infty} (-1)^{L-1} \underset{\underset{p_i+q_i \geq 1}{ \underline{p} , \underline{q} }}{\sum} V_{\underline{0,p,0,q}} (\vec{p},z,r,\vec{l}),
 \end{equation}
with $(\vec{p},\zeta) = E_\rho(\vec{p},z)$. Similar estimates as in Lemma \ref{ctkim} lead us to 
 \begin{equation}
 \rho^{-1} \big \|   V_{\underline{0,p,0,q}}(\vec{p},z) \big \|^{\sharp}  \leq 5 (L+1) C_{\chi}^{L+1} (\rho \mu)^{-L} \prod_{i=1}^{L} (\sqrt{ 8 \pi})^{p_i+q_i} \|  w_{ p_i, q_i}(\vec{p},z)   \|^{\sharp}.
 \end{equation}
 We set
 \begin{align}
 A( \zeta ) :=& \big \| \hat{t} (\vec{p}, \zeta , r , \vec{l} )  - ( r - m^{-1} \vec{p} \cdot \vec{l} ) \big \|^{\sharp} \notag \\
 =&\underset{  (r,\vec{l})\in \mathcal{B} }{\sup} \vert \partial_r \hat{w}_{0,0}(\vec{p}, \zeta , r , \vec{l} )  - 1 \vert + \sum_{j=1}^{3} \underset{ (r,\vec{l})\in \mathcal{B} }{\sup} \Big \vert \partial_{l_j} \hat{w}_{0,0}(\vec{p}, \zeta , r , \vec{l} )  + \frac{p_j }{m} \Big \vert.
 \end{align}
From (\ref{hatw0}), we deduce that 
 \begin{align*}
 A( \zeta ) & \leq  \underset{  (r,\vec{l})\in \mathcal{B} }{\sup} \vert \partial_r w_{0,0}(\vec{p},z , r , \vec{l} )  - 1 \vert + \sum_{j=1}^{3} \underset{  (r,\vec{l})\in \mathcal{B} }{\sup} \Big \vert \partial_{l_j} w_{0,0}(\vec{p}, z, r , \vec{l} )  + \frac{p_j}{m} \Big \vert \\[4pt]
 &\quad + \rho^{-1} \sum_{L=2}^{\infty}   \underset{\underset{p_i+q_i \geq 1}{ \underline{p} , \underline{q} }}{\sum} \left(  \underset{  (r,\vec{l})\in \mathcal{B} }{\sup} \vert \partial_r V_{\underline{0,p,0,q}} (\vec{p},z,r,\vec{l})   \vert +  \sum_{j=1}^{3}  \underset{  (r,\vec{l})\in \mathcal{B} }{\sup} \vert \partial_{l_j} V_{\underline{0,p,0,q}} (\vec{p},z , r , \vec{l} )    \vert   \right)\\[4pt]
 & \leq  \gamma +    \rho^{-1} \sum_{L=2}^{\infty}    \text{ } \underset{\underset{p_i+q_i \geq 1}{ \underline{p} , \underline{q} }}{\sum} \big \| V_{\underline{0,p,0,q}} (\vec{p},z) \big \|^{\sharp} \\[4pt]
 & \leq  \gamma +   5 \sum_{L=2}^{\infty}   (L+1) C_{\chi}^{L+1} (\rho \mu)^{-L} \left( \sum_{p+q \geq 1}   (\sqrt{8 \pi})^{p+q} \|  w_{ p , q}(\vec{p},z)  \|^{\sharp} \right)^{L}\\[4pt]
  & \leq  \gamma +   5 \sum_{L=2}^{\infty}   (L+1) C_{\chi}^{L+1} (\rho \mu)^{-L} \left(  \sqrt{8 \pi}  \xi \|  \underline{w}(\vec{p},z)    \|_{ \xi , \ge 1}^{\sharp} \right)^{L} \\[4pt]
  & \leq  \gamma +   5  C_{\chi} \sum_{L=2}^{\infty}   (L+1) \left( \frac{C_{\chi}    \xi       \varepsilon}{\rho \mu} \right)^{L} \leq \gamma  + 60   C_{\chi}   \left( \frac{C_{\chi}  \xi       \varepsilon}{\rho \mu} \right)^{2} \leq \gamma + \frac{\varepsilon}{2},
 \end{align*}
 where we used that $\sum_{L=2}^{\infty} (L+1) a^{L} \leq 12 a^{2}$ for $a \leq 1/2$,  $\varepsilon \ll \rho^2 \mu^2$, and  absorbed $\sqrt{8 \pi}$ into $C_{\chi}$.  To finish the proof of Theorem \ref{contr},  we need to check that $| w_{0,0}(\vec{p},\zeta,0, \vec{0})  + \zeta | \le \varepsilon / 2$. By definition,
 \begin{equation}
 \zeta=  -w_{0,0}(\vec{p},z,0, \vec{0}) \rho^{-1},
 \end{equation}
 which together with equation (\ref{hatw0}) and the above calculation, implies that
 \begin{equation}
\vert w_{0,0}(\vec{p},\zeta,0, \vec{0})  + \zeta \vert \leq \frac{\varepsilon}{2}.
 \end{equation}

 \end{proof}

\end{appendix}

\nocite{*}
\bibliographystyle{plain}
\bibliography{main}

\end{document}